\documentclass[aps,prl,longbibliography,reprint,superscriptaddress]{revtex4-2}

\usepackage{hhline}

\usepackage{array}
\usepackage{layout}
\usepackage{mathtools}

\usepackage{xcolor}
\usepackage{graphicx}
\usepackage{dcolumn}
\usepackage{tikz}
\usepackage{subcaption}
\usepackage{dirtytalk}
\usepackage{setspace}
\usepackage{enumitem}

\usepackage{ragged2e}
\usepackage{verbatim}

\usepackage{bm}
\usepackage{algorithmic}
\usepackage[utf8]{inputenc}
\usepackage{amsmath}
\usepackage{hyperref}
\usepackage{cleveref}
\usepackage{zx-calculus}
\usepackage[margin=1in]{geometry}

\usepackage{amssymb}
\usepackage{siunitx}
\usepackage{float}
\usepackage{amsthm}
\usepackage{tikz-cd}
\usepackage{bbm}
\hypersetup{
    colorlinks = true,
    linkcolor = violet,
    linkbordercolor = {white}
}
\usepackage{quantikz}
\usepackage{braket}

\theoremstyle{definition}
\newtheorem{theorem}{Theorem}[section]
\theoremstyle{definition}
\newtheorem{claim}[theorem]{Claim}
\theoremstyle{definition}
\newtheorem{definition}[theorem]{Definition}
\theoremstyle{definition}
\newtheorem{conjecture}[theorem]{Conjecture}
\theoremstyle{definition}
\newtheorem{corollary}[theorem]{Corollary}
\theoremstyle{definition}
\newtheorem{remark}[theorem]{Remark}

\theoremstyle{definition}

\theoremstyle{definition}
\newtheorem{lemma}[theorem]{Lemma}
\theoremstyle{definition}
\newtheorem{proposition}[theorem]{Proposition}

\theoremstyle{remark}

\begin{document}

\title{Equivalence Classes of Quantum Error-Correcting Codes}

\author{Andrey Boris Khesin}
\affiliation{Department of Mathematics, Massachusetts Institute of Technology, Cambridge, MA}%

\author{Alexander M. Li}
\affiliation{C. Leon King High School, Tampa, FL}

\date{June 17, 2024}

\begin{abstract}
Quantum error-correcting codes (QECC's) are needed to combat the inherent noise affecting quantum processes. Using ZX calculus, we present QECC's as ZX diagrams, graphical representations of tensor networks. In this paper, we present canonical forms for CSS codes and CSS states (which are CSS codes with 0 inputs), and we show the resulting canonical forms for the toric code and certain surface codes. Next, we introduce the notion of prime code diagrams, ZX diagrams of codes that have a single connected component with the property that no sequence of rewrite rules can split such a diagram into two connected components. We also show the Fundamental Theorem of Clifford Codes, proving the existence and uniqueness of the prime decomposition of Clifford codes. Next, we tabulate equivalence classes of ZX diagrams under a different definition of equivalence that allows output permutations and any local operations on the outputs. Possible representatives of these equivalence classes are analyzed. This work expands on previous works in exploring the canonical forms of QECC's in their ZX diagram representations.

\bigskip

\end{abstract}

\maketitle


\section{Introduction}

\selectfont

The work done in the past half century on quantum computing has brought large-scale quantum computers closer to reality. Today, quantum computers can employ a low number (up to a few hundred) of qubits, in the form of photons and nuclear spins \cite{xu2023two}, but they have been used mainly for experiments. Quantum computers differ from classical computers and classical supercomputers through the use of qubits rather than bits. The properties of quantum mechanics inherent in qubits, including superposition and entanglement, allow quantum computers to efficiently simulate quantum systems, making certain calculations much more efficient when done on quantum computers \cite{nielsen2002quantum}. 

However, as with classical information processing systems, quantum information processing systems face noise that disrupts information transmission between the sender and receiver. Because of the vulnerability of qubits to this noise, one of the principal challenges in quantum computing is to account for this noise \cite{kim2023evidence}. To this end, quantum error-correcting codes are developed so that quantum information can be transmitted successfully in the presence of noise~\cite{nielsen2002quantum}. An important restriction on quantum error-correcting codes stems from the no-cloning theorem -- while classical computers can copy bits, quantum mechanics does not allow for the cloning of unknown qubits, and the measurement of a qubit eliminates the information available in the qubit \cite{nielsen2002quantum}. As such, constructing suitable quantum error-correcting codes presents new challenges compared to their classical counterparts. The first advances into quantum error-correcting codes against general errors came with the Shor code~\cite{shor1995scheme}, published in 1995 and the Steane code~\cite{steane1996multiple}, in 1996. Other examples of quantum error-correcting codes are the five qubit code~\cite{knill2001benchmarking} and the toric code~\cite{ec1}.


A number of approaches have been created to represent the components of quantum error-correcting codes. The \textit{stabilizer} formalism is a method that expresses quantum error-correcting codes in terms of stabilizers, operators that, when applied to certain stabilizer states, preserve the state~\cite{gottesman1997stabilizer}. This approach borrows ideas from group theory to represent the whole class of stabilizers with a finite number of generators. To make the idea of quantum error-correcting codes visual, recent advances have made progress on the topic of presenting \textit{graph states}~\cite{vandennest2004graphical, hu2022improved}.

\begin{figure*}
\centering
    \includegraphics[scale=0.46]{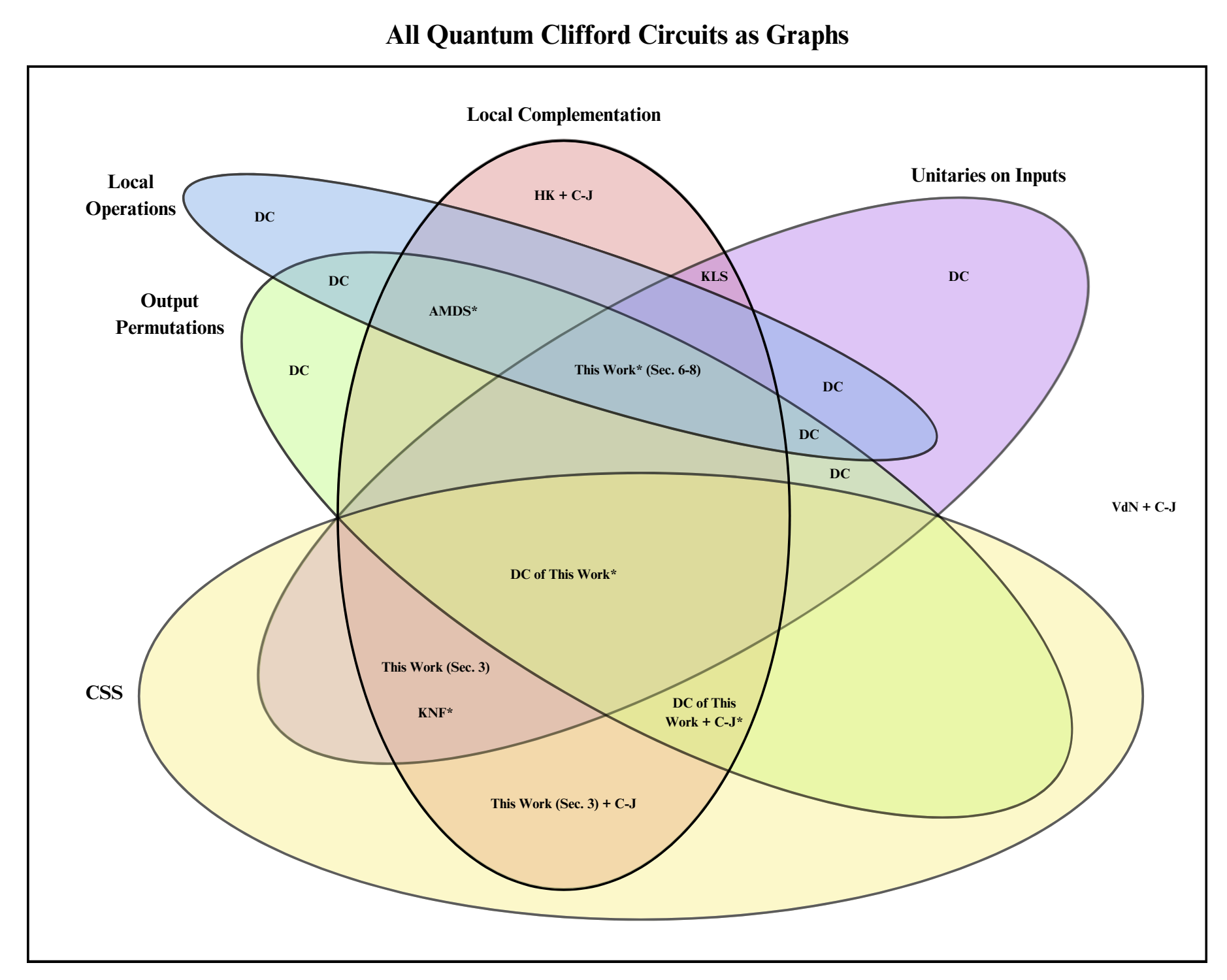}

    \caption{\justifying A summary of the work done on the equivalence classes of graphs of quantum Clifford encoders. Four of the categories -- output permutations, local operations, local complementation, and unitaries on inputs -- are different equivalences that preserve the information received and alter the graph in some way. The CSS category separates work done on encoders for CSS codes from work done on general Clifford encoders. Note that the CSS and local operations categories do not intersect because equivalence by local operations can transform a CSS code into a non-CSS code. The colored regions of the diagram are labeled with works that have been done on finding canonical forms under the corresponding equivalences. The asterisk (*) denotes results that give non-unique forms, and DC denotes results that are ``direct consequences" of previous works. The other abbreviations used are explained in the main text.
}

    \label{venn-diagram}
\end{figure*}

Following the work on graph states, work has been done on representing Clifford codes using ZX calculus~\cite{coecke2008interacting, coecke2011interacting,backens2014zx}. The properties of ZX calculus that allow it to replace the stabilizer tableau formalism (a tabulated form of the generators of the stabilizers) are its universality (it can express every quantum operation), soundness (tableaus can derive equivalence of ZX calculus diagrams), and completeness (ZX calculus diagrams can derive equivalence of tableaus)~\cite{coecke2008interacting,backens2014zx}. This graphical language has had various applications in quantum information~\cite{zx1, zx2, zx3, zx4} and quantum computation problems~\cite{de2020zx, kissinger2020reducing}.


\bigskip

\Cref{venn-diagram} gives a summary of the work done in graphically presenting quantum Clifford encoders. The set of all quantum Clifford encoders as graphs is split into different categories based on the type of code represented (CSS or general Clifford codes) and four different equivalence operations (output permutations, local operations, local complementation, and unitaries on inputs). Note that equivalence under local operations means equivalence under any local Clifford operations on the outputs. All four of these categories represent different equivalences, which preserve the information the receiver gets while changing the encoder's ZX diagram in some way. Some of them change the encoder (output permutations, local operations, and unitaries on inputs), and a different set of them change the code (output permutations and local operations).

Expressing quantum Clifford encoders as graphs and finding canonical forms for equivalent graphs has had recent advances in the past few decades (see Figure \ref{venn-diagram}). Van den Nest et al.'s (VdN) work \cite{vandennest2004graphical} provides a conversion between any Clifford state and a graph with local Clifford gates, and the Choi-Jamiolkowski (C-J) isomorphism \cite{choi1975completely,jamiolkowski1972linear} extends this to a conversion between any Clifford circuit and a graph with local Clifford gates. Starting from the graphs of the Clifford encoders, we can find the canonical forms of the direct consequences (DC's) located only within the bubbles for output permutations, local operations, and unitaries on inputs. For output permutations, we remove the numbering on the outputs. For local operations, we remove all local Clifford operations on the outputs. For unitaries on inputs, we remove input-input edges, local Clifford operations on the inputs, and the numbering on the inputs. Any combination of these three become DC in Figure \ref{venn-diagram}. The Hu-Khesin (HK) form from \cite{hu2022improved} provides a canonical form for quantum Clifford states. In the context of quantum encoders, this is equivalent to having no inputs and only outputs. Also, the Khesin-Lu-Shor (KLS) form from \cite{KLS} built on the HK form, providing a canonical form for Clifford encoders. The KLS paper shows the process of transforming stabilizer tableaus into the ZX calculus, then performing operations that preserve equivalence to transform the graph into its canonical form. Adcock, Morley-Short, Dahlberg, and Silverstone (AMDS) \cite{adcock2020mapping} considered equivalence of graph states under local complementations and the effects of relabeling the nodes. Kissinger's ZX Normal Form (KNF) \cite{kissinger2022phase} found a way to represent the CSS codes using internal measurement nodes. Section \ref{sec: kl forms} establishes the KL canonical form of CSS codes and states. In Figure \ref{venn-diagram}, DC's of Section \ref{sec: kl forms} (within the CSS bubble) follow through a removal of the numbering on output nodes. In Sections \ref{sec 4} to \ref{sec:bipartite forms}, we consider equivalence classes under all four equivalences for general Clifford encoders.

\bigskip

In this paper, \Cref{section:background} contains key definitions and background on ZX calculus and Clifford encoders. Section \ref{sec: kl forms} presents our main result, the KL canonical form for CSS codes, giving a unique, phase-free form for CSS codes that minimizes the number of nodes and clearly shows the $Z$ stabilizers and logical $Z$ operators. These results are also extended to CSS states, which will be defined later. Section \ref{section:surfacetoric} contains our work on toric codes and specific surface codes, and it shows different forms of these codes, including the canonical form for the toric code based on Section \ref{sec: kl forms}. These two sections build on the KLS forms for Clifford codes \cite{KLS} and recent work \cite{huang2023graphical, li2023graphical, kissinger2022phase} that introduced the normal form of CSS codes, which can efficiently determine the stabilizers from the ZX normal form. Our representations of CSS codes will also have this property. Furthermore, our representations reduce the number of nodes used in the ZX diagrams so that each node corresponds to either an input or output.

Section \ref{sec: prime codes} introduces the prime code diagrams, which focus on diagrams that are  composed of one connected component. Furthermore, we prove the Fundamental Theorem of Clifford Codes, showing the unique prime decomposition of Clifford codes.

Section \ref{sec 4} provides another definition of equivalence, permitting outputs to be permuted as a valid operation among equivalent graphs. The reason this definition of equivalence is also considered is that changing the order of the outputs does not change any code parameters. Simplifications on the set of Clifford encoders considered are given to narrow down the search for the canonical form. In Sections \ref{sec: tabulations from code} and \ref{sec:bipartite forms}, we show work done on identifying equivalence classes and finding representative forms. We analyze the equivalence class sizes and the presence of bipartite forms among these classes. Section \ref{sec:bipartite forms} expands on the equivalence classes containing bipartite forms and considers some classes that do not have bipartite forms.

\section{Background}
\label{section:background}

In this section, we define key terms and background on error-correcting codes and ZX calculus.

First, we define the following matrices.
\begin{definition} \label{def:def444}

    The \textit{Pauli matrices} are $$I \equiv \begin{pmatrix}1&0\\0&1\end{pmatrix}, \quad X \equiv \begin{pmatrix}
    0&1\\1&0
    \end{pmatrix},$$ $$Y \equiv \begin{pmatrix}
       0&-i\\i&0
    \end{pmatrix},\quad Z\equiv \begin{pmatrix} 1&0\\0&-1\end{pmatrix}.$$

    The Pauli matrices represent quantum gates that can act on qubits and alter their state. All four gates are Hermitian, and the three gates $X,Y,Z$ are pairwise anti-commutative. The \textit{Pauli operators on $n$ qubits} are $n$-fold tensor products of Pauli matrices, multiplied by a factor of the form $i^k$ where $k \in \{0,1,2,3\}$ and $i = \sqrt{-1}$.
\end{definition}

The Pauli operators are all equal to their conjugate transposes, so all Pauli operators are Hermitian and unitary. The Pauli operators form a group, called the \textit{Pauli group}. 

Pauli operators can act on states in multi-qubit systems. For example, in a three qubit system, the tensor product $Z\otimes Z \otimes I$ will make $Z$ act on the first qubit, $Z$ act on the second qubit, and $I$ act on the third qubit. The notation for the tensor product can be simplified to $Z_1Z_2$, with the subscripts showing which qubits the operators are acting on. We may also write these three gates as $ZZI$, omitting the tensor product symbols. Other Pauli operators on multiple qubits can be written analogously. 

Other quantum gates that are commonly used in quantum error-correcting codes are the Hadamard ($H$), controlled-NOT (CNOT), phase ($S$), and $\pi/8$ ($T$) gates:
{\setstretch{1}
\[H \equiv \frac{1}{\sqrt{2}} \begin{pmatrix}
1 & 1 \\ 1& -1 \end{pmatrix}, \quad \text{CNOT} \equiv \left(\begin{smallmatrix}
1&0&0&0\\ 0&1&0&0\\ 0&0&0&1\\0&0&1&0
\end{smallmatrix}\right),\]
\[S \equiv \begin{pmatrix}
1 & 0 \\ 0&i\end{pmatrix}, \quad T \equiv \begin{pmatrix}
1&0\\0&e^{i \pi  /4}\end{pmatrix}.\]
}

These operations have the property of universality, the ability to approximate any operator to arbitrary accuracy~\cite{nielsen2002quantum}.

We define \textit{encoders} of quantum error-correcting codes, or QECC's, as families of quantum processes that apply a transformation on some number of input qubits, mapping it to a given range.
Specifically, we will only be concerned with the case of \textit{full-rank} or \textit{non-degenerate} encoders where the encoding operation is injective, meaning that the dimension of the range is at least as large as the number of inputs.
We write that an encoder takes $k$ inputs, or logical qubits, and gives $n$ outputs, or physical qubits.
Note that an encoder can mean any quantum map in the family of such processes defined by the image.
Thus, two different circuits that differ only by a unitary operation on the inputs represent the same encoder.

Encoding quantum information into a larger number of qubits provides redundancy, making it possible to correct certain errors. The \textit{stabilizers} of an encoder are a set of commuting operations that can be composed with the output of an encoder without any change to the overall process. In the stabilizer formalism, the stabilizers determine the entire quantum code \cite{gottesman1996class}.
A family of encoders are \textit{Clifford codes}, quantum error-correcting codes such that each stabilizer is a Pauli operator on $n$ qubits. All stabilizers of a Clifford code on $k$ inputs form a group isomorphic to $\mathbb{Z}_2^{n-k}$ and can be defined by a set of linearly independent \textit{generators}. For $k$-to-$n$ codes, or $[n,k]$ codes, we have $n-k$ generators. The code maps the input qubits onto elements of the \textit{codespace}, the range of the code, which is the intersection of the +1 eigenspaces of the code's stabilizers.


ZX calculus is a graphical language used for expressing quantum processes. It makes the representation of Clifford codes graphical, and we use the conventions as described in \cite{backens2014zx, KLS}. The $Z$ (or green) nodes \zx{\zxZ{\alpha}} and $X$ (or red) nodes \zx{\zxX{\beta}} represent tensors that can be used to represent quantum operations such as qubits, gates, or measurements. Each node has a phase, with empty nodes representing a phase of 0. A $Z$ node with $n$ inputs, $m$ outputs, and phase $\alpha$ is equivalent to applying the operation $\ket{0}^{\otimes m}\bra{0}^{\otimes n} + e^{i\alpha}\ket{1}^{\otimes m}\bra{1}^{\otimes n}$, ignoring normalization. An $X$ node with $n$ inputs, $m$ outputs, and phase $\beta$ is equivalent to applying the operation $\ket{+}^{\otimes m}\bra{+}^{\otimes n} + e^{i\beta}\ket{-}^{\otimes m}\bra{-}^{\otimes n}$, ignoring normalization.

The ZX diagram represents a tensor network, and connections between nodes are inner products on the indices of these nodes. It is possible to convert efficiently between quantum circuits and ZX diagrams, with a specific example given in Appendix \ref{appendix sec: convert KLS to circuit}. Hadamard gates, represented by \zx{\zxH{}}, can be placed on edges between nodes or free edges. Another convention for Hadamard gates placed on edges is to color the edge in blue rather than the default black. In this work, edges with Hadamards will be represented with blue edges while edges without Hadamards will be represented with black edges.

ZX calculus has a set of basic rewrite rules that allow for diagrams to be converted into equivalent forms. Because ZX calculus with respect to Clifford codes and states is complete \cite{backens2014zx}, these basic rewrite rules can be used to derive all other rewrite rules.

\begin{definition}[Basic rewrite rules \cite{van2020zx}]
\label{def: basic rewrite rules}
The basic rewrite rules hold when the colors are interchanged, and all other rewrite rules are derived from these eight. Note that this is not a minimal set of rules \cite{coecke2018picturing}. 
\begin{enumerate}

    \item Merging/un-merging rule: Two $Z$ (or $X$) nodes with phases $\alpha$ and $\beta$ that are connected by edges with no Hadamards may be combined into a single $Z$ (or $X$) node with phase $\alpha + \beta$. The resulting node has all the external edges of the two original nodes. A node may also be un-merged into two nodes, and the partition of the external edges connected to each node can be done arbitrarily.

    \begin{minipage}{0.4\textwidth}
    \begin{flushright}
        \includegraphics[scale=0.30]{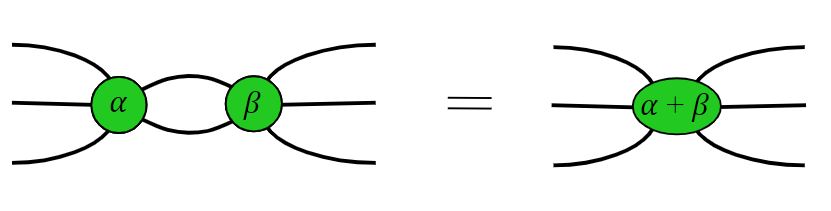}
    \end{flushright}
    \end{minipage}

    \item Identity removal rule: $Z$ (or $X$) nodes with phase 0 and exactly two edges can be removed.

    \begin{minipage}{0.30\textwidth}
    \begin{flushright}
        \includegraphics[scale=0.40]{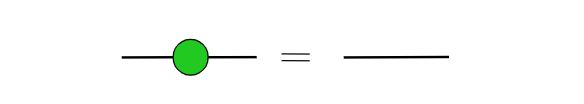}
    \end{flushright}
    \end{minipage}

    \item Hadamard cancellation rule: Two Hadamard gates, sharing an edge and both having exactly two edges, can be cancelled since $HH = I$.
    
    \item $\pi$-copy rule: A $Z$ (or $X$) node of phase $\pi$ slides through and copies onto all the other edges of an adjacent $X$ (or $Z$) node. The $X$ (or $Z$) node has its phase negated. This rule will be used repeatedly in Section \ref{sec: kl forms}.
    Shown below is an example with a $Z$ $\pi$ node.

    \begin{minipage}{0.4\textwidth}
    \begin{flushright}
        \includegraphics[scale=0.55]{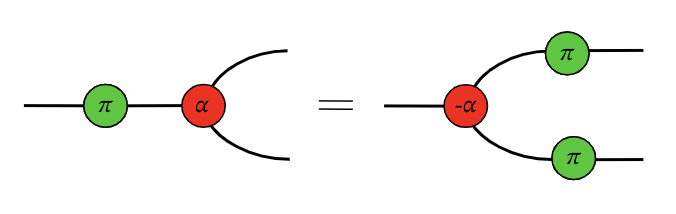}
    \end{flushright}
    \end{minipage}

    \item State-copy rule: A $Z$ (or $X$ node) with a phase of $a\pi$ with $a \in \{0,1\}$ can be copied through onto all of the other adjacent edges of an $X$ (or $Z$ node) with any phase $\alpha$.

    \begin{minipage}{0.4\textwidth}
    \begin{flushright}
        \includegraphics[scale=0.24]{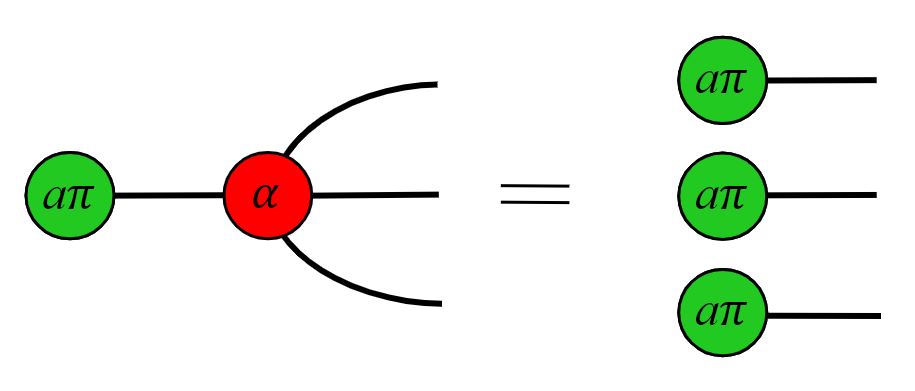}
    \end{flushright}
    \end{minipage}

    \item Color change rule: A $Z$ (or $X$) node can be exchanged for an $X$ (or $Z$) node if all edges adjacent to the node have Hadamards added to them.

    \begin{minipage}{0.4\textwidth}
    \begin{flushright}
        \includegraphics[scale=0.32]{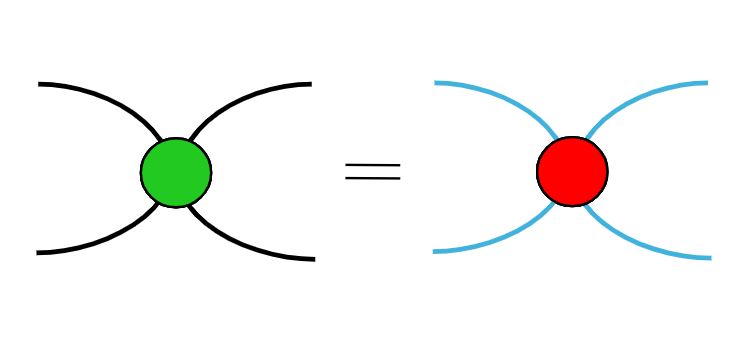}
    \end{flushright}
    \end{minipage}

    \item Bi-algebra rule: By acting on an edge between a $Z$ and an $X$ node, each external edge gets one node, and a complete bipartite graph is formed between these new nodes. An example is shown below. There may be one or more (rather than two) edges coming in from the left side of the graph, and there may be one or more edges exiting on the right side of the graph.

    \begin{minipage}{0.4\textwidth}
    \begin{flushright}
        \includegraphics[scale=0.25]{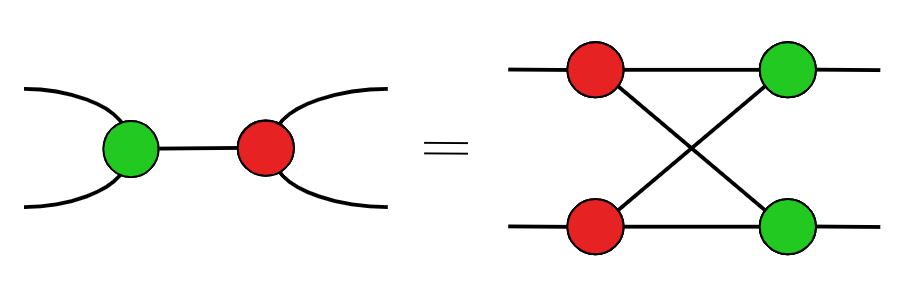}
    \end{flushright}
    \end{minipage}

    \item Hopf rule: If a $Z$ node and $X$ node share multiple edges that have no Hadamards, two of these shared edges may be removed together.

    \begin{minipage}{0.4\textwidth}
    \begin{flushright}
        \includegraphics[scale=0.20]{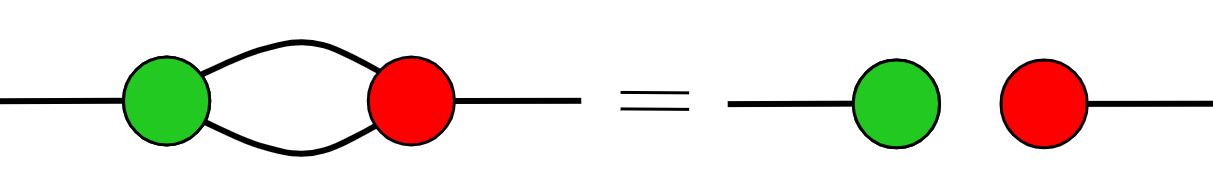}
    \end{flushright}
    \end{minipage}

\end{enumerate}

\end{definition}

Additional rules that will be used later are given below.

\begin{definition}[Derived rewrite rules]Below are two rules that can be derived from the set of basic rewrite rules above.
\label{definition: derived rewrite rules}

\begin{enumerate}
       \item Loop rule: Self-loops on a node can be removed. If the self-loop has a Hadamard, then removing the loop adds a phase of $\pi$ to the node.

    \begin{minipage}{0.4\textwidth}
    \begin{flushright}
        \includegraphics[scale=0.27]{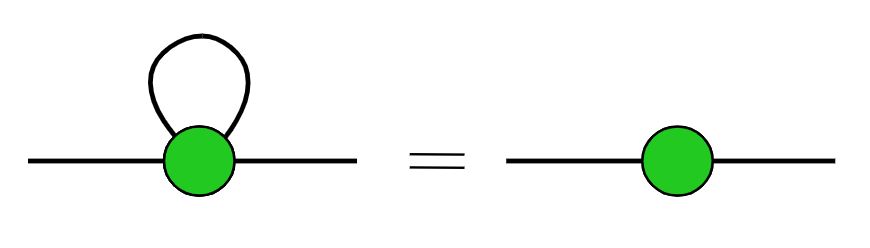}
    \end{flushright}
    \end{minipage}
       
    \item Hadamard-sliding rule: This rule allows the colors of two adjacent vertices to be swapped while switching neighbors and toggling the edges between the neighbors of the two vertices. 

    \begin{minipage}{0.4\textwidth}
    \begin{flushright}
        \includegraphics[scale=0.20]{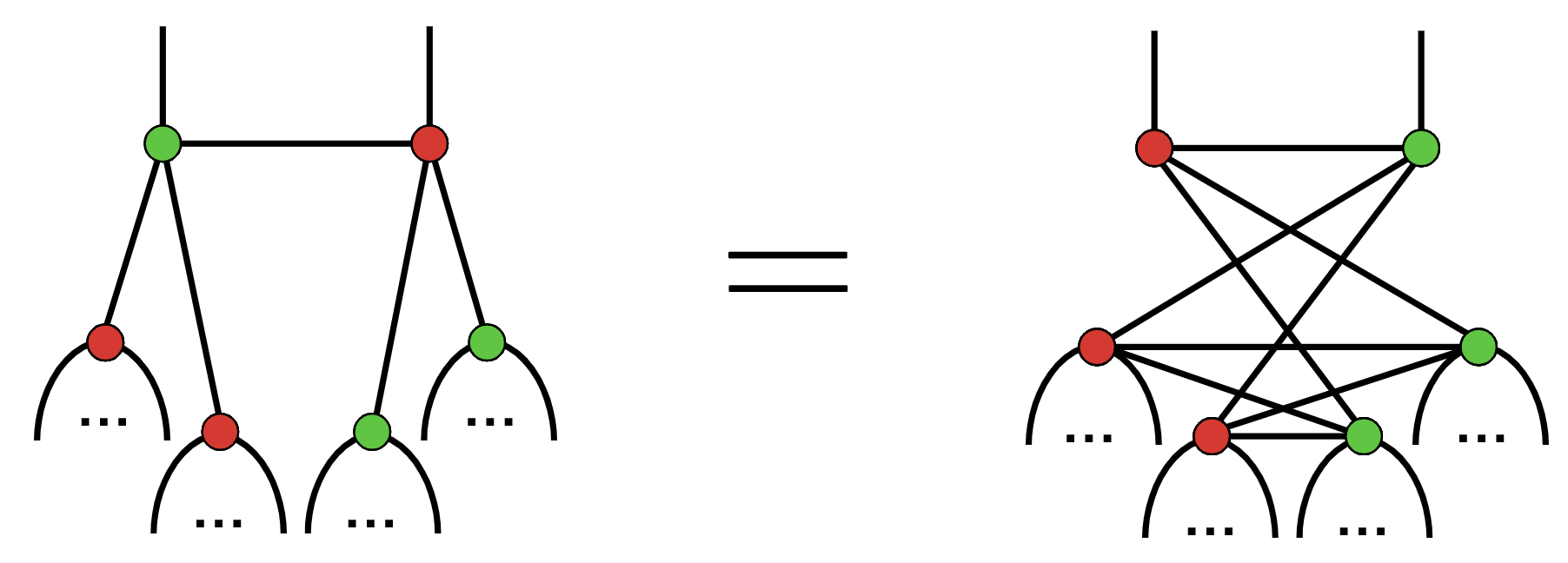}
    \end{flushright}
    \end{minipage}
\end{enumerate}

\end{definition}

We consider the ZX diagrams of codes that are expressed in their \textit{encoder-respecting form} \cite{KLS}, which we re-define here.

\begin{definition}
    \label{encoder-respecting form}
    The \textit{encoder-respecting form} is a ZX representation of a Clifford code that contains $Z$ and $X$ nodes, with each node corresponding to an input or output. Each node has a corresponding free edge (not connected to any other nodes), with input nodes having input edges and output nodes having output edges. The outputs may have local operations on their free edges. Each of the $k$ input nodes may only have connections with the output nodes, while each of the $n$ output nodes may have connections amongst each other. Also, the output edges are numbered from 1 to $n$.
\end{definition}

\begin{figure}[h]
\includegraphics[scale=0.45]{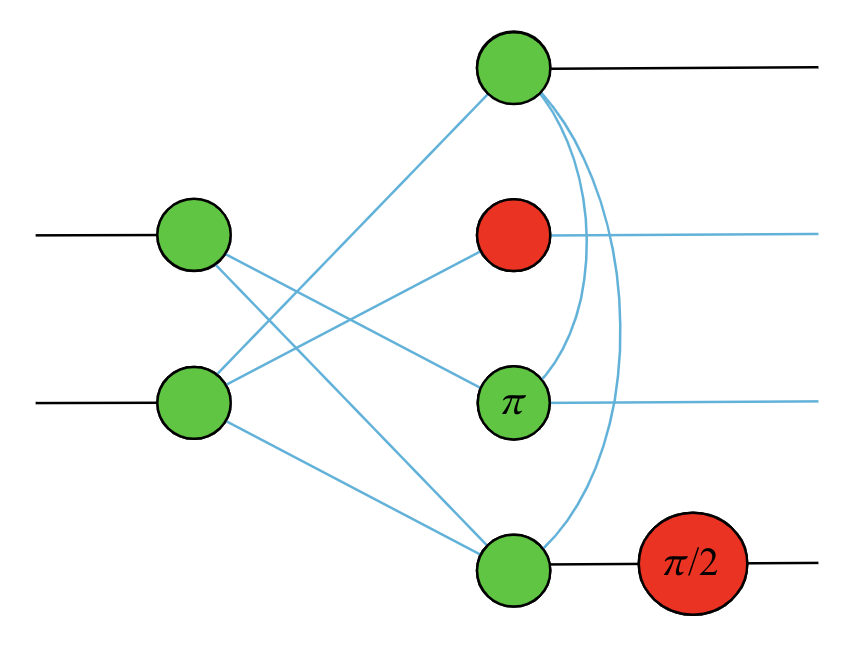}
\caption{\justifying Example of an encoder in ZX calculus. The incoming edges from the left side are input edges (sending information in) and the outgoing edges on the right side are output edges (sending encoded information out). Note the local operations applied on the output qubits, with blue edges representing edges with Hadamards.}
\label{examplee}
\end{figure}

An example of a Clifford code expressed in encoder-respecting form is shown in \Cref{examplee}. There are local operations on the output edges, as shown by the blue free edges and $X$ gate ($X$ node with phase $\pi/2$). The $X$ node is not considered an output node since it is isolated as a local operation on its neighboring $Z$ node. There are internal edges in the graph between the input and output nodes and amongst the output nodes, but not amongst the input nodes. Two qubits of are encoded into four qubits in this example.

Encoders with the same sets of stabilizers have the same Khesin-Lu-Shor (KLS) canonical form \cite{KLS}. This form consists of four rules that can be efficiently checked for a given ZX diagram.

\begin{theorem}[KLS canonical form for Clifford codes] There is a canonical form for the ZX diagrams of an encoder, where the ZX diagrams in the same equivalence class have identical stabilizers. This diagram is in encoding-respecting form, and all nodes are $Z$ nodes. The canonical form satisfies the following four rules:
\begin{enumerate}
    \item \textit{Edge rule:} All internal edges have Hadamards, and there is exactly one $Z$ node per free edge.
    \item \textit{Hadamard rule:} Output nodes with Hadamard gates on their free edges cannot share an edge with a lower-numbered output node or with an input node.
    \item \textit{RREF rule:} The adjacency matrix representing the edges between input nodes and output nodes is in reduced row-echelon form (RREF).
    \item \textit{Clifford rule:} In the RREF matrix, the pivot columns of the input to output adjacency matrix correspond to pivot output nodes. There are no local Clifford operations on the pivot or input nodes, or their free edges. There are also no input-input edges or pivot-pivot edges.
    
\end{enumerate}
\end{theorem}

A ZX diagram satisfying these four rules is the unique KLS canonical form for the equivalence class where all the ZX diagrams have the same stabilizers. Also, a given ZX diagram can be efficiently transformed to its KLS canonical form using a series of operations.

Furthermore, the KLS canonical forms may be efficiently transformed into quantum circuits.

\begin{theorem}
Consider an $[n,k]$ encoder given in its KLS form. Then, it can be efficiently transformed into an equivalent quantum circuit using the following steps.

\begin{enumerate}
\item Start with $k$ open wires representing the inputs of the circuit.
\item Add a $\ket{0}$ state for each of the $n-k$ non-pivot output nodes.
\item Apply an $H$ gate to all $n$ wires.
\item Apply a $CX$ gate between the wires corresponding to the edges between inputs and non-pivot outputs. The input node is the target qubit, and the output node is the controlled qubit.
\item Apply a $CZ$ gate between the wires corresponding to the edges between only outputs.
\item Apply the local operations attached to the outputs.
\end{enumerate}
   
\end{theorem}

The complete proof of this theorem can be found in Appendix \ref{appendix sec: convert KLS to circuit}. This procedure works by building up the encoder's quantum circuit representation in layers, starting from the input qubits, which correspond to the ZX diagram input nodes, adding auxiliary qubits that encode the information from the input qubits, and connecting the wires using the appropriate gates.

The \textit{neighborhood} $N(v)$ of a vertex $v$ in a graph $G=  (V,E)$ is the set of all vertices in $V$ adjacent to $v$, not including $v$ itself. An operation commonly used to transform equivalent ZX diagrams between each other is defined below.

\begin{definition}
    \label{localcomplementation}
    Consider a simple graph $G = (V,E)$, where $V$ is the set of vertices and $E$ is the set of undirected edges between vertices. Consider a vertex $v\in V$. By performing a \textit{local complementation} about vertex $v$, all edges connecting two vertices in $N(v)$ are toggled. That is, if the edge existed before the local complementation, it is removed; if it did not exist before, it is added.
\end{definition}

\section{Canonical form for CSS codes and states}
\label{sec: kl forms}

Calderbank-Shor-Steane (CSS) codes are a commonly studied class of quantum error-correcting codes constructed starting from two classical codes \cite{zarei2017strong, steane1999enlargement, sarvepalli2009sharing, harris2018calderbank}. The generators of a CSS code's stabilizers can be chosen such that each of the generators is either a Pauli operator with only $I$ and $Z$ gates or a Pauli operator with only $I$ and $X$ gates.

We introduce the new notion of a CSS state.

\begin{definition}

    A \textit{CSS state} is a CSS code with 0 inputs.
\end{definition}

We now present the Khesin-Li (KL) canonical form of CSS codes and CSS states in ZX calculus, which we prove later in this section. The KL form for CSS codes is a special case of the KLS forms (which are for Clifford codes).

All KL forms are in encoder-respecting form (see Definition \ref{encoder-respecting form}), so every input and output node has its own free edge, and the output nodes are numbered from $1$ to $n$.

\begin{figure}
    \includegraphics[scale=0.35]{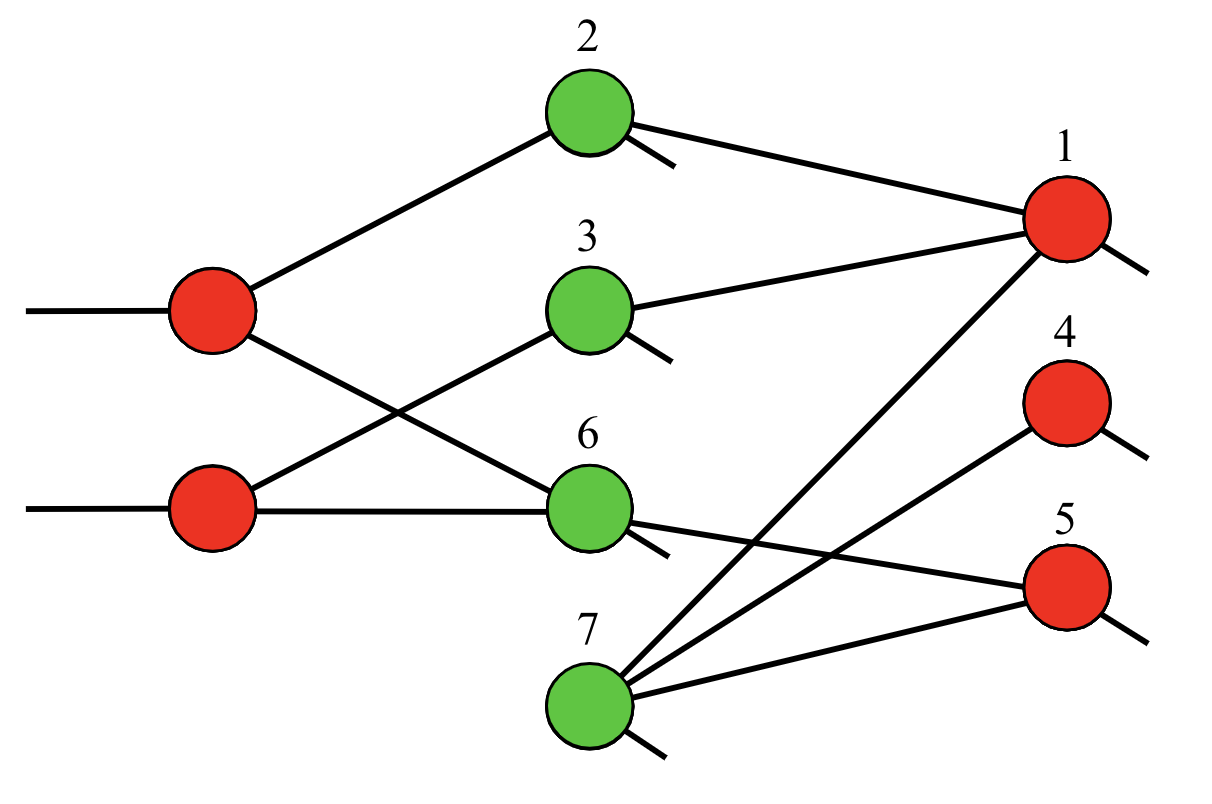}

    \caption{\justifying
        Example of a CSS code in KL canonical form. The short diagonal edges are the free output edges of the 7 output nodes.}
            \label{fig: KL form example}

        \end{figure}

\begin{theorem}[KL canonical form for CSS codes]
\label{CSS code canonical form}
Any CSS code can be expressed uniquely in ZX calculus under the following rules.
\begin{itemize}
\item \textit{Bipartite rule:} The nodes can be split into two groups, one consisting of the input nodes (which are all $X$ nodes) and output $X$ nodes, and one consisting of the output $Z$ nodes. The only interior edges allowed are between nodes of different groups.
\item \textit{Phase-free rule:} All input and output nodes have phase 0, and there are no local operations on any free edges.
\item \textit{RREF rule:} The adjacency matrix between the input $X$ nodes and output $Z$ nodes is in reduced row echelon form (RREF). The adjacency matrix between the output $X$ nodes and \textit{all output nodes}, where the output $X$ nodes are marked as connected to themselves, is also in RREF.

    
\end{itemize}

\end{theorem}

\begin{remark}
    Note that the first part of RREF rule is equivalent to the one set in the KLS form, in which the adjacency matrix between the inputs and outputs is in RREF. Additionally, note that the second part of the RREF rule is equivalent to stipulating that no output $X$ node be connected to a lower-numbered $Z$ node. This implies that the second part of the RREF rule is equivalent to the Hadamard rule set in the KLS form, in which outputs without Hadamards cannot connect to lower-numbered output nodes or input nodes.
\end{remark}

An example of the KL form is shown in Figure \ref{fig: KL form example}. The Bipartite rule is made clear by the division of the nodes into columns, with the first and last columns of nodes forming the first group ($X$ nodes) and the middle column forming the second group ($Z$ nodes). The edges satisfy the constraint that the output $X$ nodes only connect to higher-numbered output $Z$ nodes. Lastly, the input $X$ nodes to output $Z$ nodes adjacency matrix is in RREF.

In CSS states, the part in the RREF rule pertaining to input nodes is unnecessary because of the lack of input nodes while the Bipartite and Phase-free rules still apply. As a corollary, we also find the canonical form for CSS states.
\begin{corollary}[Canonical form for CSS states]
\label{canonical of CSS state}
    Any CSS state can be expressed uniquely in ZX calculus under the following rules.
    \begin{itemize}
        \item \textit{Bipartite rule:} The nodes can be split into two groups, one with $X$ nodes and the other with $Z$ nodes. The only interior edges allowed are between nodes of different groups.
        \item \textit{Hadamard rule:} Each $X$ node can only connect to higher-numbered $Z$ nodes.
        \item \textit{Phase-free rule:} All output nodes have phase 0, and there are no local operations on any free edges.
    \end{itemize}
\end{corollary}

Now, we build up to the proof of Theorem \ref{CSS code canonical form}. To do this, we show that each KL form corresponds to a distinct CSS code by counting the number of CSS codes and comparing it to the number of KL diagrams, and we later show that any KL form can be converted into the stabilizer representation of the CSS code. These two steps establish the bijection between the KL forms and CSS codes. We start by finding the number of CSS codes.

\begin{lemma}
\label{lemma: number of CSS codes}
    The number of CSS codes with $n$ physical qubits, $p$ $Z$ stabilizers, and $q$ $X$ stabilizers is
\begin{equation}\frac{\prod\limits_{i=1}^{p} (2^n - 2^{i -1})}{\prod\limits_{i =1}^p(2^p - 2^{i  -1})}\cdot \frac{\prod\limits_{i  =1}^q (2^{n-p} - 2^{i-1})}{\prod\limits_{i=1}^q (2^q - 2^{i-1})}.\label{eq: CSS states}\end{equation}
\end{lemma}

\begin{proof}
    Consider choosing the $Z$ stabilizers of the CSS code's stabilizer tableau. Because each stabilizer must be linearly independent from each other, there are $\prod\limits_{i=1}^p(2^n - 2^{i-1})$ ways to choose $p$ independent stabilizers. However, different sets of $p$ generators could represent the same set of stabilizers. For a given set of stabilizers, there are $2^p - 2^{i-1}$ ways to choose the $i$th generator, so we must divide to find $$\frac{\prod\limits_{i=1}^{p} (2^n - 2^{i -1})}{\prod\limits_{i =1}^p(2^p - 2^{i  -1})}$$as the number of ways to choose the set of $Z$ stabilizers.

    Next, each $X$ stabilizer must commute with all of the $Z$ stabilizers. Since the number of $X$ stabilizers that commute or anti-commute with any single $Z$ stabilizer is equal, there are $2^n /2^p = 2^{n-p}$ $X$ stabilizers to choose from. Using a similar analysis as above, there are
    $$\frac{\prod\limits_{i  =1}^q (2^{n-p} - 2^{i-1})}{\prod\limits_{i=1}^q (2^q - 2^{i-1})}$$ways to choose the set of $X$ stabilizers. 

    Multiplying these two counts gives the number of CSS codes with the given parameters, as desired.
\end{proof}

Now, we find the number of KL diagrams with similar parameters. We will show later that the parameters chosen below make the KL forms correspond exactly to the CSS codes with $n$ outputs, $p$ $Z$ stabilizers, and $q$ $X$ stabilizers. 

Consider a KL form of a CSS code. Let there be $p$ output $X$ nodes and $n - p  -q = k$ input nodes, so that there are $n - p$ output $Z$ nodes, of which $k$ are pivot nodes in the RREF adjacency matrix between the input nodes and output $Z$ nodes.

First, consider the output nodes. Note that there are $n$ total nodes among the $X$ and $Z$ groups. To count the number of ways to connect edges between the output $X$ and $Z$ nodes in the KL form, note that the Hadamard rule restricts the connections to those from lower-numbered output $X$ nodes to higher-numbered output $Z$ nodes.

\begin{lemma}
    \label{lemma: number of ways to connect green output to red output}
    In the KL form, the number of ways to connect the $n-p$ output $Z$ nodes and $p$ output $X$ nodes is
    $$\frac{\prod\limits_{i=1}^{p} (2^n - 2^{i -1})}{\prod\limits_{i =1}^p(2^p - 2^{i  -1})}.$$
\end{lemma}

\begin{proof}
In the proof of this lemma, we only consider the output nodes of the encoder.

Consider constructing ({``building up"}) the ZX diagram. Each time an output node is added, some number of possible connections is possible between the newly added output node and the nodes that have already been placed. The nodes are added in order, so each node has a higher index than all the nodes that came before it.

Let $p'$ be the number of remaining $X$ nodes that need to be added in the ZX diagram. Let $q'$ be the number of remaining $Z$ nodes that need to be added in the ZX diagram. The function $f(p,p',q')$ counts the number of ways to add nodes and edges starting from some arbitrary state that has $p'$ remaining $X$ nodes and $q'$ remaining $Z$ nodes. The total number of $X$ nodes after placing all nodes will be $p$. We find a recursive relation for $f(p,p',q')$.

When $p' = 0 $ and $q' \ne 0$, all of the remaining $Z$ nodes can be added in and connected arbitrarily to the $p$ pivots in $2^{pq'}$ ways, giving
$$f(p,0,q') = 2^{pq'}.$$When $p' \ne 0 $ and $q' = 0$, all of the remaining $X$ nodes can be added in, but they cannot connect to anything since they necessarily are higher-numbered than all of the $Z$ nodes. Therefore,
$$f(p,p',0) = 1.$$

When $p' \ne 0$ and $q' \ne 0$, either a $X$ node is added or a $Z$ node is added (note there are $2^{p-p'}$ ways to connect edges between the new $Z$ node and $p-p'$ current $X$ nodes), giving the recursive equation
$$f(p,p',q') = f(p,p'-1, q') + 2^{p-p'}f(p,p',q'-1).$$

We can check that the following function $f(p,p',q')$ satisfies this recursive equation and the base cases considered above:

$$f(p,p',q') = 2^{(p-p')q'} \frac{\prod\limits_{i = 1}^{p'}(2^{q' + i} - 1)}{\prod\limits_{i = 1}^{p'}(2^i-1)}.$$

The total number of ways to connect the $n-p$ $Z$ and $p$ $X$ nodes must be $f(p, p, n-p)$. Evaluating the above expression with these parameters gives a count of 
$$ \frac{\prod\limits_{i = 1}^{p}(2^{n-p + i} - 1)}{\prod\limits_{i = 1}^{p}(2^i-1)} = \frac{\prod\limits_{i = 1}^{p}(2^{n} - 2^{i-1})}{\prod\limits_{i = 1}^{p}(2^p-2^{i-1})},$$as desired.

\end{proof}

Now, we count the number of ways to form connections between the input nodes and output $Z$ nodes.

\begin{lemma}
    \label{lemma: input to green output connections}
    In the KL form, the number of ways to connect the $k$ input nodes and $n-p$ output $Z$ nodes is
    $$\frac{\prod\limits_{i  =1}^q (2^{n-p} - 2^{i-1})}{\prod\limits_{i=1}^q (2^q - 2^{i-1})}.$$
\end{lemma}

\begin{proof}
    Now, we have to count the number of ways to connect the input nodes and output $Z$ nodes to form an RREF adjacency matrix between them. Imagine each input node and its corresponding pivot node (in the RREF matrix) as a single \textit{super-node}. Super-nodes may not connect with each other, since pivots may not connect with other inputs, by the RREF rule. Also, the pivot node is the lowest-numbered node among the $Z$ outputs that the input node connects to, so the super-node can only connect with non-pivot nodes that are higher-numbered.

    Therefore, we have $k$ super-nodes and $q = n-p-k$ single output $Z$ nodes, and connections are restricted to those between super-nodes and higher-numbered single output $Z$ nodes. This is equivalent to $f(k, k, q)$ from the proof of Lemma \ref{lemma: number of ways to connect green output to red output}, so the number of ways to connect the input nodes and output $Z$ nodes must be
$$ \frac{\prod\limits_{i = 1}^{k}(2^{q + i} - 1)}{\prod\limits_{i = 1}^{k}(2^i-1)} = \frac{\prod\limits_{i  =1}^q (2^{n-p} - 2^{i-1})}{\prod\limits_{i=1}^q (2^q - 2^{i-1})},$$giving the desired result.

\end{proof}

From Lemmas \ref{lemma: number of ways to connect green output to red output} and \ref{lemma: input to green output connections}, we find the total number of KL forms by multiplying the two expressions found, giving 
$$\frac{\prod\limits_{i=1}^{p} (2^n - 2^{i -1})}{\prod\limits_{i =1}^p(2^p - 2^{i  -1})}\cdot \frac{\prod\limits_{i  =1}^q (2^{n-p} - 2^{i-1})}{\prod\limits_{i=1}^q (2^q - 2^{i-1})}.$$

Since this matches with Lemma \ref{lemma: number of CSS codes}, the number of KL forms is indeed equal to the number of CSS codes, under certain parameters, so each KL form can correspond to a distinct CSS code. Now, we need to show that restricting the CSS codes to $n$ physical qubits, $p$ $Z$ stabilizers, and $q$ $X$ stabilizers is equivalent to restricting the KL forms to $n$ output nodes, $n-p$ output $Z$ nodes, and $p$ output $X$ nodes.

Starting from the KL form of a QECC, we can determine the $Z$ stabilizers, $X$ stabilizers, and logical $Z$ operators of the code.

\begin{lemma}
\label{lemma: red nodes have Z checks}
Each of the output $X$ nodes corresponds to its own linearly independent $Z$ check.
\end{lemma}

\begin{proof}
Consider node $i$, which is one of the output $X$ nodes. Then, we can determine a $Z$ stabilizer that includes the operator $Z_i$ by sliding a $Z$ gate through node $i$. By the $\pi$-copy rule from Definition \ref{def: basic rewrite rules},
this $Z$ gate splits into $Z$ gates on all the other incident edges of node $i$. Each of these $Z$ gates travels down an incident edge and combines with the output $Z$ node at the other end of the edge. Ultimately, this results in a phase of $\pi$ on all the neighbors of node $i$.


Since this is equivalent to placing $Z$'s on all the neighbors of node $i$, we have made a $Z$ stabilizer, which is the product of $Z_i$ and the $Z$ gates on the neighboring output $Z$ nodes.

The $Z$ checks of the output $X$ nodes are linearly independent from each other because each $Z$ check contains a $Z$ operator on a distinct output $X$ node.
\end{proof}

\begin{lemma}
\label{lemma: X checks for non-pivots}
    Each of the non-pivot output $Z$ nodes corresponds to its own linearly independent $X$ check.
\end{lemma}

\begin{proof}
    Consider node $i$, which is a non-pivot output $Z$ node. We can determine an $X$ stabilizer by sliding an $X$ gate through this node. By the $\pi$-copy rule, this $X$ gate splits into $X$ gates on all the other incident edges of node $i$. This results in a phase of $\pi$ on all the neighbors of node $i$. For each of the neighbors that are input nodes, an $X$ node of phase $\pi$ can be un-merged (see Definition \ref{def: basic rewrite rules} for the merging/un-merging rule) from the input $X$ node and slid towards the input's corresponding pivot node. By the $\pi$-copy rule, all other incident edges, including the free output edge, of the pivot node get a $X$ $\pi$ node. Then, the output $X$ neighbors of the pivot node receive additional phases of $\pi$.

    Since this is equivalent to placing $X$'s on all the pivots of input nodes connected to $i$ and placing $X$'s on all the output $X$ nodes that end up with phase $\pi$, we have made an $X$ stabilizer, which is the product of the $X$ gates just described and $X_i$.

    The $X$ checks of the non-pivot output $Z$ nodes are linearly independent from each other because each $X$ check contains an $X$ operator on a distinct non-pivot output $Z$ node.
    
\end{proof}

The \textit{logical operator} of a code maps an element of the codespace onto another element of the codespace. We denote the logical $Z$ operators as $Z_{L}$ and the logical $X$ operators as $X_L$. The $Z$ and $X$ stabilizers can be used to determine all the logical operators of a CSS code. Analogously, the $Z$ stabilizers and logical $Z$ operators of a CSS code can be used to determine all the $X$ stabilizers.

\begin{lemma}
\label{lemma: inputs have logical ops}
The adjacencies of the inputs determine the logical $Z$ operators.
\end{lemma}

\begin{proof}
Consider node $i$, which is one of the input nodes. Similar to Lemma \ref{lemma: red nodes have Z checks}, we can determine a logical $Z$ operator by sliding a $Z$ gate through node $i$.

Because all of the input nodes are $X$ nodes, the $Z$ gate will split into $Z$ gates onto all the incident edges of node $i$. Ultimately, all the neighbors of node $i$ will have a phase of $\pi$ due to $Z_{L,i}$. 

Since this is equivalent to placing $Z$'s on all these neighboring nodes of node $i$, we have made a logical $Z$ operator, which is the product of the $Z$ gates on the neighboring output $Z$ nodes.

Note that, in the adjacency matrix between the input nodes and output $Z$ nodes, each input node has a corresponding pivot. Therefore, the logical $Z$ operators determined by the input nodes will be linearly independent from each other because each has a $Z$ gate on a distinct element (i.e. the pivot node) in the set of output $Z$ nodes.
\end{proof}

We are now ready to prove our main result.

\begin{proof}[Proof of Theorem 3.2]

For a QECC with $n$ physical qubits, the number of stabilizers and logical operators adds up to $n$. Considering the KL form of a CSS code with $p$ red output nodes and $n-k$ input nodes, we find, from Lemmas \ref{lemma: red nodes have Z checks}, \ref{lemma: X checks for non-pivots}, and \ref{lemma: inputs have logical ops}, that there are $p$ $Z$ stabilizers, $q$ $X$ stabilizers, and $k = n-p-q$ logical $Z$ operators. Since these three numbers add up to $n$, all stabilizers and logical operators are accounted for by these three lemmas. This means that the KL forms with $p$ output $X$ nodes, $n-p$ output $Z$ nodes, and $k$ input nodes correspond exactly to the CSS codes with $n$ physical qubits, $p$ $Z$ stabilizers, and $q$ $X$ stabilizers.

From Lemmas \ref{lemma: red nodes have Z checks} and \ref{lemma: X checks for non-pivots}, we also see there is a clear way to convert from a KL form into a representation of the CSS code entirely in terms of its $Z$ stabilizers and $X$ stabilizers.

Then, because KL forms can be converted into a stabilizer representation and the number of KL forms is equal to the number of CSS codes of analogous parameters, it follows that there is a bijection between CSS codes with $n$ physical qubits, $p$ $Z$ stabilizers, and $q$ $X$ stabilizers and KL forms with $n$ output nodes, of which $p$ are red output nodes and $n-p$ are green output nodes, finishing the proof for Theorem \ref{CSS code canonical form}.
\end{proof}

This construction allows us to prove several propositions.
First, we complete correspondence between the KL form and the stabilizer tableau.
Lemmas \ref{lemma: red nodes have Z checks}, \ref{lemma: X checks for non-pivots}, and \ref{lemma: inputs have logical ops} explain how to construct both
types of stabilizers and $Z$ logical operators.
Although this is enough to determine the $X$ logical
operators, we can also use the following.

\begin{proposition}\label{prop: logical x ops}
   The adjacencies of the pivots determine the logical $X$ operators.
\end{proposition}
\begin{proof}
    Consider node $i$, which is one of the input
nodes. We can determine a logical $X$ operator by sliding a $X$ gate through node $i$.
Since all inputs nodes are $X$ nodes and the $X$ gate is an $X$ node with a phase of $\pi$, we can merge and unmerge it with the input node to move it along the edge connecting the input to its pivot, an output $Z$ node.
Then, we use the $\pi$-copy rule to turn the $X$ gate
into $X$ gates on each other edge connected to the pivot.
Note that this means that an $X$ gate will be placed on the pivot's free edge as well as each other internal edge connected to the pivot.
Since the pivot node cannot be connected to any inputs other than $i$ and the graph is bipartite, this means that these $X$ gates can be moved towards the output $X$ nodes and then merged and unmerged along each of those outputs' free edges.

Note that, in the adjacency matrix between the
input nodes and output $Z$ nodes, each input node
has a corresponding pivot. Therefore, the logical $X$
operators determined by the input nodes will be linearly independent from each other because each has
an $X$ gate on a distinct element (i.e. the pivot node).
\end{proof}




\begin{remark}
    The stabilizers of a CSS code need to commute with each other and with each logical operation, while the logical $Z$ and $X$ operations on a single input should anti-commute.
    From \ref{lemma: inputs have logical ops} and \ref{prop: logical x ops} we see that logical $Z$ and $X$ operations can only overlap on pivot nodes, and since each only has a local operation on the pivot node corresponding to the input of the logical operator, the only anti-commuting operations will be logical $Z$ and $X$ operations on the same input.

    For stabilizer commutation, we see from \ref{lemma: red nodes have Z checks} and \ref{lemma: X checks for non-pivots} that a $Z$ and $X$ stabilizer can only intersect at an $X$ output node or at its $Z$ output neighbours, as those are the elements of a $Z$ stabilizer.
    Each of the $Z$ output nodes present in the $Z$ stabilizer will also place a term on the neighbouring $X$ output node.
    This will make the parity of total overlapping nodes even, making sure the stabilizers commute.

    The same argument as above shows why $Z$ stabilizers commute with logical $X$ operations. Since $X$ stabilizers only have gates on $Z$ output nodes at the corresponding non-pivot node and on several pivot nodes, a logical $Z$ operation will either not intersect it or intersect the stabilizer exactly twice, at a pivot and a non-pivot connected to the input of the logical $Z$.
    Either way, this proves that the stabilizers commute with the logical operations and each other.
\end{remark}

We may consider CSS codes in terms of only its $Z$ stabilizers and logical $Z$ operators, which is equivalent to the usual representation of the codes in terms of their $Z$ stabilizers and $X$ stabilizers. From Lemmas \ref{lemma: red nodes have Z checks} and \ref{lemma: inputs have logical ops}, the $Z$ stabilizers and logical $Z$ operators can be found directly from the KL form by looking at the connections to output $Z$ nodes and input nodes, respectively. 


The above construction allows us to easily transform a KL diagram into a stabilizer tableau.
The reverse is also easy to accomplish by applying the following procedure.
If we start with a stabilizer code, we first compute its $Z$ stabilizers and logical $Z$ operations.
We then row-reduce the list of $Z$ stabilizers, identifying the pivot columns.
We create $n$ output nodes, with an $X$ node for each pivot of the row-reduced $Z$ stabilizers and a $Z$ node for each non-pivot.
Note that these are not the same as the pivot and non-pivot nodes in the context of the adjacencies from the input nodes.
Each output node is connected to a free edge and the connections between the $Z$ and $X$ output nodes are made in accordance with the operations in the row-reduced list of $Z$ stabilizers.
The list of logical $Z$ operations is then transformed by the stabilizer operations until none of the logical operations have any terms on nodes which are pivots of the stabilizers.
The remaining logical operations will only have terms on the $Z$ output nodes, at which point they can be row-reduced and connected to input $X$ nodes, using the row-reduced matrix as an adjacency matrix.
This completes the construction of a KL form from a stabilizer tableau.

In some cases, it is possible to further reduce the number of nodes in a diagram in its canonical KLS form. Some of these simplifications are also relevant to the KL form.

\begin{definition}
    The \textit{reduced KLS form} can be found by completing the following steps:
    \begin{enumerate}
        \item We first remove any output nodes of phase zero and degree 2 using the identity removal rule.
        \item If an output node now has more than one free output edge, is only connected to an input node, and the input node is only connected to this output node, the Hadamards on the output node's free edges can be manipulated so that the topmost free edge of the node has no Hadamard.
        \item Any input nodes of phase zero and degree 2 may be removed.
    \end{enumerate}
\end{definition}

The second step to the reduced KLS form can be done by using the Hadamard-sliding rule (see the definition \ref{definition: derived rewrite rules}) on the output node and a connected input node. Unitary operations that appear on the input node. The output node's free edges all get an additional Hadamard gate due to the rewrite rule, causing the toggling of Hadamards on all the free edges. Out of the two possible Hadamard configurations, one configuration results in the topmost free edge having no Hadamard.

These reduced forms allow for further (although limited) reduction of the number of nodes in the KLS form of a Clifford code.


\section{Toric and Surface codes}
\label{section:surfacetoric}

As an example of the CSS codes explored in the previous section, we now consider toric codes, a specific class of surface codes with periodic boundary conditions, as introduced in \cite{kitaev2003fault} and certain surface codes, as explained in \cite{kissinger2022phase}. We first construct a symmetrical form for the toric code starting from Kissinger's \textit{ZX normal form}, a representation of CSS codes in ZX calculus that clearly shows stabilizers using internal nodes representing quantum measurements and shows logical operators using connections from input nodes to output nodes. The symmetrical form reduces the number of nodes, and the KL canonical form can be quickly derived starting from this new form. We show a similar symmetrical form for some square surface codes with an odd number of output nodes.

\begin{figure}[h]

\begin{center}
\includegraphics[scale=0.33]{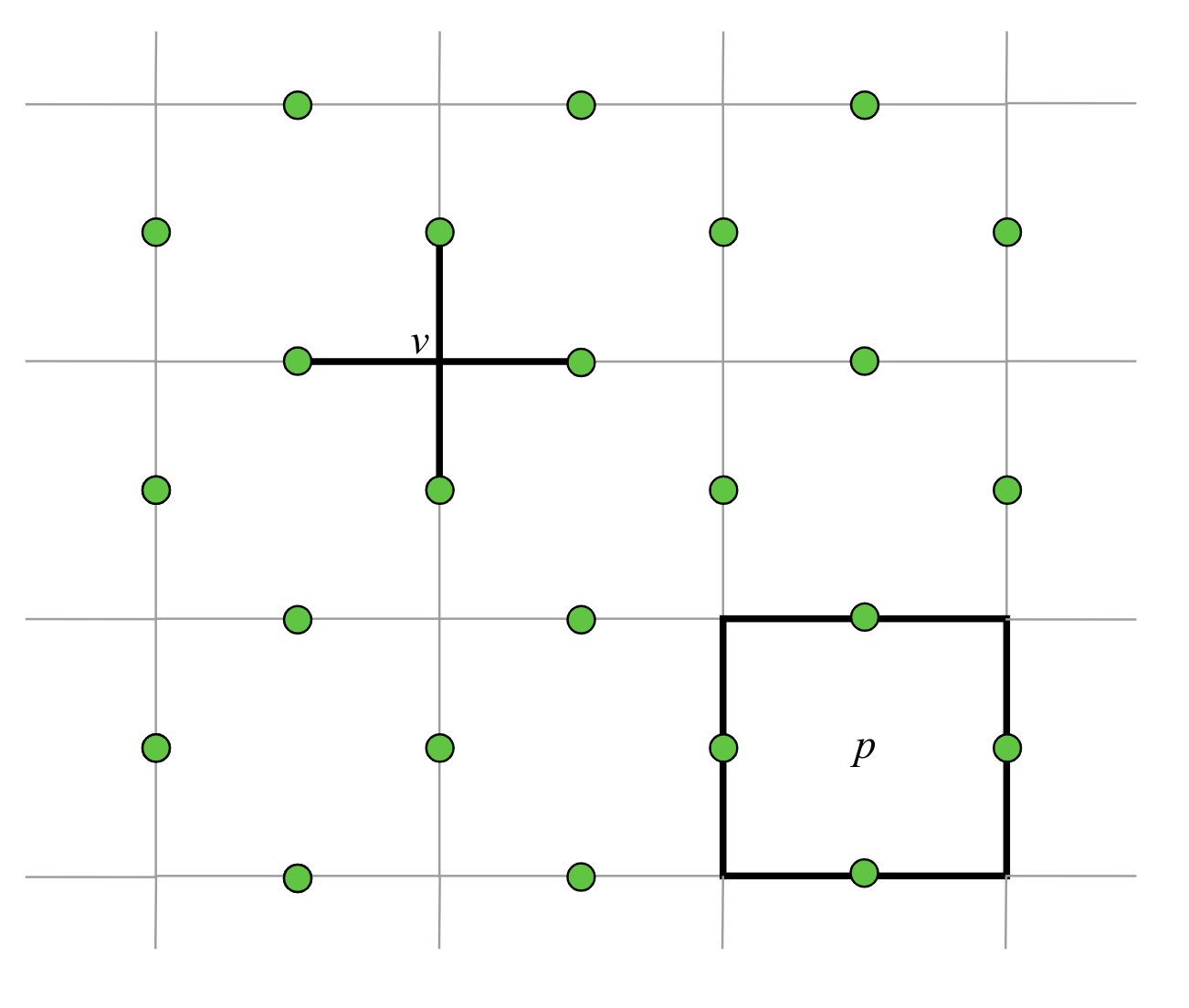}
\end{center}

    \caption{\justifying A section of the torus after placing it onto a 2-dimensional plane. The stabilizers corresponding to the vertices (\textit{v} is an example) have $X$ gates on the nodes immediately surrounding the vertex. The stabilizers corresponding to the plaquettes ($p$ is an example) have $Z$ gates on the nodes immediately surrounding the plaquette. All nodes have a default green color.}

    \label{generic toric code diagram}

\end{figure}

\begin{figure}
    \centering
    \includegraphics[scale=0.6]{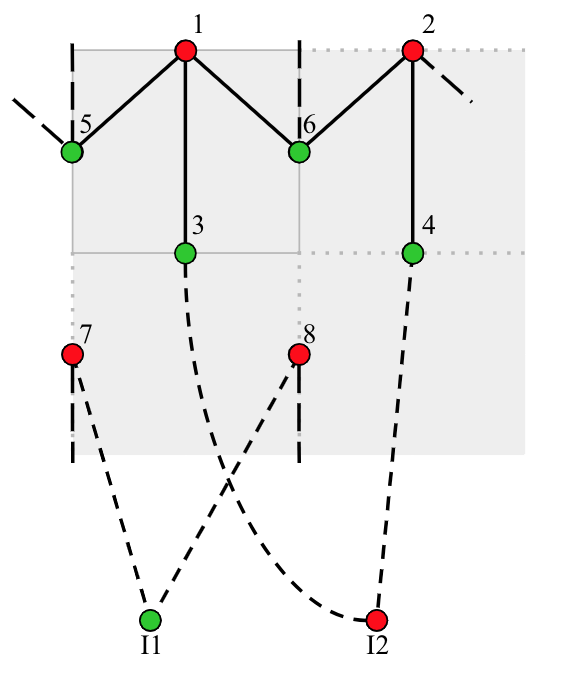}
    
    \caption{\justifying The 2-by-2 toric code in ZX calculus. Longer-dashed edges represent edges that wrap around the torus. For example, the vertical dashed edge coming from node 5 meets node 7 and the dashed edge from node 2 meets node 5. The shorter-dashed edges are input-output edges. The free input and output edges are not shown.}
    \label{final 2 by 2}
    
\end{figure}
\begin{figure}
    \centering
    \includegraphics[scale=0.55]{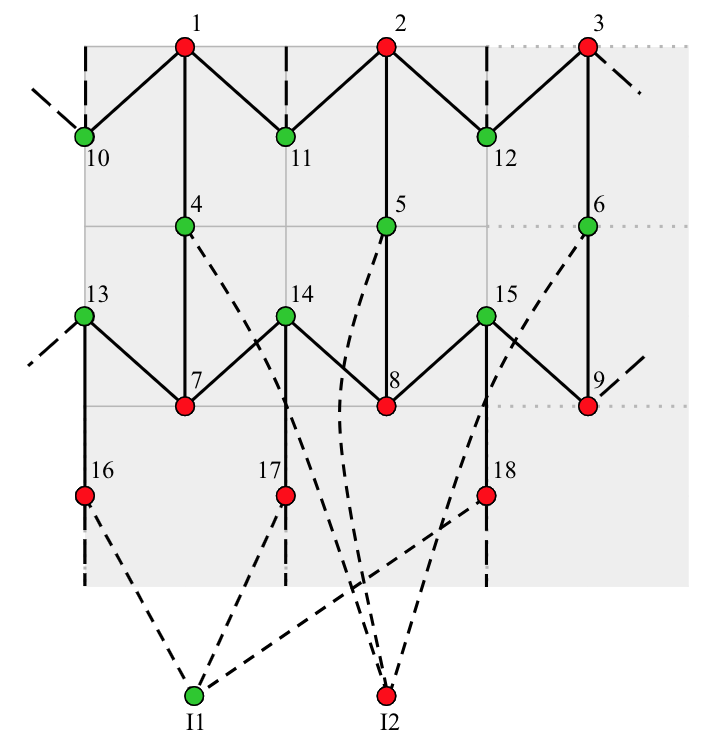}
    \caption{\justifying The 3-by-3 toric code in ZX calculus. Longer-dashed edges represent edges that wrap around the torus. For example, the vertical dashed edge coming from node 10 meets node 16 and the dashed edge from node 3 meets node 10. The shorter-dashed edges are input-output edges. The free input and output edges are not shown.}
    \label{final 3 by 3}
\end{figure}

\begin{figure*}
\subfloat[{\justifying The ZX normal form of the 4-by-4 toric code. The output nodes, which are all the $Z$ nodes, are shown with free edges protruding from them. The $X$ nodes in the toric grid are internal nodes.}]{\includegraphics[scale=0.30]{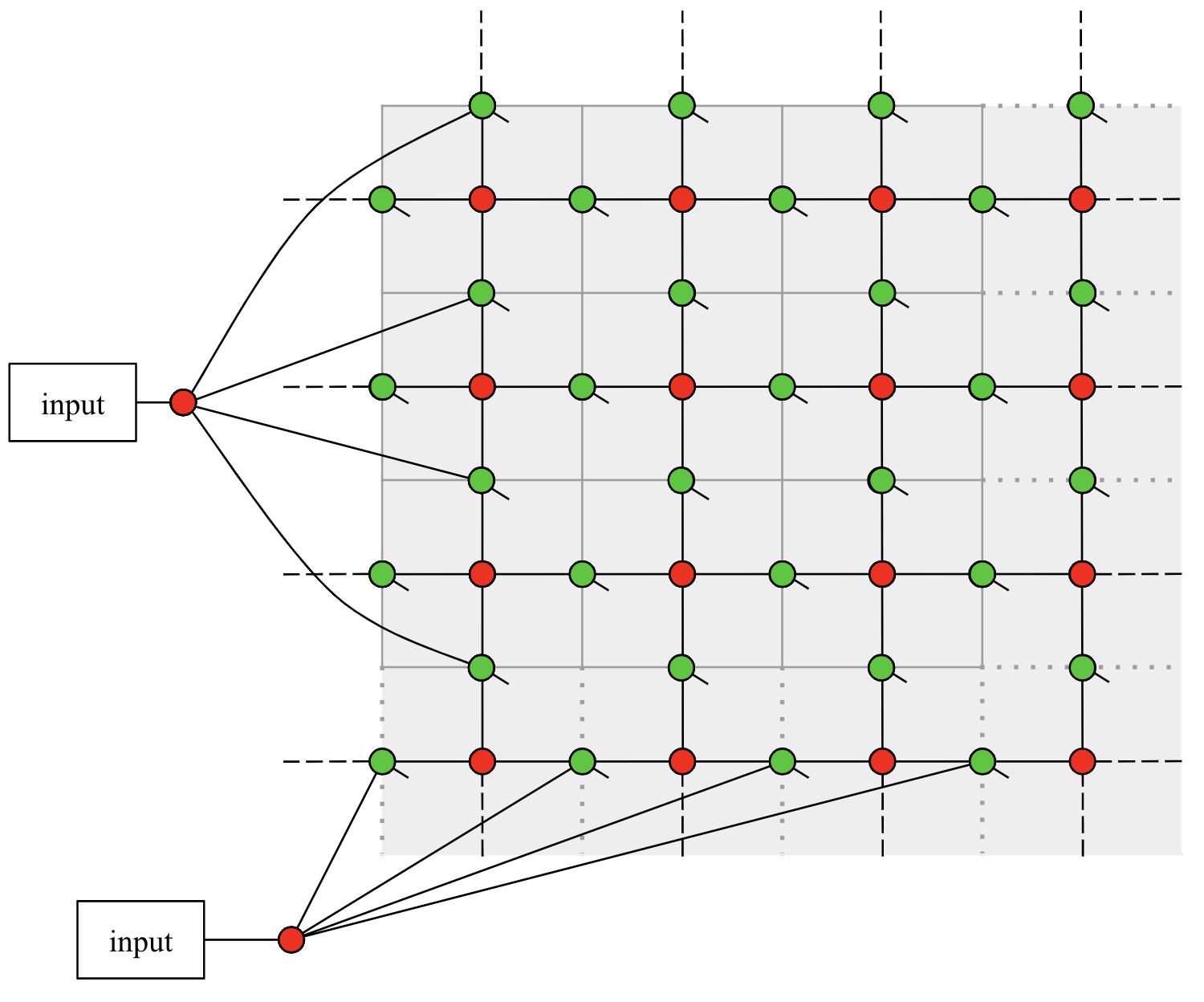}}
\hfill
\subfloat[{\justifying A symmetrical form of the 4-by-4 toric code. There is a symmetry between the $X$ and $Z$ nodes. Though this is not in KL form, it is easily converted into the KL form.}]{\includegraphics[scale=0.36]{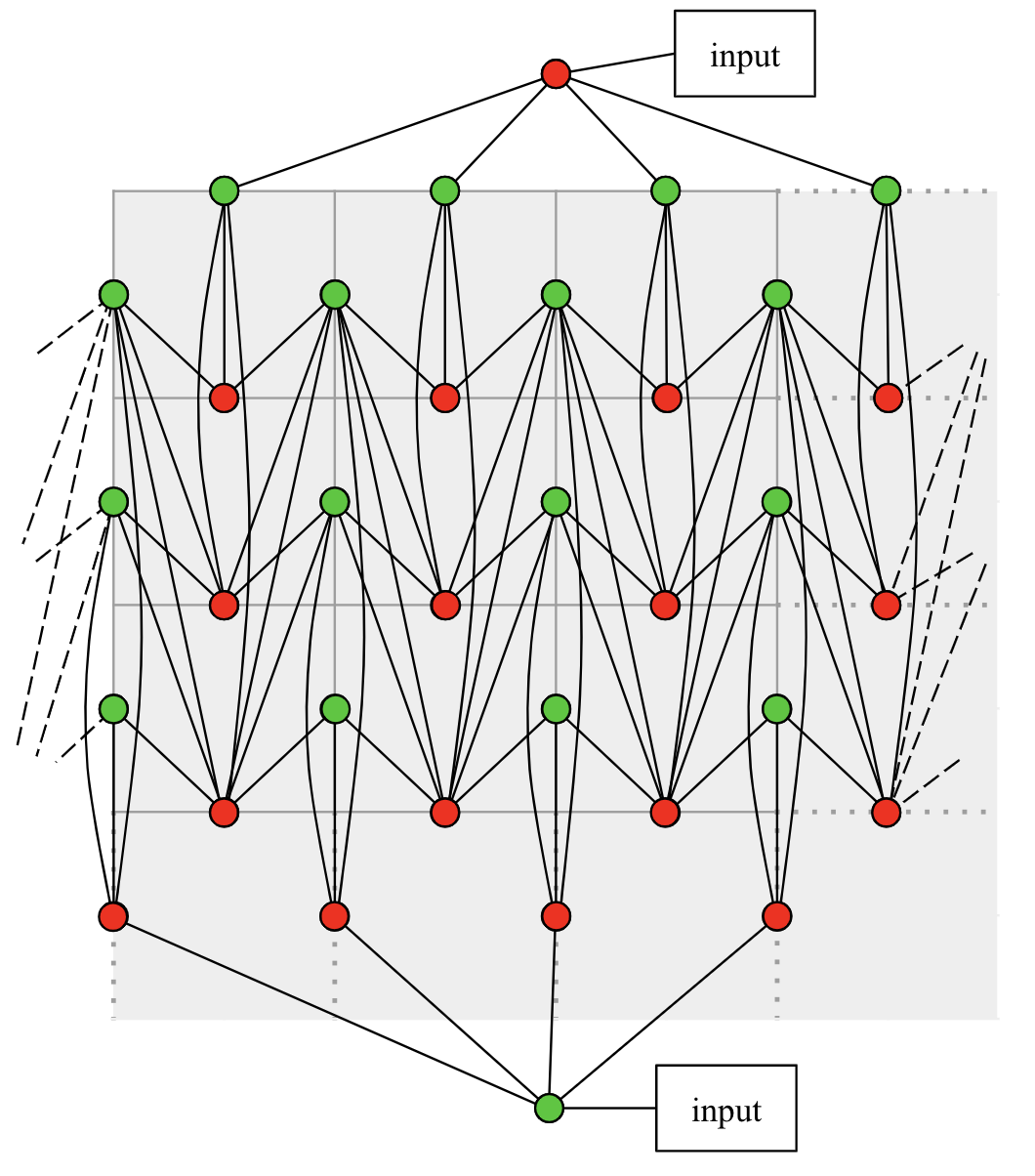}}

\caption{\justifying The 4-by-4 toric code, shown in two equivalent ZX diagrams. Note that, while the ZX normal form is local in both the vertical and horizontal directions, it has $4^2 = 16$ more nodes than the symmetrical form in (b).}
\label{4-by-4}
\end{figure*}

We begin with a definition of toric codes.

\begin{definition}
\label{defn: toric code}
    A \textit{toric code} is a quantum error-correcting code that can be represented on a three-dimensional torus $\mathcal{T}$. For an $m \times n$ toric code, $\mathcal{T}$ is wrapped by $m-1$ circles parallel to the plane of the major circle and $n-1$ circles perpendicular to the plane of the major circle. 

    A node is placed at the midpoint of each of the $2mn$ edges on $\mathcal{T}$. The stabilizers are defined as follows.
    
    The four nodes surrounding each of the $mn$ four-sided faces form a $Z$ check (stabilizer with only $Z$'s and $I$'s) consisting of $Z$'s on these four nodes and $I$'s on all other nodes.

    The four nodes surrounding each of the $mn$ intersections form an $X$ check (stabilizer with only $X$'s and $I$'s) consisting of $X$'s on these four nodes and $I$'s on all other nodes.

    An illustration of these stabilizers are shown in Figure \ref{generic toric code diagram}.
\label{def:thing}
\end{definition}

We now determine the structure of the symmetrical forms of the 2-by-2 and 3-by-3 toric code's ZX diagram. We do this by using the stabilizers to deduce the output edges with Hadamards and internal edges between output nodes.

We present the ZX calculus form of the 2-by-2 toric code in \Cref{final 2 by 2}. This has an arrow-like structure among the output-output edges, as seen by the group of nodes 1, 3, 5, and 6, as well as nodes 2, 4, 5, and 6. 

Also, the resulting 3-by-3 toric code is shown in \Cref{final 3 by 3}, with its full derivation given in Appendix \ref{appendix sec: construct 3-by-3 toric}. Note that, by wrapping this pattern around a torus, it would be horizontally periodic. The edges among vertices 1, 4, 10, and 11 form an upward-arrow-like structure. Similarly, nodes 2, 5, 11, and 12 form this structure and, on a torus, nodes 3, 6, 10, and 12 do so as well. By the simplicity of this diagram, it is relatively easy to read off the stabilizers by placing gates on the free output edges and seeing how the $\pi$-copy rule affects the diagram.


\begin{figure*}
    \centering
    \includegraphics[scale=0.4]{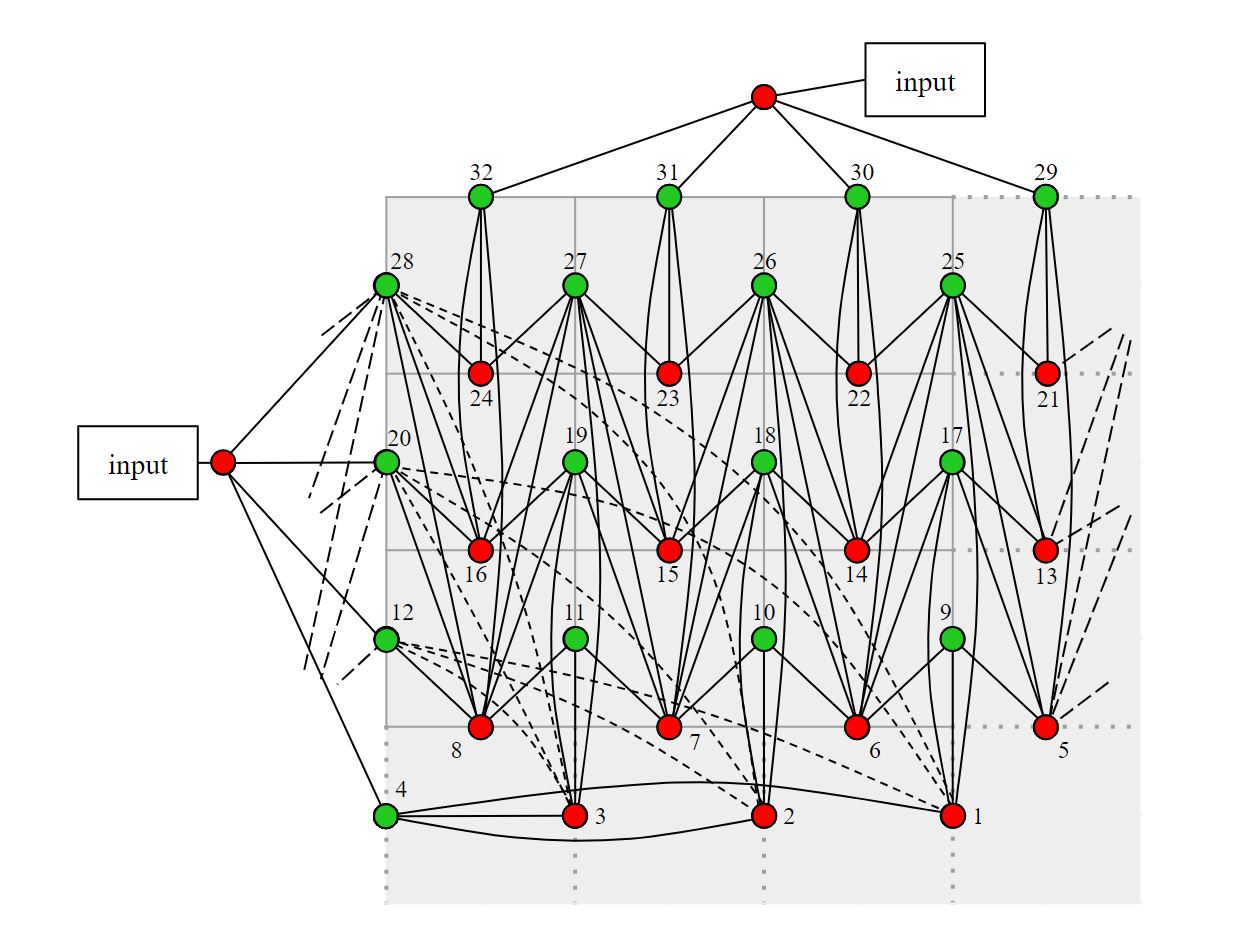}

    \caption{\justifying The 4-by-4 toric code in KL form. Long-dashed edges on the left and right edges of the grid wrap around the torus. The curved edges between nodes 1, 2, 3 and nodes 12, 20, 28 are dashed for clarity. The free output edges are not shown.}
    \label{fig: 4-by-4 toric code KL form}
\end{figure*}

For larger toric codes and the surface codes, we use the ZX calculus software \textit{Quantomatic} \cite{quantomatic} to simplify the known ZX normal form \cite{kissinger2022phase} of a code into its canonical form. In our algorithm, the main focus is on performing the bialgebra rule on internal nodes, so that, after running the first part of the algorithm, all internal nodes will be removed from the diagram, leaving only input and output nodes. Then, the Hadamard sliding rule, from Definition \ref{definition: derived rewrite rules}, will provide the operation that can repeatedly moves Hadamards until the encoder diagram is symmetric.

The procedure we follow may be written as the following algorithm.
\begin{enumerate}
\item Use basic simplifications, by merging nodes, applying the state-copy rule, applying the Hopf rule, removing scalars, removing loops, or combining two Hadamards into the identity (see Definition \ref{def: basic rewrite rules} and \ref{definition: derived rewrite rules}).
\item Apply one iteration of the bialgebra rule (see \ref{def: basic rewrite rules}) that removes an internal node. Then, apply step 1 again.
\item Apply step 2 until all internal nodes are removed.
\item Apply the Hadamard-sliding rule (see Definition \ref{definition: derived rewrite rules}) until the colors of the nodes are (mostly) alternating. (Note: In the toric code, it turns out that it is impossible for the colors to alternate every row, but the main section of nodes have alternating colors every row.)
\end{enumerate}

The reason step 2 works is that internal nodes are absorbed into neighboring nodes in the bialgebra rule (see Definition \ref{def: basic rewrite rules}).


To this end, we use the above algorithm to derive the general ZX diagram for the $m$-by-$n$ toric code.

First, we present our symmetrical form of the 4-by-4 toric code, which was derived from the ZX normal form. These are both shown in \Cref{4-by-4}. As can be seen in the diagrams, the number of output nodes is reduced by a factor of 2, and the diagram in \Cref{4-by-4}(b) retains a high degree of symmetry. When moving horizontally, it can be seen that there is are periodic patterns of nodes, with one column having 3 $Z$ nodes and 1 $X$ node and the next having 3 $X$ nodes and 1 $Z$ node. Also, the edges between the columns of nodes are local in one direction, as their length does not scale with the horizontal dimension of the toric code.

Using the algorithm on larger dimension $m$-by-$n$ toric codes shows that they have the same general structure as that of \Cref{4-by-4}(b). To construct it geometrically, first place an input $X$ node at the top of the diagram and an input $Z$ node at the bottom of the diagram. Then, on the unfolded toric grid, the first and second layers (out of $2n$ layers) of $m$ nodes each are designated as $Z$ nodes. Then, the rows alternate as rows of $X$ and $Z$ nodes until the very bottom two layers of the output nodes, which are rows of $X$ nodes. This is reflected in the 4-by-4 example.

\begin{figure}
    \centering
    \includegraphics[scale=0.4]{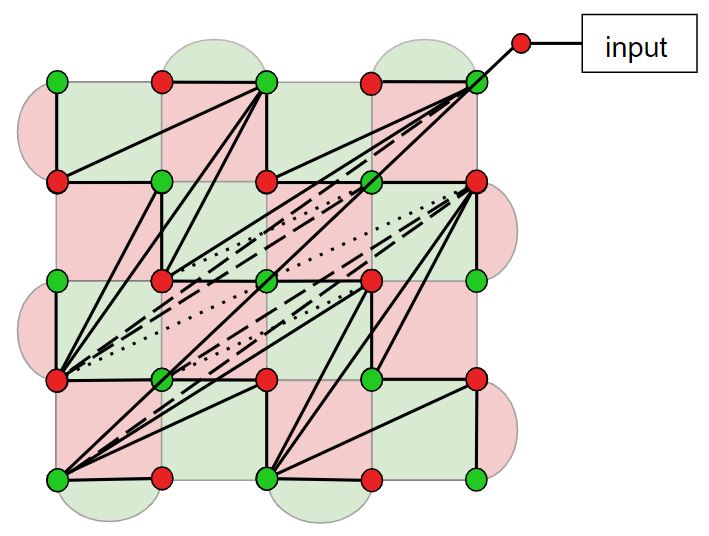}
    \caption{\justifying A 5-by-5 surface code in its ZX canonical form. Note its resemblance with the toric code diagram when tilted by $45^\circ$. The red input node connects to all 5 $Z$ nodes along the bottom-left to top-right diagonal. Some edges are dotted or dashed for clarity. The free output edges are not shown.}

    \label{surface 5-by-5}
    
\end{figure}

To draw out the edges, each of the top layer's $Z$ nodes has one edge to each of the $X$ nodes within its column. Furthermore, the second layer's $Z$ nodes have edges connecting them to each of the $X$ nodes in the neighboring columns, as well as one edge connecting it to the bottommost $X$ node in the same column. The next layer of $Z$ nodes (the fourth layer of $m$ nodes from the top) have edges to all the $X$ nodes in the neighboring columns that are in rows strictly below it, as well as one edge to the bottommost $X$ node in the same column. This pattern follows for the other $Z$ node layers. 

Since there is a symmetry between the $Z$ and $X$ nodes, we see the same arrow-shaped patterns extending upwards from layers of $X$ nodes.

We can convert this symmetrical form of the toric code into its KL form by changing the input $Z$ node into an input $X$ node and keeping the stabilizers the same. This gives Figure \ref{fig: 4-by-4 toric code KL form} as the KL form for the toric code, with the numbering reflecting the Hadamard rule.

We can also extend our results to the rotated surface codes, shown in Eq. (12) from \cite{kissinger2022phase}. This is a surface code defined such that each face shaded in green represents a $Z$ check while each face shaded in red represents an $X$ check. A $Z$ check consists of $Z$ gates on each of the nodes on the perimeter of the green face, and $X$ checks are similarly defined.


In \Cref{surface 5-by-5}, we show the result of simplifying the 5-by-5 surface code. If it is rotated by $45^\circ$ counterclockwise, it closely resembles the patterns of edges and colors seen in the general toric code. The neighboring diagonals (from bottom-left to top-right) of nodes of different colors connect in arrow-shaped patterns, just as in the toric code.

In general, the $(2k+1)$-by-$(2k+1)$ surface code can be made to have a similar structure as shown in \Cref{surface 5-by-5}.

\section{Prime Codes}\label{sec: prime codes}

When expressed in the ZX calculus, QECC's could have multiple ``separate" connected components. The tensor product of operations of the connected components is the operation of the entire ZX diagram. For example, any connected component with no free edges is a scalar. In this section, we consider codes with respect to their connected components.

We note the following statement about ZX diagrams with multiple connected components. 

\begin{proposition}
If a ZX diagram has multiple connected components that share no connections between each other, then these components are not entangled. Equivalently, if the components are entangled, then they share some connections.
\end{proposition}

We introduce the notion of \textit{prime} codes, defined as follows.

\begin{definition}
    A \textit{prime code diagram} is a ZX diagram in KLS form that cannot be expressed as a disconnected graph after a sequence of rewrite rules. 
\end{definition}

We show the following result about prime codes, which we give the name the Fundamental Theorem of Clifford Codes, or FTCC (alluding to the Fundamental Theorem of Arithmetic).


\begin{theorem}[Fundamental Theorem of Clifford Codes]

    Consider a Clifford code with an encoder-respecting form satisfying the constraint that the input-output adjacency matrix is full rank. Then, there exists a unique decomposition of the code into a product of prime codes (up to a permutation of input nodes).
\end{theorem}

We prove FTCC later in this section. To begin, we first prove the following lemma considering the properties of connected components of KLS forms.

\begin{lemma}
\label{lemma: KLS components are all prime}
    The connected components of a KLS form are all prime.
\end{lemma}

\begin{proof}
    Consider a connected component $\mathcal{D}$ in the KLS form of the code $\mathcal{C}$. Because the adjacency matrix between all inputs and outputs of $\mathcal{C}$ is full-rank and in RREF, the adjacency matrix $N$ between the inputs and outputs of $\mathcal{D}$ has pivot nodes for each of its input nodes. Also, $N$ is full rank, so $N$ must be in RREF. Furthermore, note that any row operations on $N$ preserve the input-pivot connections.

    Because the component $\mathcal{D}$ is in RREF (satisfying the RREF rule) and it is a subgraph of a KLS diagram (thus satisfying the Edge, Hadamard, and Clifford rules of KLS \cite{KLS}), $\mathcal{D}$ must also be in KLS form.

    Because the ZX calculus is complete for stabilizer quantum mechanics \cite{backens2014zx}, any two equivalent diagrams can be made equal through a sequence of basic rewrite rules (see Definition \ref{def: basic rewrite rules}). The basic rewrite rules of ZX calculus that can affect whether two nodes are in the same connected component
    are the state copy and Hopf rules (see Definition \ref{def: basic rewrite rules}). The state copy rule requires an internal node (a node without free edges) connected to exactly one node of the opposite color. None of the other rules are able to produce such an internal node, so the state copy rule cannot be applied. For example, the bialgebra rule can never cause an internal node to have only one edge. The Hopf rule requires two nodes sharing two edges. Because none of the other basic rewrite rules cause two nodes to share two edges, the Hopf rule cannot be applied either. 

    We now show no sequence of row operations on the adjacency matrix $N$ can turn the initial connected graph of $\mathcal{D}$ into disconnected components. For the sake of contradiction, suppose $\mathcal{D}$ is split into multiple connected components due to row operations. Then, row operations can turn each connected component's input-output adjacency matrix into RREF, so the adjacency matrix for $\mathcal{D}$ is disconnected and (after appropriate input permutations) in RREF. Since $N$ was initially connected and in RREF, and the RREF is unique, we reach a contradiction. Thus, no sequence of row operations on $N$ can turn $\mathcal{D}$ into disconnected components.

    Therefore, all the nodes in $\mathcal{D}$ are always in the same connected component. Since $\mathcal{D}$ is in KLS form and it cannot become disconnected, $\mathcal{D}$ is prime.

    Since $\mathcal{D}$ was chosen arbitrarily, this implies that the connected components of a KLS form are all prime.
\end{proof}

The encoders considered below will always be in encoder-respecting form. We now go through the process of converting an arbitrary encoder into its KLS form to show that the end result is the same had we converted two disconnected components into KLS (possibly having to rearrange the input nodes). The \textit{ZX-HK form} is an intermediate form of the code that satisfies the Edge and Hadamard rules (of KLS forms), but not necessarily the RREF or Clifford rules \cite{KLS}.

\begin{lemma}
\label{lemma: k(A) otimes K(B) sim K(C)}
    If $\mathcal{X}$ is an arbitrary Clifford encoder in ZX calculus, let $\text{KLS}(\mathcal{X})$ denote its KLS form in ZX calculus. If an encoder $\mathcal{C}$ can be decomposed into components $\mathcal{A}$ and $\mathcal{B}$, where there are no edges between the components, then $\text{KLS}(\mathcal{C})$ can be decomposed into components $\text{KLS}(\mathcal{A})$ and $\text{KLS}(\mathcal{B})$, again with no edges between the components.
\end{lemma}

\begin{proof}
    The conversion from an arbitrary encoder-respecting form to the KLS form involves only local complementation and and row operations on the input-output adjacency matrix \cite{KLS}.

    Transforming a code into its ZX-HK form relies only on local complementation (see section III.C from \cite{hu2022improved}). Because local complementation only toggles edges of the neighbors of a node, it cannot affect whether two nodes are in the same connected component. Thus, if the ZX-HK form of the encoders are $\mathcal{A'},\mathcal{B'},$ and $\mathcal{C'}$, respectively, we find that $\mathcal{C'}$ can be decomposed into components $\mathcal{A'}$ and $\mathcal{B'}$, with no edges between the components, since local complementation happens entirely within the components.

    The next step is to transform the ZX-HK form's adjacency matrix between the input and output nodes into RREF. When transforming $\mathcal{C'}$ into RREF, note that we can first transform the separate components $\mathcal{A'}$ and $\mathcal{B'}$ into RREF. Each input node of $\mathcal{C'}$ has a corresponding pivot (output) node. Then, there exists a permutation of the inputs that will result in $\mathcal{C'}$ being in RREF, if we order the inputs to match the order of the pivots. Permuting can be done by row operations. Since the RREF of the adjacency matrix is unique, this is the desired RREF for $\mathcal{C'}$. Note that the resulting diagram $\mathcal{C''}$ can still be decomposed into components $\mathcal{A'}$ and $\mathcal{B'}$, with no edges between the components. This is because the order of the inputs within each component is unchanged.

    The next step in converting into KLS form involves local complementing at the inputs of pivot nodes. Then, if there are edges between pivot nodes in the resulting form, we perform local edge complementations on the corresponding input nodes, followed by a permutation of the same input nodes. The former does not affect whether two nodes are in the same connected component. The latter also does not alter this property because local complementation does not affect whether two nodes are in the same connected component and the input-output adjacency matrix is preserved in RREF after the permutation. This completes the transformation into the KLS form.

    Thus, we have shown that $\text{KLS}(\mathcal{C})$ can be decomposed into components $\text{KLS}(\mathcal{A})$ and $\text{KLS}(\mathcal{B})$, which share no edges.
\end{proof}

Note that the output nodes of $\mathcal{A}$ and $\mathcal{B}$ remain the same after transforming them into KLS form. Then, if we order the inputs of $\text{KLS}(\mathcal{A})$ and $\text{KLS}(\mathcal{B})$ so that they match with the ordering of the pivot output nodes, the resulting diagram is the same as $\text{KLS}(\mathcal{C})$. 

Now, we show why equivalent Clifford encoders share the same set of prime components.

\begin{theorem} 
Given two Clifford encoders with the same codespace, each entirely made up of prime components, they have equivalent ZX diagrams up to a permutation of the inputs.

\label{thm: two equivalent encoders entirely made of primes}
\end{theorem}

\begin{proof}

First, we note that two encoders with the same codespace must have the same KLS form. Therefore, the statement is equivalent to showing that, given a Clifford encoder $\mathcal{C}$ entirely made up of prime components, its ZX diagram is equivalent to the KLS form up to a permutation of inputs.

We now induct on the number of prime components in $\mathcal{C}$. If $\mathcal{C}$ has exactly one prime component, then its entire diagram must be in KLS form. Since the KLS form for $\mathcal{C}$ is unique, this means $\mathcal{C}$ and its KLS form are identical, and thus equivalent up to a permutation of inputs.

Now, suppose all encoders with $i \leq n$ prime components are equivalent to their KLS form up to a permutation of inputs. We want to show that this implies any encoder $\mathcal{C}$ with $n+1$ prime components is equivalent to its KLS form up to a permutation of inputs.

Denote the connected component of $\mathcal{C}$ that has output node 1 as $\mathcal{A}$. Then, the remaining part of $\mathcal{C}$ that shares no edges with $\mathcal{A}$ can be denoted as $\mathcal{B}$. Note that $\mathcal{A}$ is already in KLS form, since $\mathcal{A}$ must be prime. By the inductive hypothesis, we have that $\mathcal{B}$ is equivalent to $\text{KLS}(\mathcal{B})$ up to input permutations, so, after an appropriate sequence of rewrite rules, $\mathcal{C}$ can be decomposed into $\mathcal{A}$ and $\text{KLS}(\mathcal{B})$. 

 By Lemma \ref{lemma: k(A) otimes K(B) sim K(C)}, we have that $\text{KLS}(\mathcal{C})$ can be decomposed into $\text{KLS}(\mathcal{A})$, which is just $\mathcal{A}$, and $\text{KLS}(\mathcal{B})$. Evidently, this means that $\mathcal{C}$ is equivalent to $\text{KLS}(\mathcal{C})$ up to input permutations.

\end{proof}

Now, we can prove the Fundamental Theorem of Clifford Codes.

\begin{proof}[Proof of Theorem 5.3 (FTCC)]
    We begin with a Clifford code with an encoder-respecting form satisfying the constraint that the input-output adjacency matrix is full rank. Then, after converting it into its KLS form, Lemma \ref{lemma: KLS components are all prime} gives that all its connected components are prime. Therefore, we have constructed a decomposition of the code into primes.

    Now, we show why this construction is unique. For the sake of contradiction, suppose there is another different decomposition of the code into primes. Then, by Theorem \ref{thm: two equivalent encoders entirely made of primes}, these two ZX diagrams are equivalent up to a permutation of inputs. Note that we can rearrange the inputs at will using row operations without changing the structure of any of the primes, so this second decomposition must be the same as the one constructed above.

    Thus, the decomposition exists and it is unique.
\end{proof}

\section{Another Definition of Equivalence}\label{sec 4}

Previous works have examined the equivalence classes of graphs under local complementation \cite{adcock2020mapping, bahramgiri2007enumerating, bouchet1993recognizing}. In Clifford codes, the presence of designated input and output vertices makes the definition of equivalence more exotic.

In the following sections, we consider only the ZX diagrams for Clifford codes that have no local operations on the free output edges.

\begin{figure*}

\begin{minipage}{.2\textwidth}

Transforming the adjacency matrix:
$$\begin{pmatrix}
     0 & 0 & 1 & 0 & 1 & 0 & 1 \\
    0 & 0 & 0 & 1 & 0 & 0 & 0 \\
    1 & 0 & 0 & 0 & 0 &0  & 0 \\
    0 & 1 & 0 & 0  & 0 & 0 & 1\\
    1 & 0 & 0 & 0 & 0 & 0 & 0\\
    0 & 0 & 0 & 0 & 0 & 0 & 0\\
    1 & 0 & 0 & 1 &0 & 0 & 0\\
    
\end{pmatrix}$$
  
\end{minipage}
\begin{minipage}{.3\textwidth}
\vspace{0.8 cm}
  \includegraphics[scale=0.20]{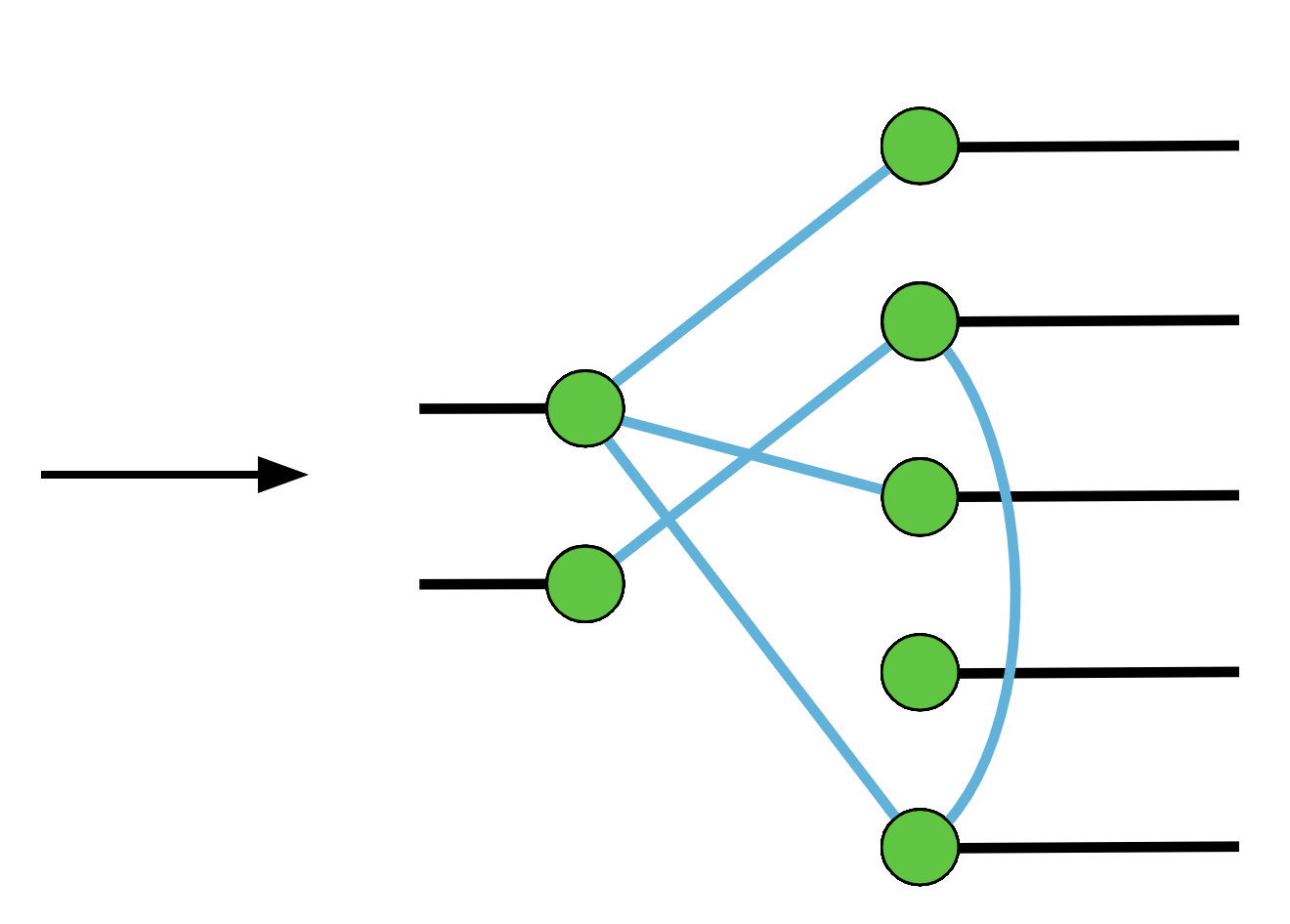}
\end{minipage}

\caption{\justifying   In the adjacency matrix, the first row corresponds to the first input, the second row corresponds to the second input, and so on. After the inputs, the following row corresponds to the first output, and the other outputs follow. In the ZX diagram on the right, all edges shown are internal edges. For clarity, the Hadamard gates are not shown.}

\label{transforming example}
    
\end{figure*}

\begin{definition}
\label{defn: locally equivalent}
    Two ZX diagrams are \textit{locally equivalent} if and only if one can be converted to the other through a sequence of local complementations and local operations on the free output edges.
\end{definition}

We now provide another definition of equivalence that turns the focus to the topological structure of the code, letting us ignore specific orderings of output nodes and the local operations attached on the free output edges.

\begin{definition}
\label{defn:equivalent clifford codes}
    Two Clifford codes $\mathcal{C}_1$ and $\mathcal{C}_2$ are \textit{equivalent} if and only if the ZX diagrams are locally equivalent or locally equivalent after some permutation of the output nodes and/or applications of any unitary operators on the inputs.
\end{definition}


We now list five different operations which keep encoder graphs equivalent.

\begin{conjecture}
\label{operations}
The ZX diagrams for two Clifford codes $\mathcal{C}_1$ and $\mathcal{C}_2$ are \textit{equivalent} if and only if the diagram for one of the codes can be reached from the other after a sequence of operations consisting only of the following:
\begin{enumerate}[noitemsep]
\item Local complementing about any vertex of the graph.
\item Permuting the output vertices.
\item Permuting the input vertices.
\item Performing row operations on the adjacency matrix of input to output edges.
\item Removing an input-input edge.
\item Applying local operations on the output edges.
\end{enumerate}
\end{conjecture}

All of the operations in Conjecture~\ref{operations} are reversible, so, if code $\mathcal{C}_1$ can be made equivalent to $\mathcal{C}_2$, the reverse is also true.

Operation 1, local complementation as in Definition~\ref{localcomplementation}, is included to account for equivalence of encoder graphs based on their entanglement \cite{adcock2020mapping}.

Two encoder diagrams should also be equivalent if the information they produce can be ordered differently to become the same. In this way, operations 2 and 3 reflect this, since connections among the vertices of the graph remain the same and these operations only change the order in which the information is inputted or outputted.

Operation 4 consists of adding rows of the adjacency matrix between input and output nodes in modulo 2. Considering some input node, any stabilizer of the code must have $X$ or $Y$ gates on an even number of the input node's neighboring output nodes. When another row of the input-output adjacency matrix is added to the row for this input node, any stabilizer still has $X$ and $Y$ gates on an even number of the input node's neighboring output nodes, since we effectively add two even numbers together modulo 2. Row operations thus preserve the stabilizers of the code, giving a possible operation between equivalent encoders.

Operation 5 takes away a unitary operation from the input vertices, which is allowed by Definition \ref{defn:equivalent clifford codes}. Lastly, operation 6 preserves equivalence since all local operations can be removed by multiplying by their corresponding conjugate, which is allowed by Definition \ref{defn: locally equivalent}.


Note that this definition of equivalence does not allow two encoders to be in the same equivalence class if they only differ by an extra output (which is not connected to anything else). That is, if the two encoders differ by a quantum state, this definition of equivalence marks them as different. Therefore, this implies we focus on section Y of \Cref{venn-diagram}, instead of section X.

As an example of these extra outputs/states, see \Cref{transforming example}. The output vertex labeled 4 is not connected to any input or output. It does not provide any more encoding of information from the inputs than if it was not present. For this reason, section X of \Cref{venn-diagram} is more useful for practical purposes.

There are some simplifications that can be made on the set of encoders we consider by using the above operations. This is so that we consider only encoder graphs that could possibly be non-equivalent.

Operations 3 and 4 of Conjecture~\ref{operations} allow the input-to-output portion of the encoder diagram to be expressed in RREF. All encoder diagrams considered from here on are expressed in RREF, as in the RREF rule from \cite{KLS}. Note that the inputs have corresponding pivot output nodes in the RREF.

Continuing from the RREF of the encoder graph, operation 2 from Conjecture~\ref{operations} can be used to move the pivot nodes to have lower-numbered indices than all the other outputs. In this way, the first output node can be made into the pivot node corresponding to the first input node, the second output node can be made into the pivot node corresponding to the second input node, and so on. Thus, these $k$ pivot nodes are fixed among the top of the output nodes. For brevity, the other $n-k$ non-pivot output nodes are called \textit{free} output nodes.

In Conjecture~\ref{operations}, no operation was included that affected local Clifford gates at the nodes. Therefore, this definition of equivalence neglects the presence of phase changing gates at vertices of the encoder's graph. This is because local Clifford gates change the qubits using a unitary operation but does not contribute to changes in entanglement of the qubits in any way. Therefore, for our purposes, we remove all local Clifford gates present at the nodes of the ZX diagram for the encoder.

A further simplification, carried over from \cite{KLS}, is that graphs with pivot-pivot edges are omitted, since they can always be transformed into a graph without pivot-pivot edges using a sequence of local complementations.

We can group these simplifications into the following result.

\begin{claim}
\label{claim: simplifications}
    To find the distinct encoders in an equivalence class of the set of Clifford codes, it is sufficient to find the distinct encoders in the equivalence class satisfying the following constraints:
    \begin{itemize}
        \item The input-output adjacency matrix is in RREF.
        \item The pivot nodes are the lowest-numbered output nodes, and they are ordered to match the order of the input nodes.
        \item Local Clifford gates on the free output edges are all removed, leaving only output $Z$ nodes with phase 0.
        \item Diagrams with pivot-pivot edges are not included.
    \end{itemize}

\end{claim}

As an example of the simplifications on the diagrams, as well as how the adjacency matrices transform into encoders, see \Cref{transforming example}.

\begin{figure*}[t!]
    \centering
    \begin{subfigure}[t]{0.5\textwidth}
        \centering
         \bgroup
\def\arraystretch{1.5}

\setlength{\tabcolsep}{0.5em} 
    \begin{tabular}{c|c|c|c|c}
        $n = 1$ & $n = 2$ & $n = 3$ &  $n = 4$ & $n = 5$ \\
        \hline
        1 & 1 & 2 & 6 & 17 \\
    \end{tabular}

    \egroup
    \vspace{0.04 cm}
        \caption{\justifying Number of equivalence classes for $[n,1]$ codes.}
    \end{subfigure}%
    ~ 
    \begin{subfigure}[t]{0.5\textwidth}
        \centering
        \bgroup
\def\arraystretch{1.5}

\setlength{\tabcolsep}{0.5em} 
    \begin{tabular}{c|c}
        Rep: & \includegraphics[scale=0.09]{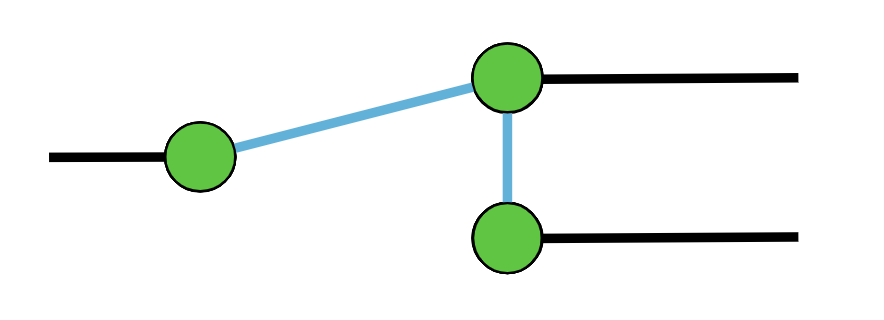} \\
        \hline
         Size: & 3\\
    \end{tabular}
\egroup
        \caption{\justifying $[2,1]$ codes equivalence class reps. and sizes.}
    \end{subfigure}

\begin{center}

\bgroup
\def\arraystretch{1.5}

\setlength{\tabcolsep}{0.5em}

\begin{tabular}{c|c|c}

Rep: & \includegraphics[scale=0.08]{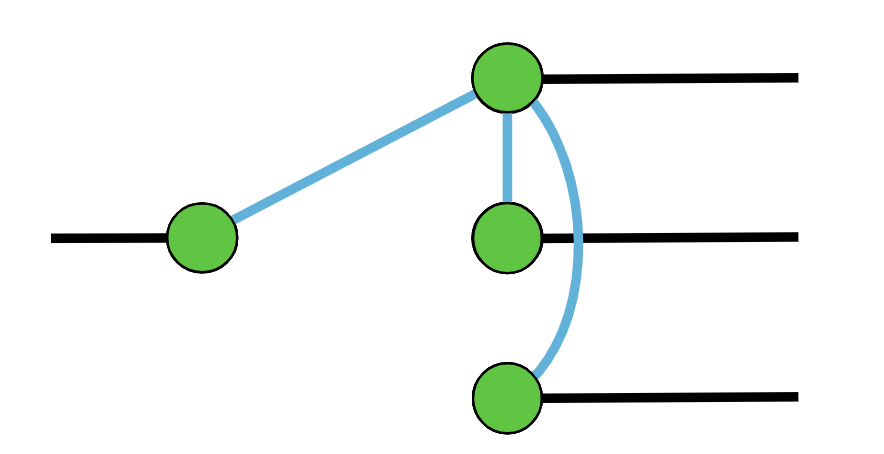} & \includegraphics[scale=0.08]{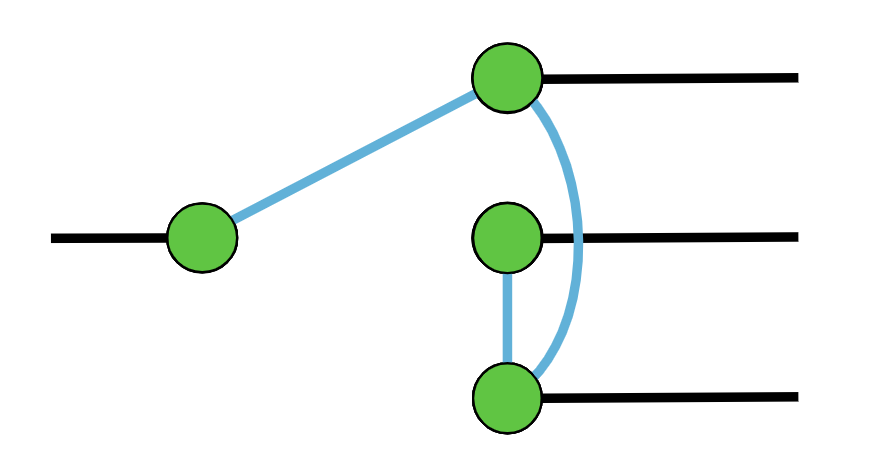}\\
\hline
Size: & 3&21 \\
\end{tabular}
\egroup

\bigskip

{\small (c) $[3,1]$ codes equivalence classes showing the size of the class underneath a representative.}

\bigskip

\bgroup
\def\arraystretch{1.5}

\setlength{\tabcolsep}{0.5em}

\begin{tabular}{c|c|c|c|c|c|c}

Rep: & \includegraphics[scale=0.08]{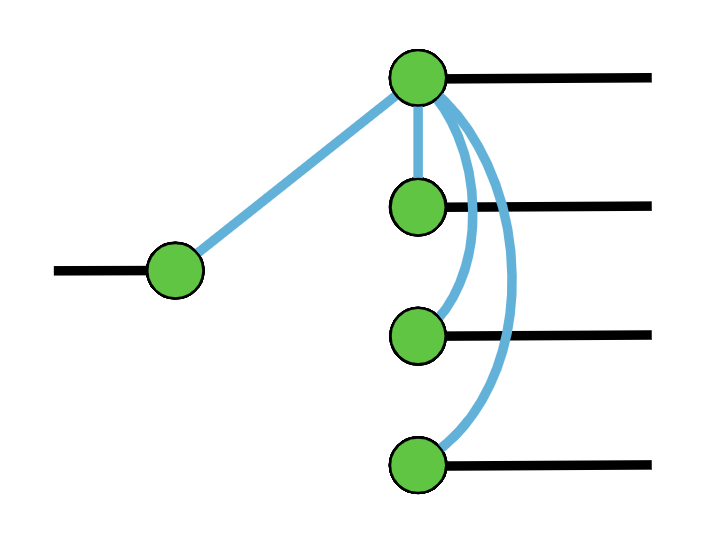}& \includegraphics[scale=0.08]{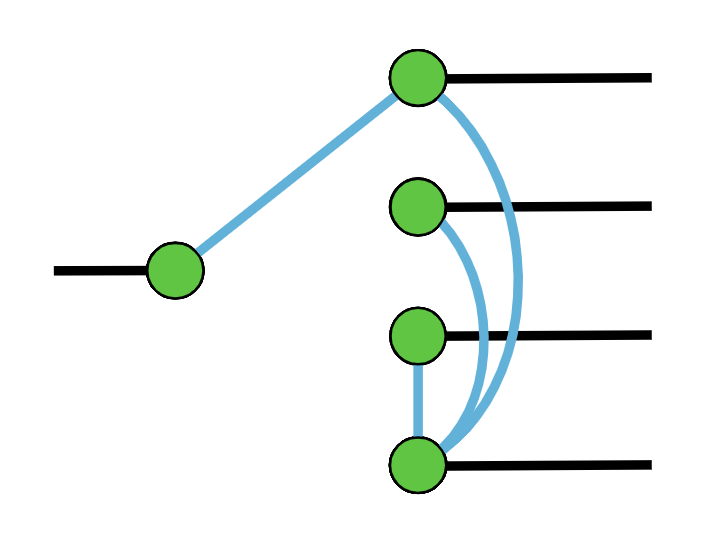} &  \includegraphics[scale=0.08]{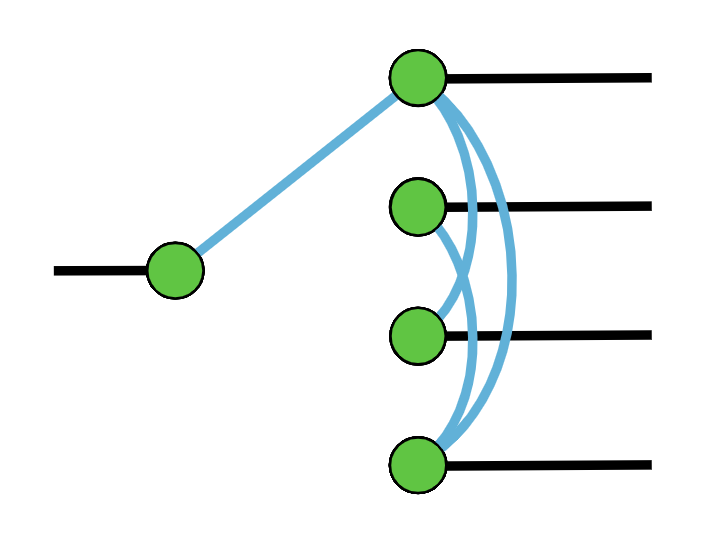}& \includegraphics[scale=0.08]{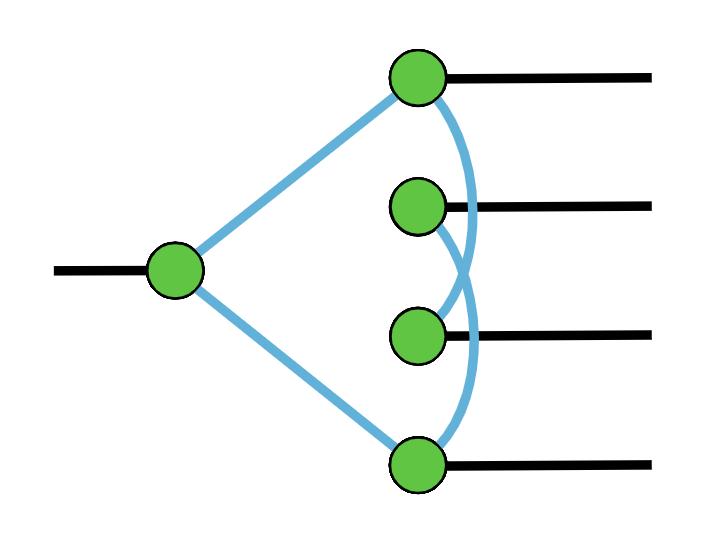} &  \includegraphics[scale=0.08]{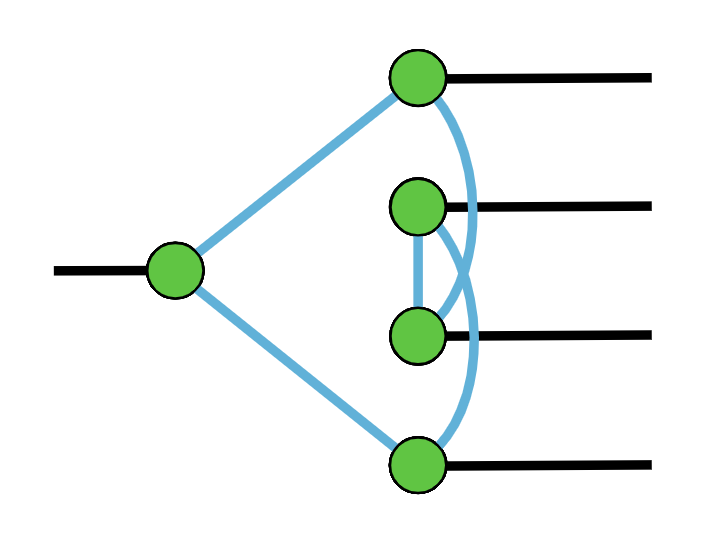}& \includegraphics[scale=0.08]{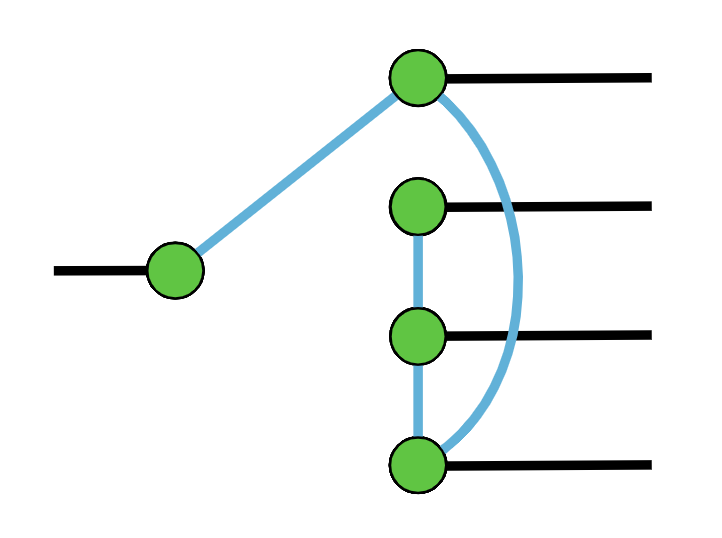}\\
\hline
Size: & 3&30 & 45 & 54 & 84 & 198 \\
\end{tabular}
\egroup

\bigskip

{\small (d) $[4,1]$ codes equivalence classes showing the size of the class underneath a representative.}

\bigskip    

\bigskip

\bgroup
\def\arraystretch{1.5}

\setlength{\tabcolsep}{0.5em}

\begin{tabular}{c|c|c|c|c|c}

Rep: & \includegraphics[scale=0.08]{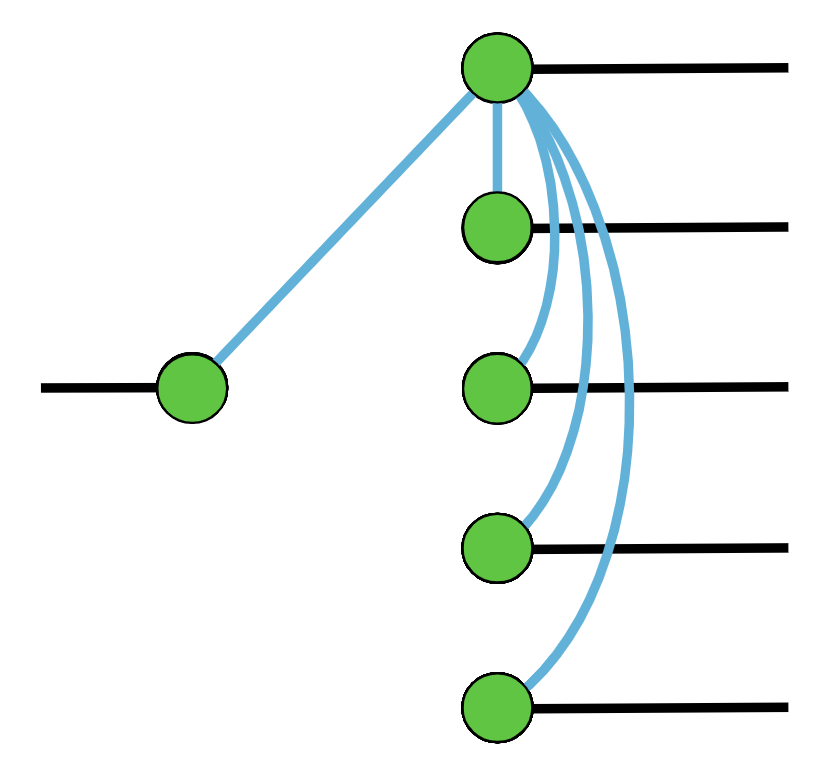}& \includegraphics[scale=0.08]{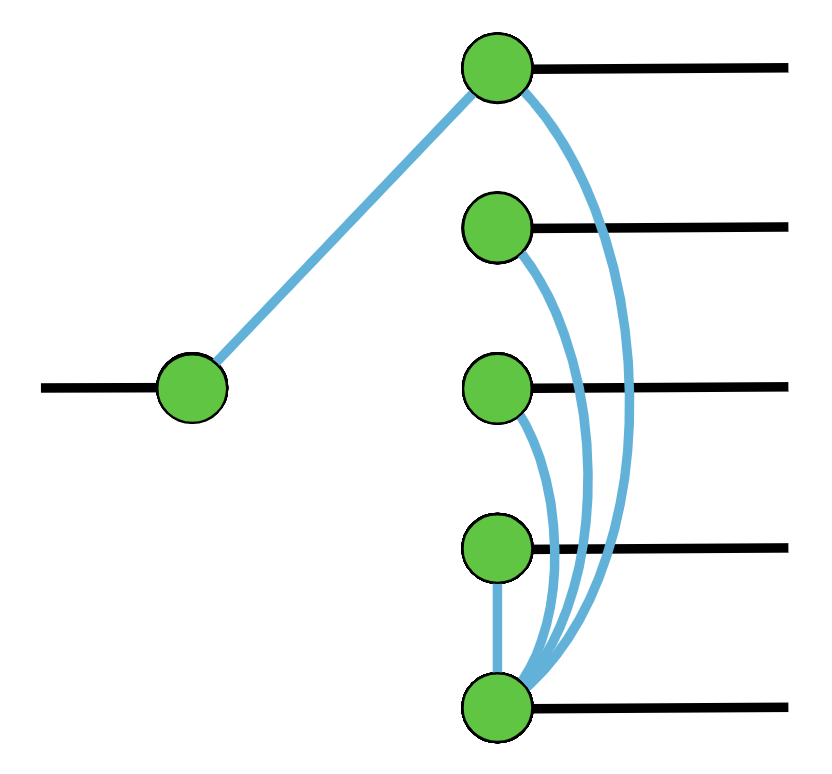} &  \includegraphics[scale=0.08]{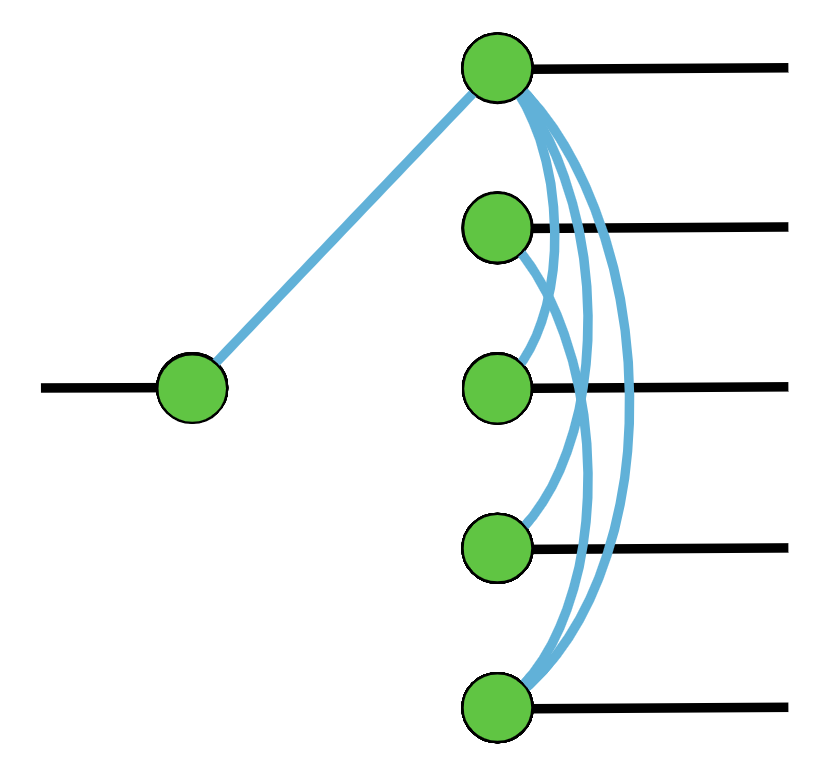}& \includegraphics[scale=0.08]{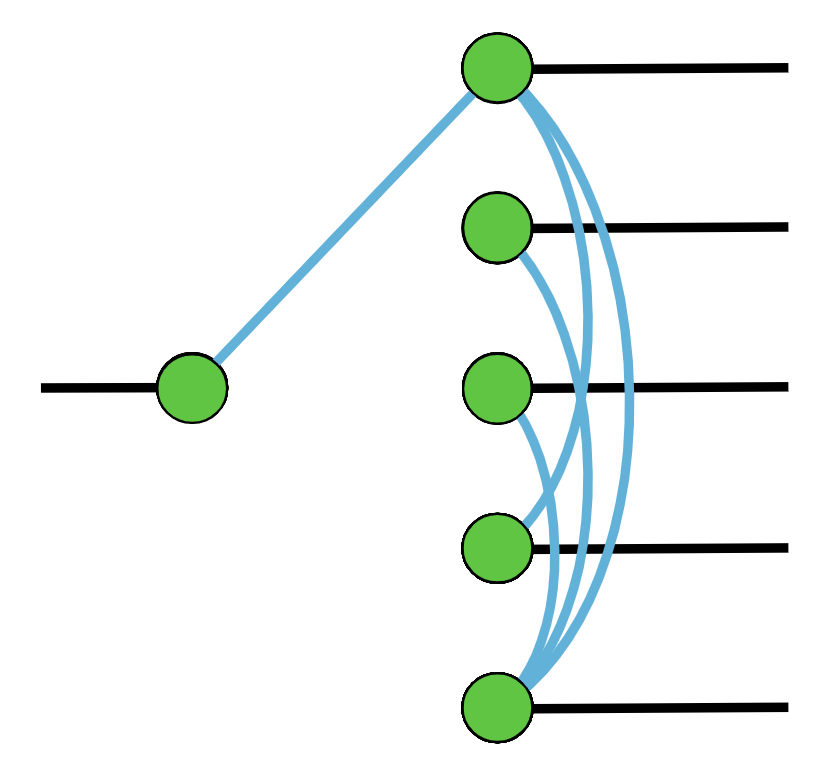} &  \includegraphics[scale=0.08]{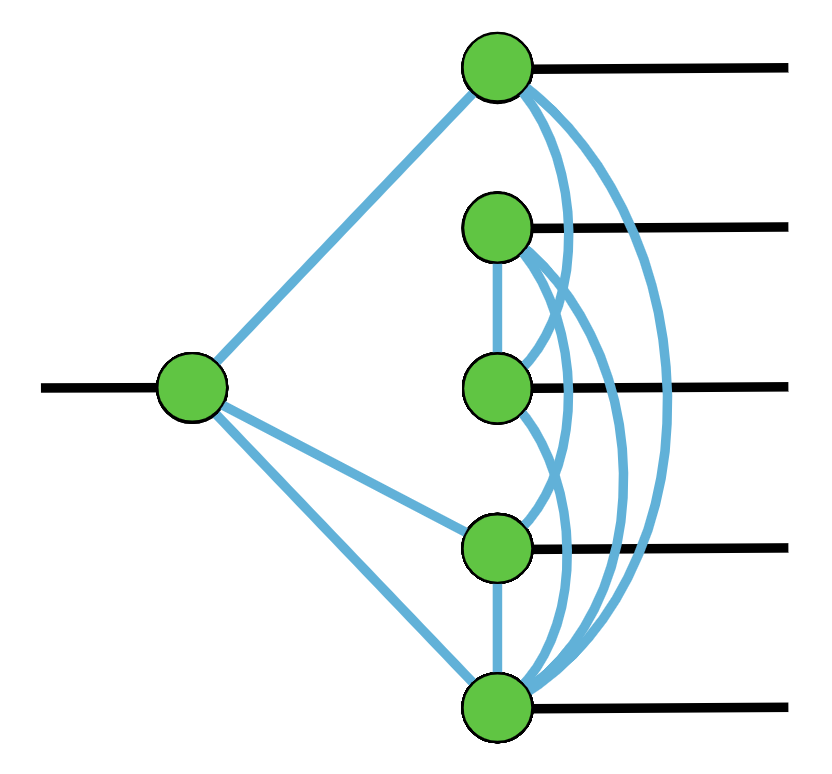}\\
\hline
Size: & 3&39 & 78 & 84 & 84 \\

\hhline{=|=|=|=|=|=}

\includegraphics[scale=0.08]{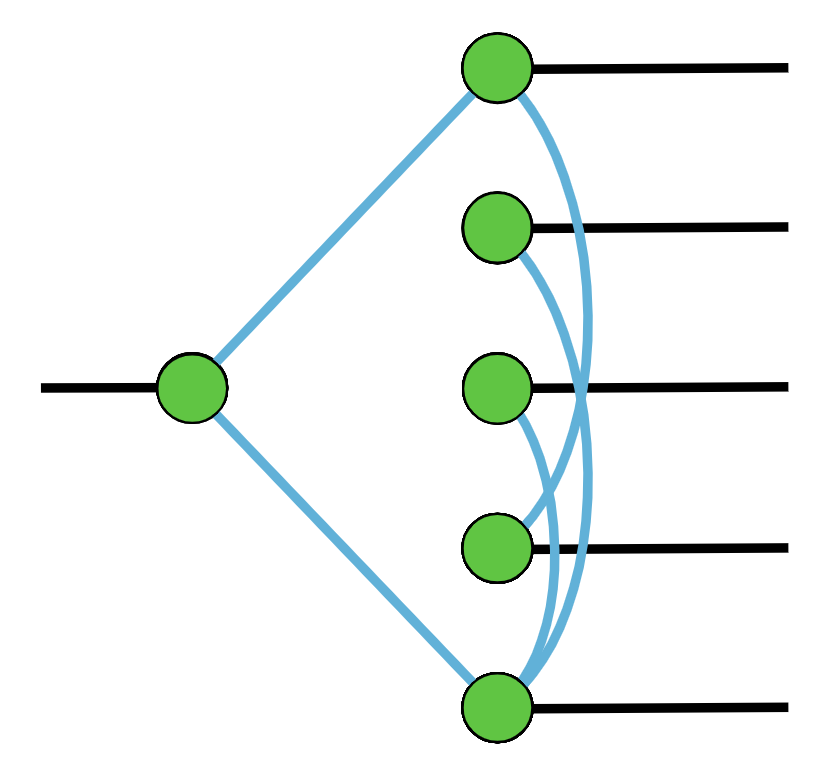}& \includegraphics[scale=0.08]{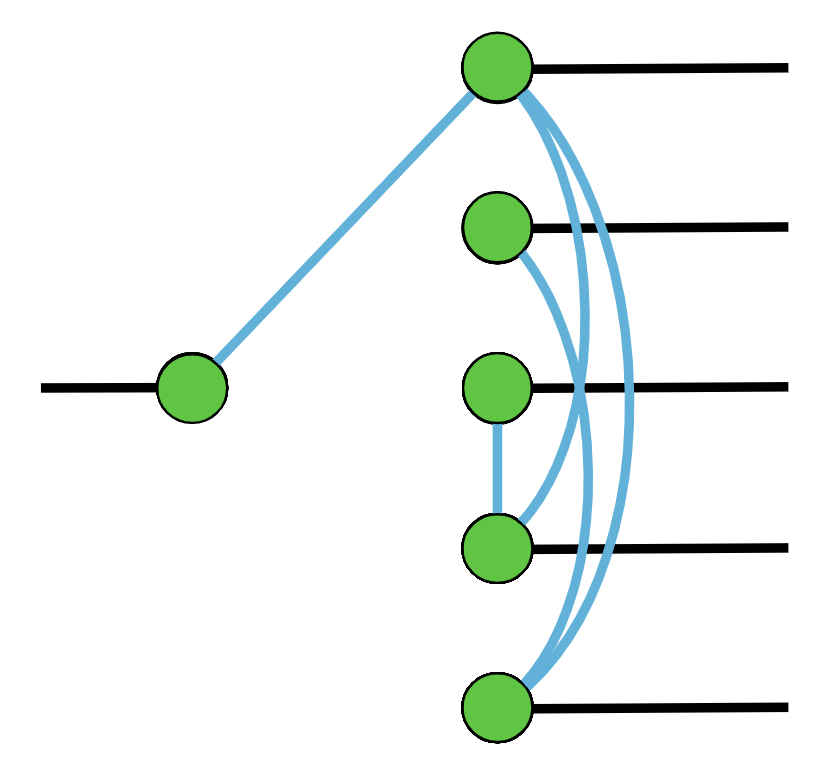} &  \includegraphics[scale=0.08]{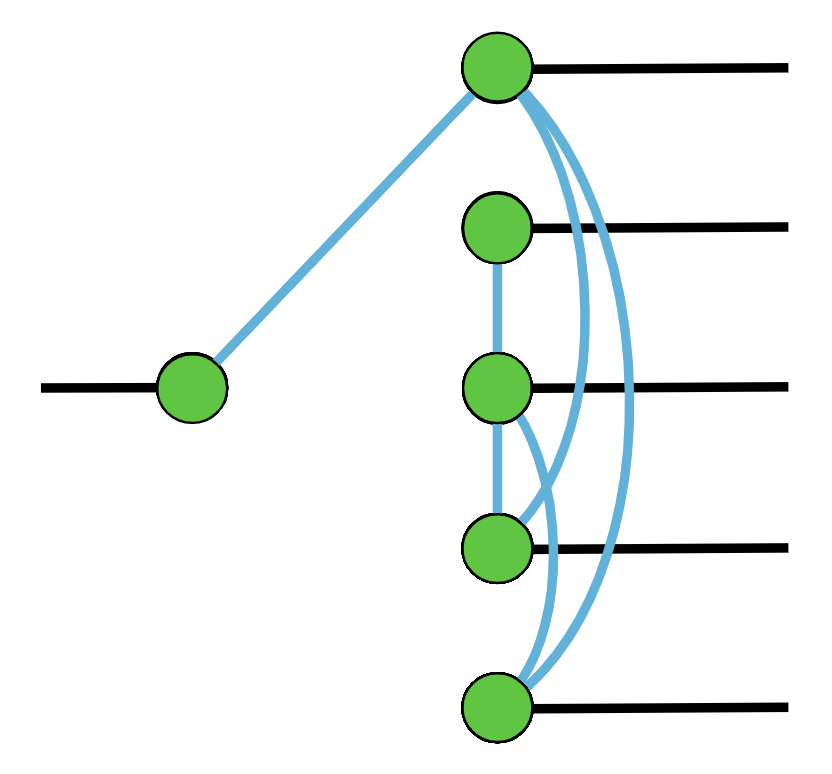}& \includegraphics[scale=0.08]{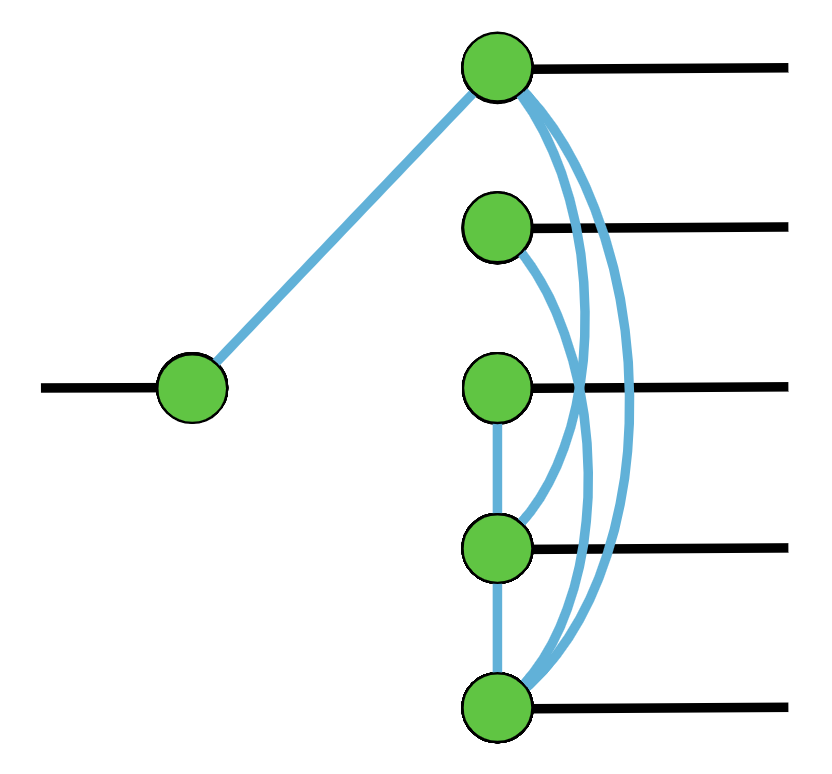} &  \includegraphics[scale=0.08]{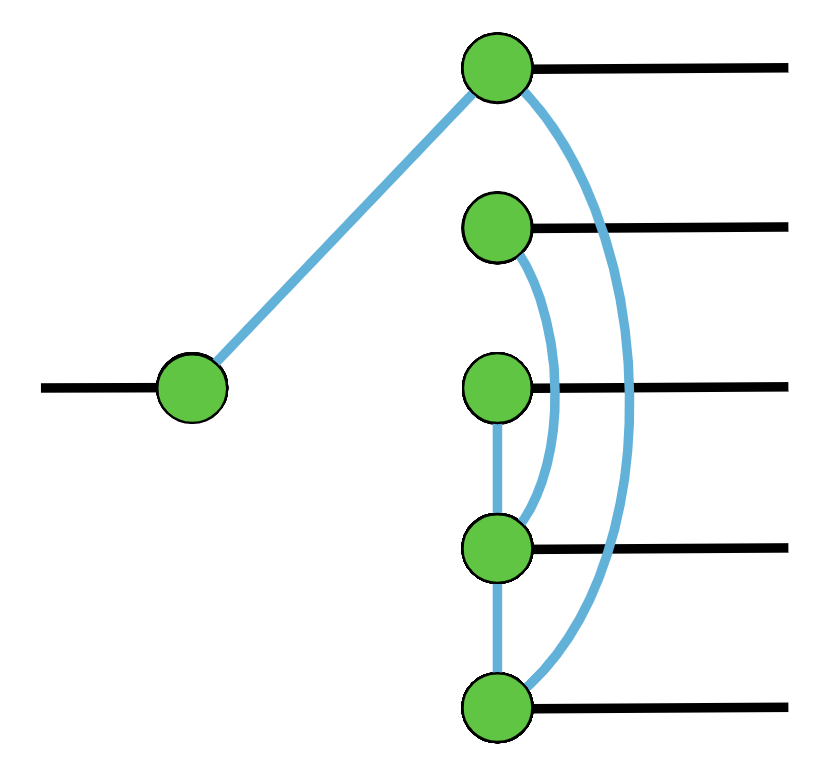}& \includegraphics[scale=0.08]{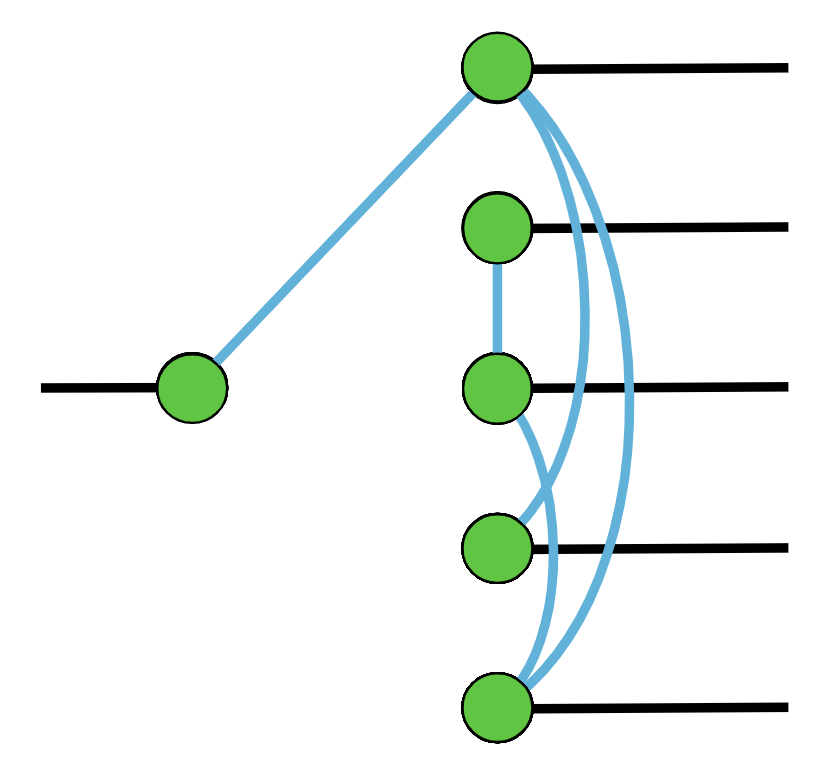}\\
\hline
204 & 297&306 & 315 & 360 & 540 \\

\hhline{=|=|=|=|=|=}

\includegraphics[scale=0.08]{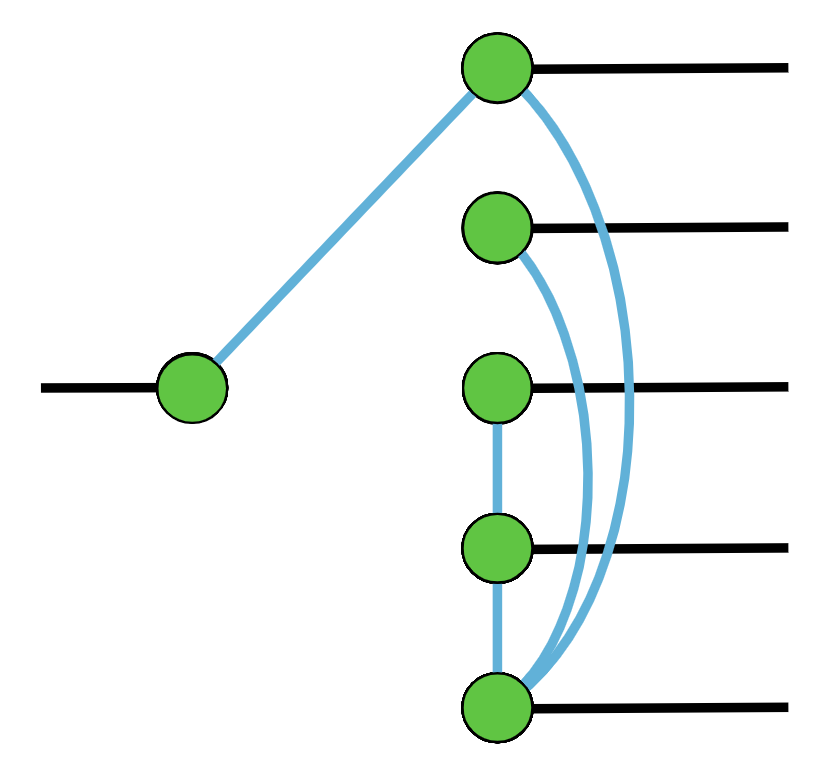}& \includegraphics[scale=0.08]{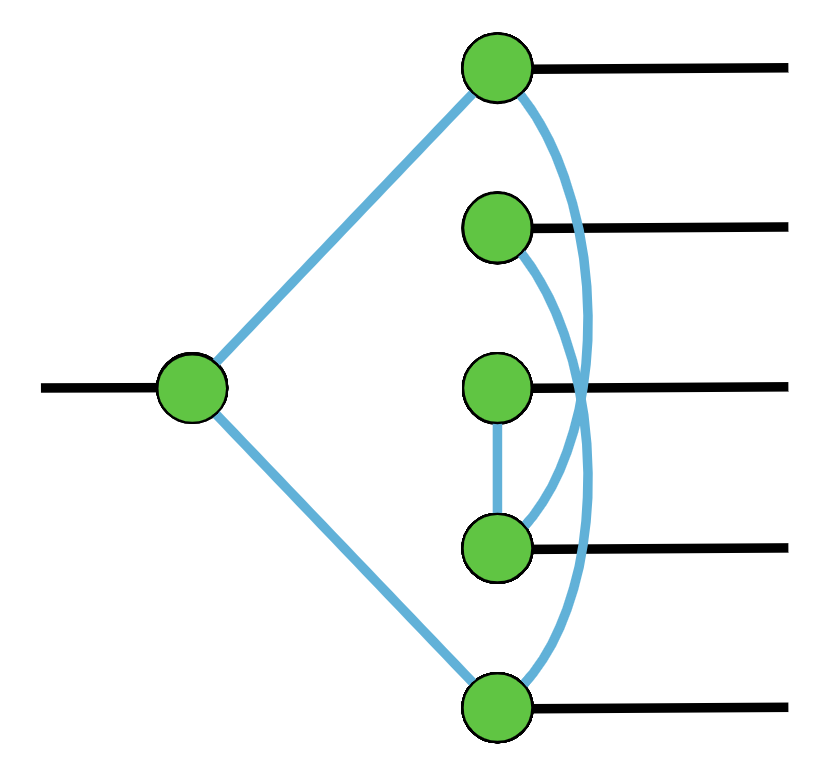} &  \includegraphics[scale=0.08]{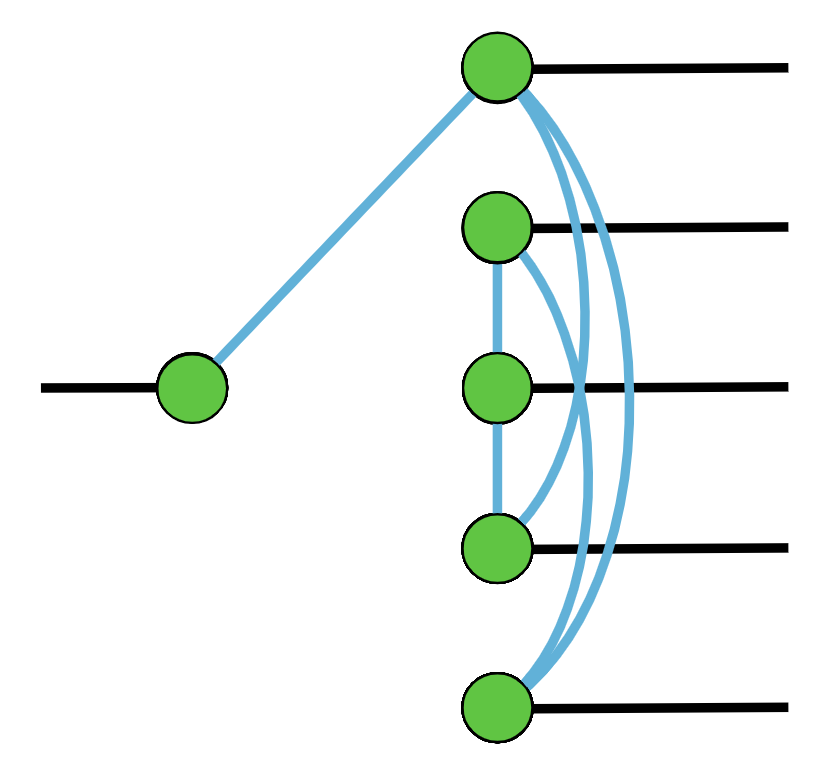}& \includegraphics[scale=0.08]{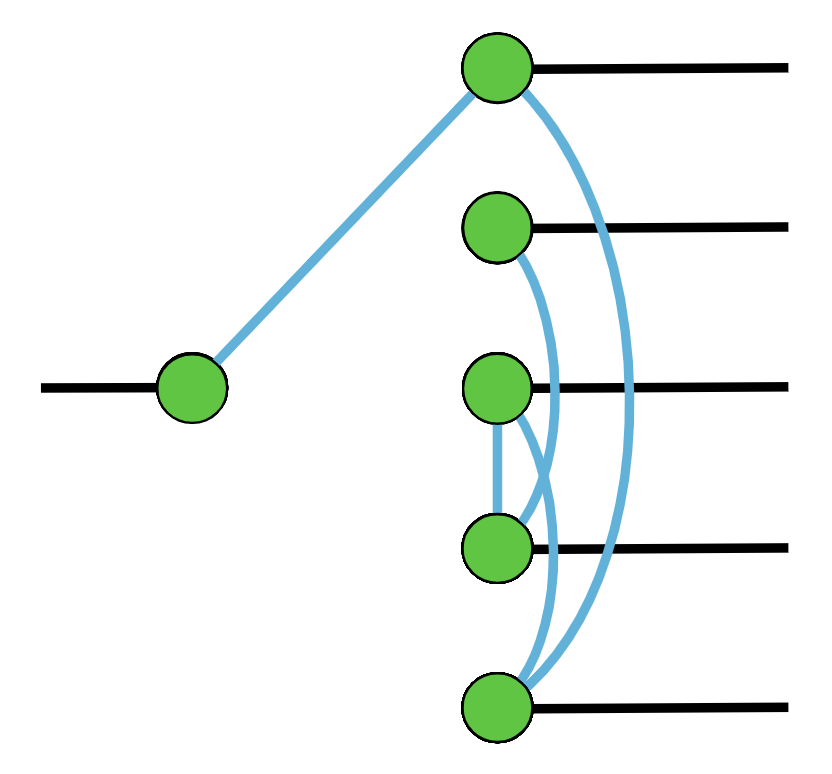} &  \includegraphics[scale=0.08]{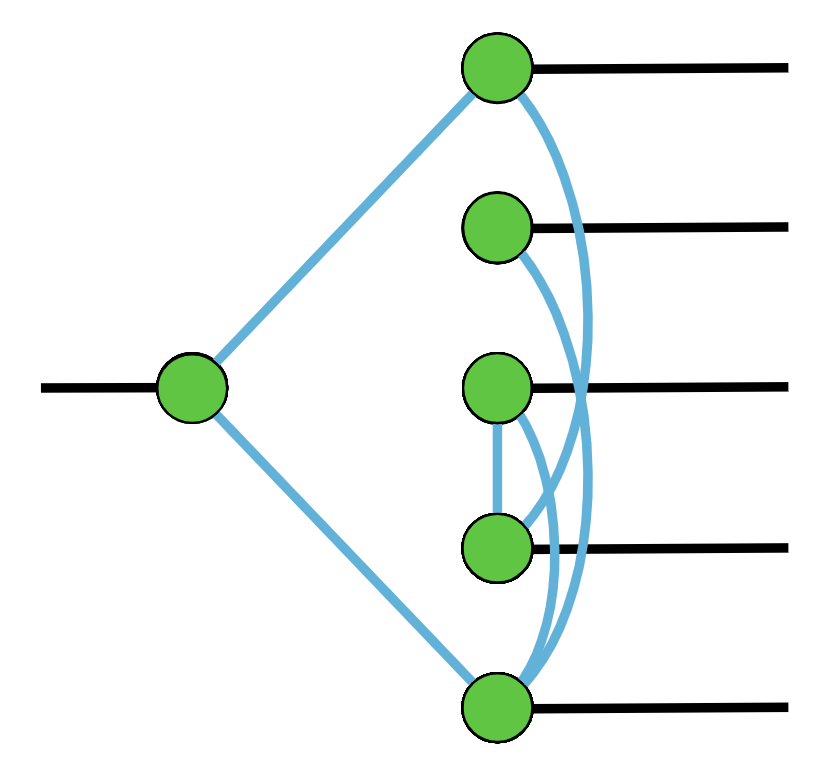}& \includegraphics[scale=0.08]{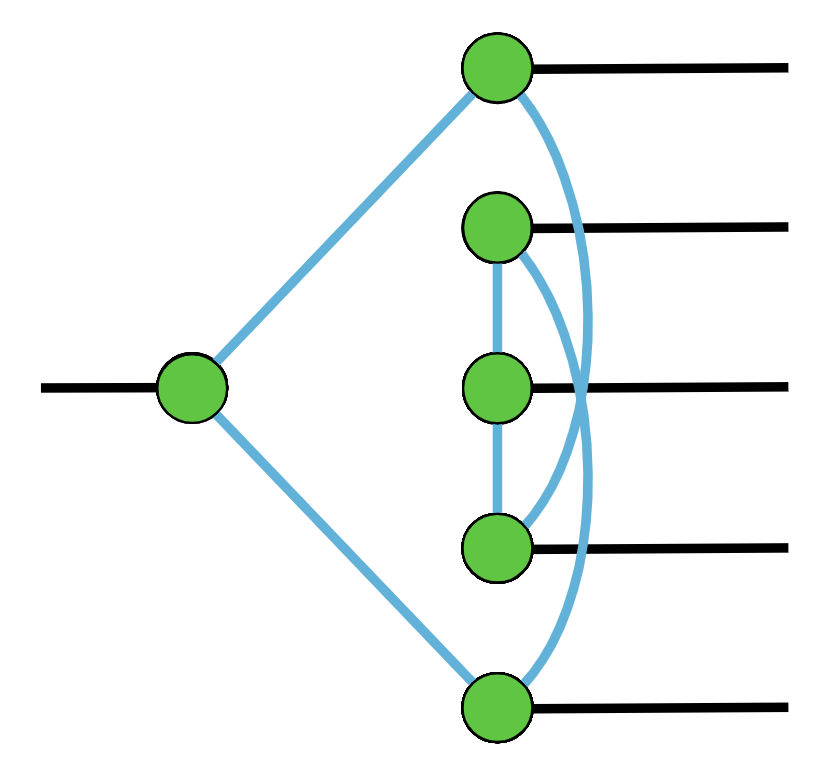}\\
\hline
558 & 1332&1404 & 2376 & 3024 & 3276 \\
\end{tabular}
\egroup

\bigskip

{\small (e) $[5,1]$ codes equivalence classes showing the size of the class underneath a representative.}

\bigskip

\end{center}
    
    \caption{\justifying (a) shows the number of equivalence classes for $[n,1]$ encoder graphs. (b-e) show an element of the equivalence classes to denote the representative of the class and gives the size of the class. We only consider classes in which every graph is prime.}

    \label{[n,1] codes big tables}
\end{figure*}

\begin{figure*}[t!]
    \centering
    \begin{subfigure}[t]{0.5\textwidth}
        \centering
         \bgroup
\def\arraystretch{1.5}

\setlength{\tabcolsep}{0.5em} 
    \begin{tabular}{c|c|c|c}
        $n = 2$ & $n = 3$ & $n = 4$ &  $n = 5$ \\
        \hline
        0 & 1 & 4 & 18 \\
    \end{tabular}

    \egroup
    \vspace{0.23 cm}
        \caption{\justifying Number of equivalence classes for $[n,2]$ codes.}
    \end{subfigure}%
    ~ 
    \begin{subfigure}[t]{0.5\textwidth}
        \centering
        \bgroup
\def\arraystretch{1.5}

\setlength{\tabcolsep}{0.5em} 
    \begin{tabular}{c|c}
        Rep: &\includegraphics[scale=0.08]{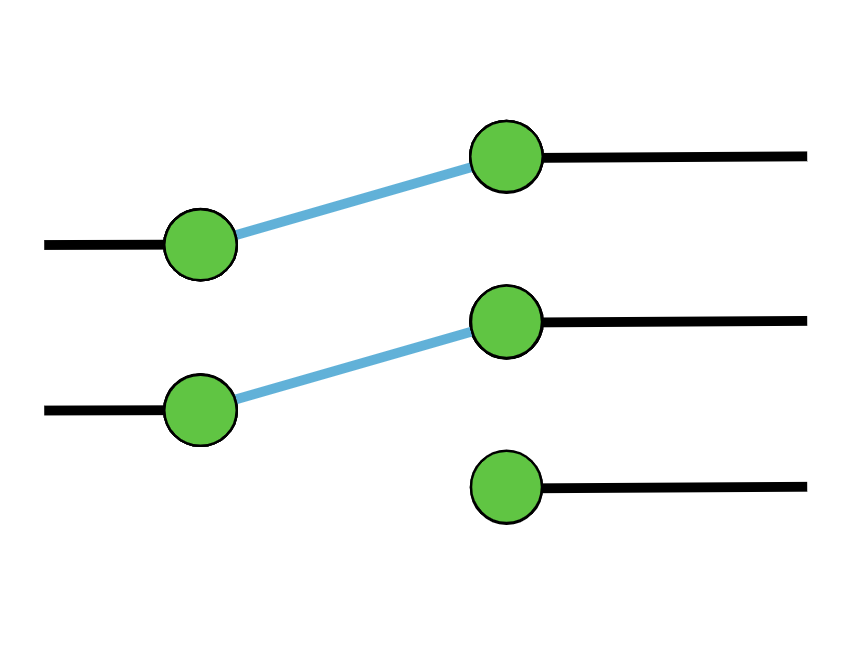} \\
        \hline
         Size: & 9\\
    \end{tabular}
\egroup
        \caption{\justifying $[3,2]$ codes equivalence class reps. and sizes.}
    \end{subfigure}

\begin{center}

\bgroup
\def\arraystretch{1.5}

\setlength{\tabcolsep}{0.5em}

\begin{tabular}{c|c|c|c|c}

Rep: & \includegraphics[scale=0.08]{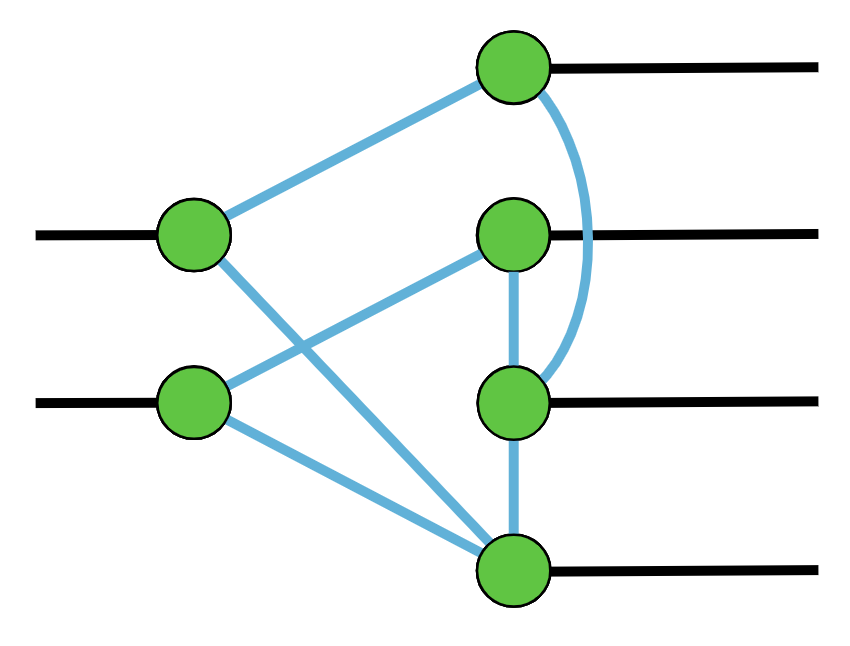}& \includegraphics[scale=0.08]{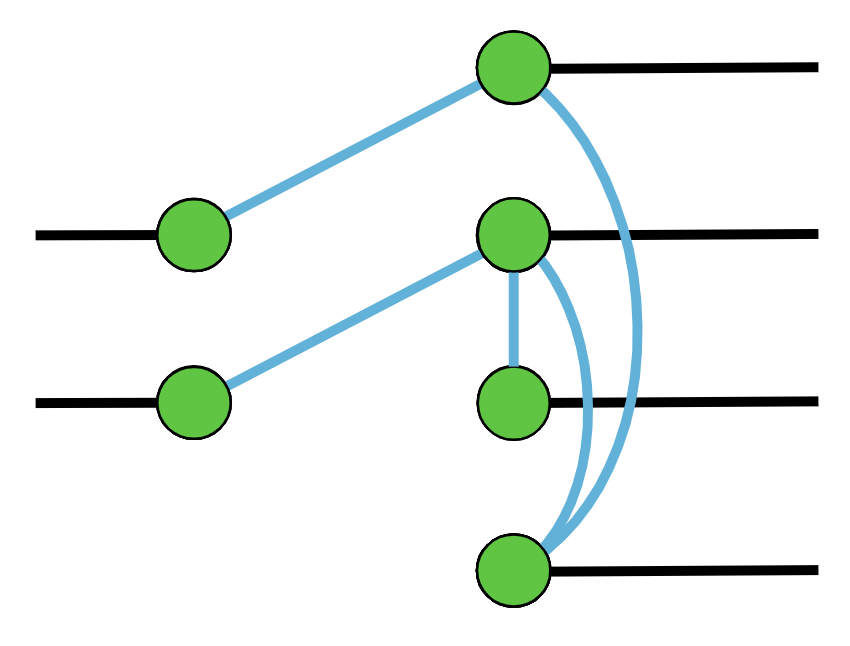}& \includegraphics[scale=0.08]{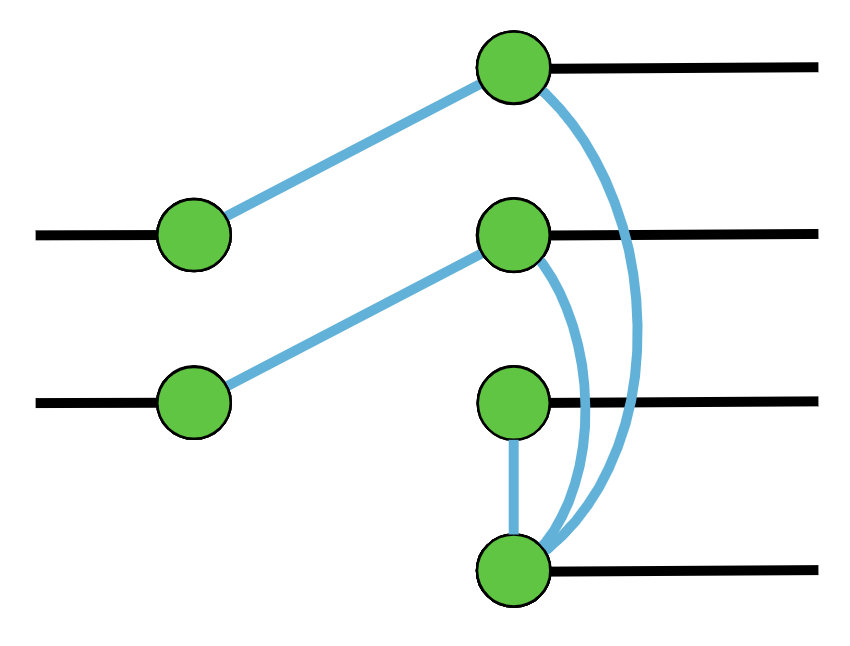}&\includegraphics[scale=0.08]{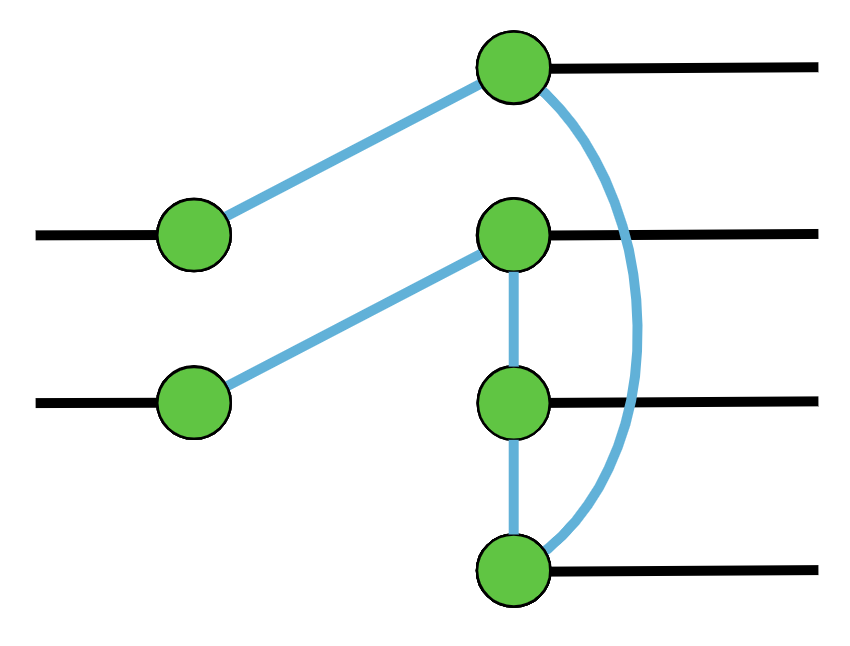}\\
\hline
Size: & 36&45&99&234
\end{tabular}
\egroup

\bigskip

{\small (c) $[4,2]$ codes equivalence classes showing the size of the class underneath a representative.}

\bigskip
\bgroup
\def\arraystretch{1.5}

\setlength{\tabcolsep}{0.5em} 
    \begin{tabular}{c|c|c|c|c|c}

        \includegraphics[scale=0.08]{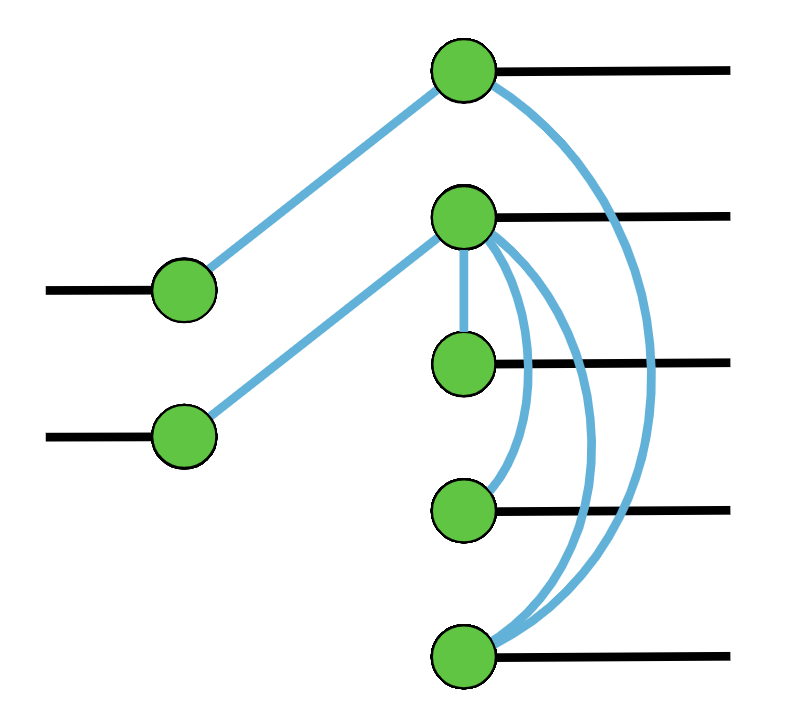} & \includegraphics[scale=0.08]{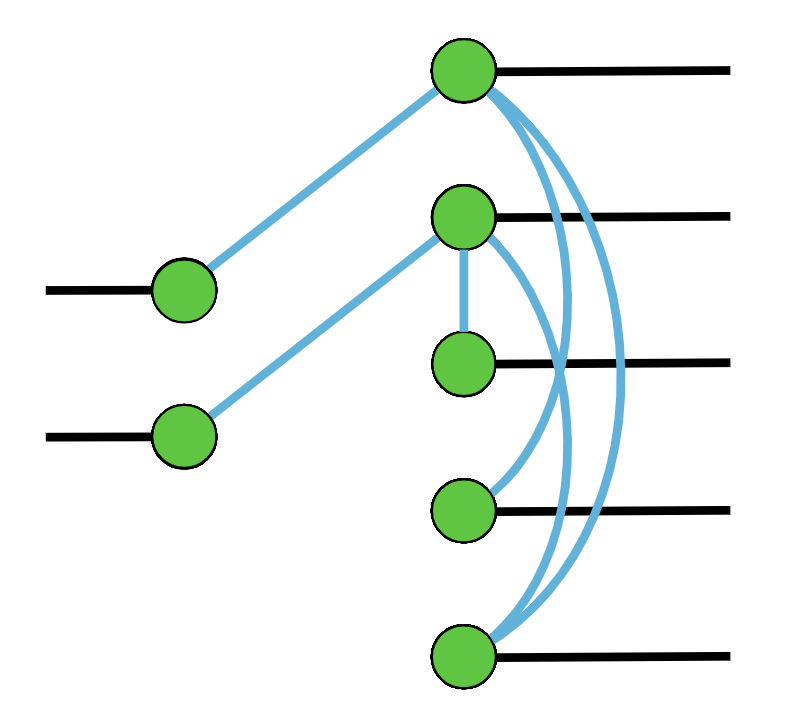} &                                     \includegraphics[scale=0.08]{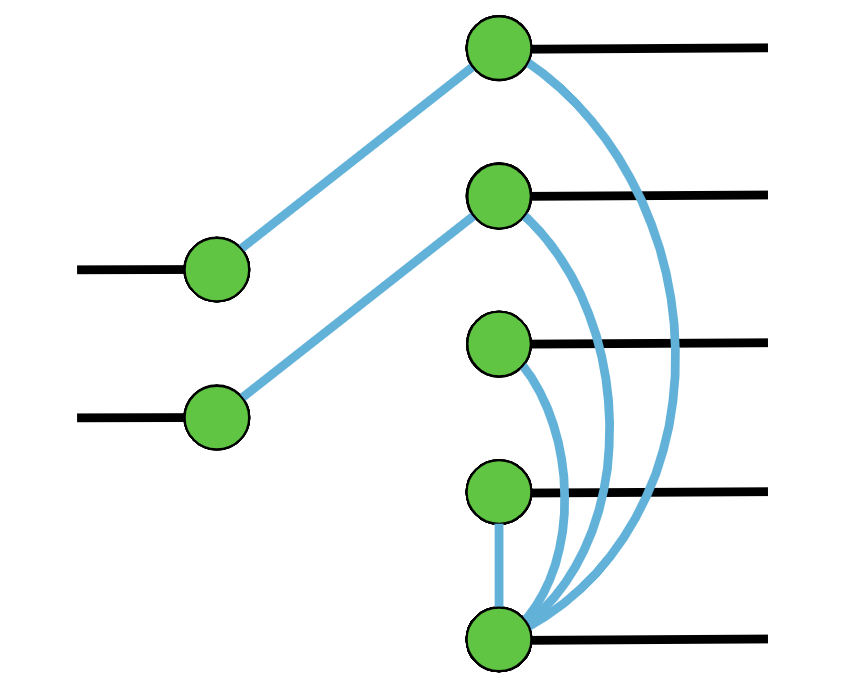} &  \includegraphics[scale=0.08]{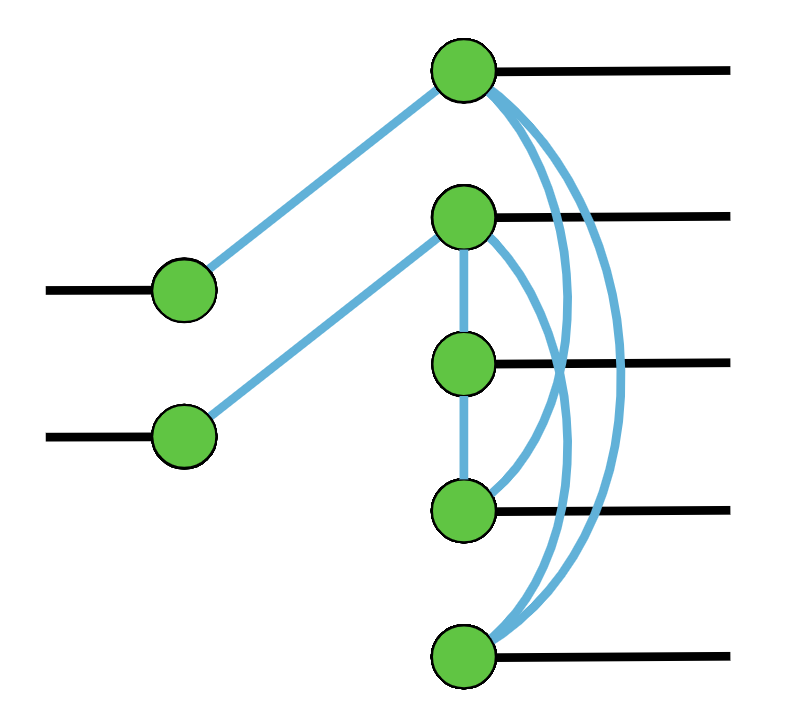} &                    \includegraphics[scale=0.08]{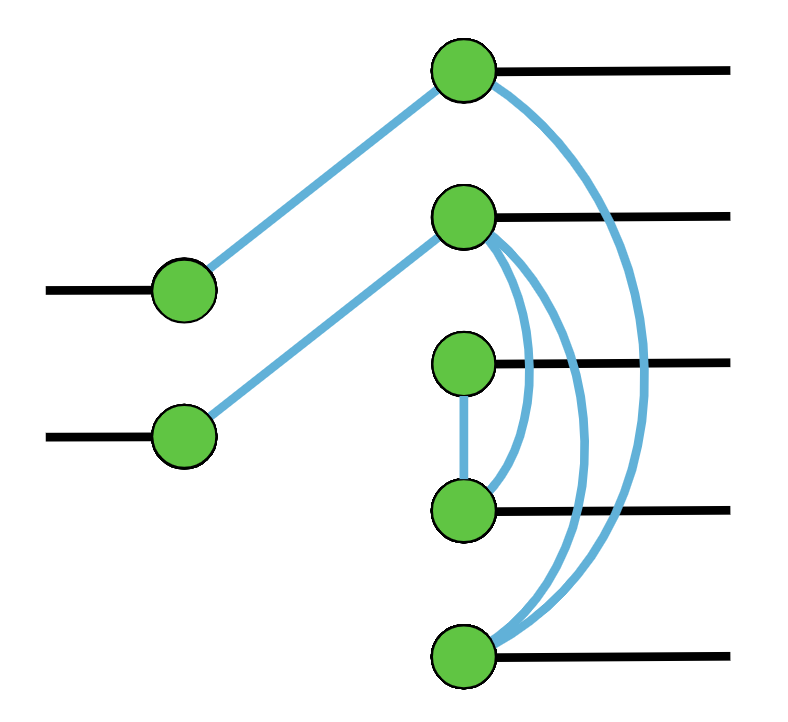} &                \includegraphics[scale=0.08]{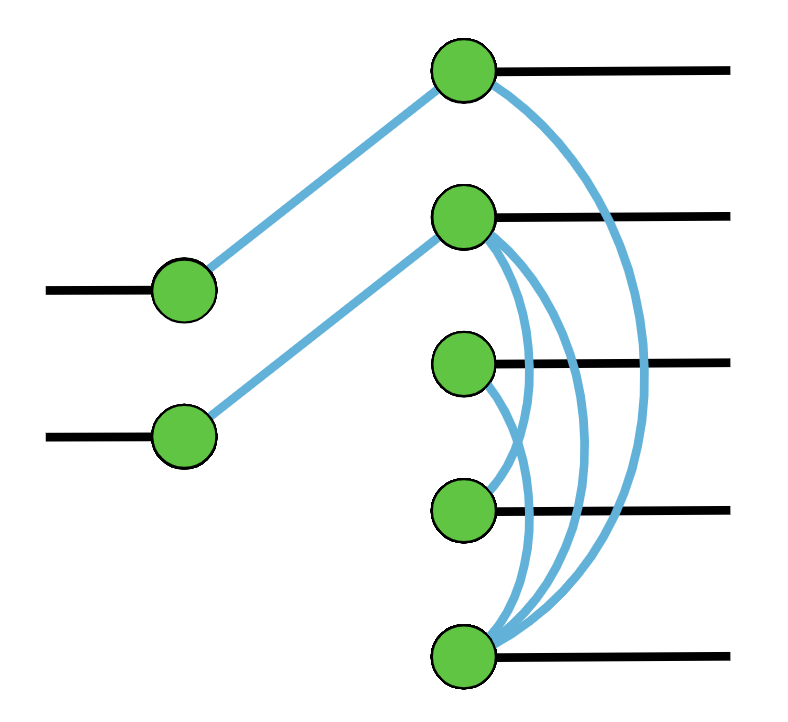} \\
         \hline
          63 & 108 & 144 & 414 & 459 & 486\\

                \hhline{=|=|=|=|=|=}

        \includegraphics[scale=0.08]{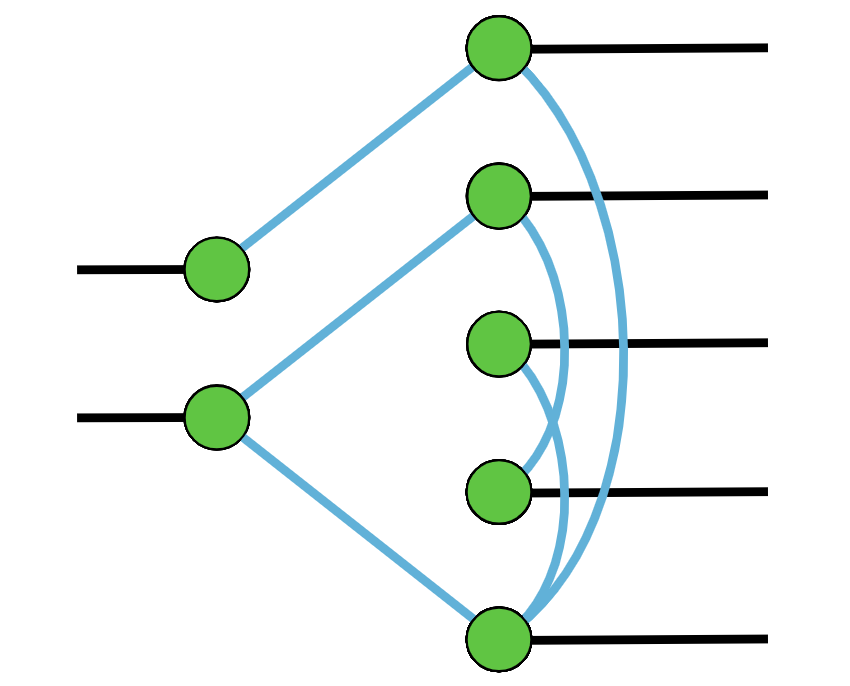}&                \includegraphics[scale=0.08]{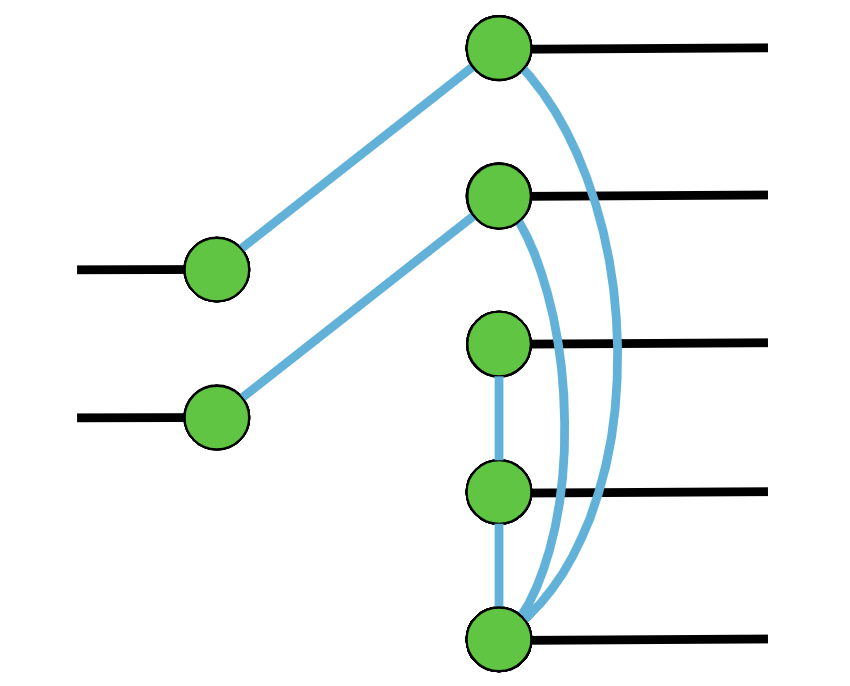}&                  \includegraphics[scale=0.08]{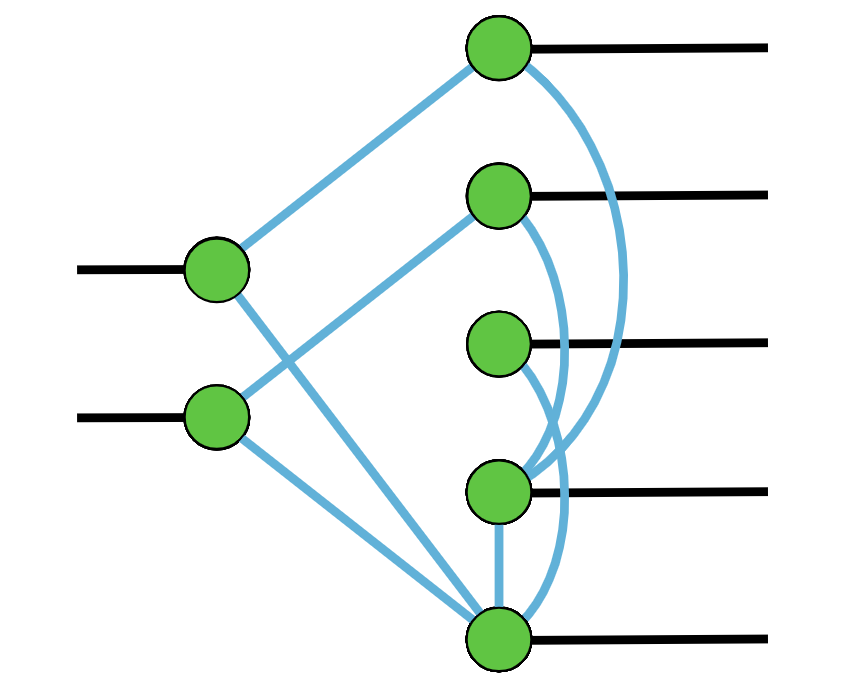} &                 \includegraphics[scale=0.08]{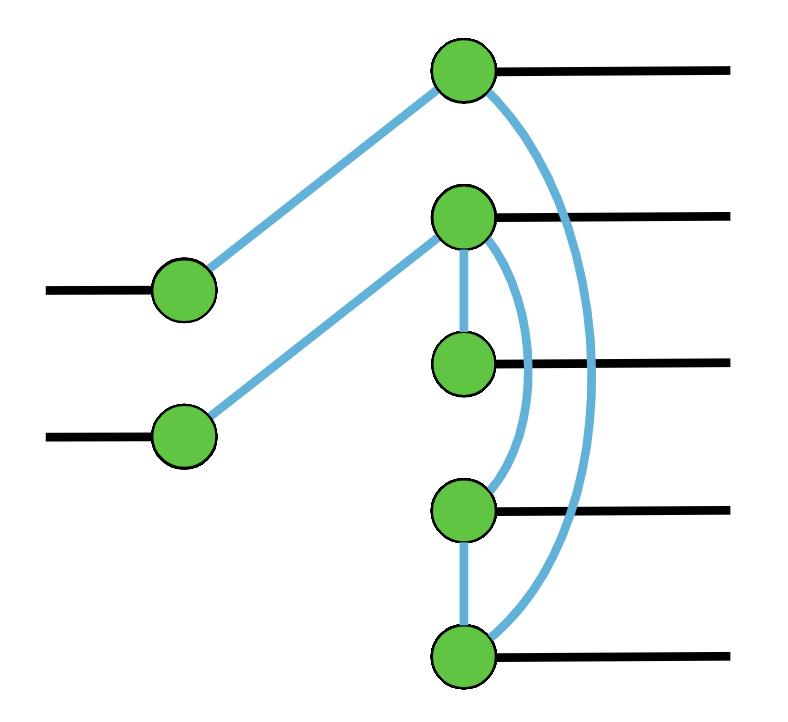} & \includegraphics[scale=0.08]{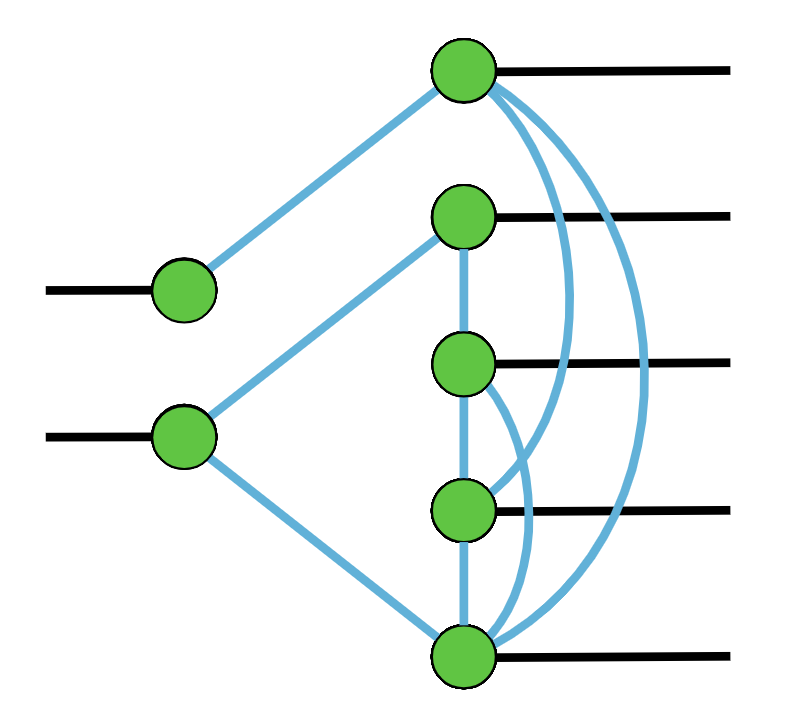} &                      \includegraphics[scale=0.08]{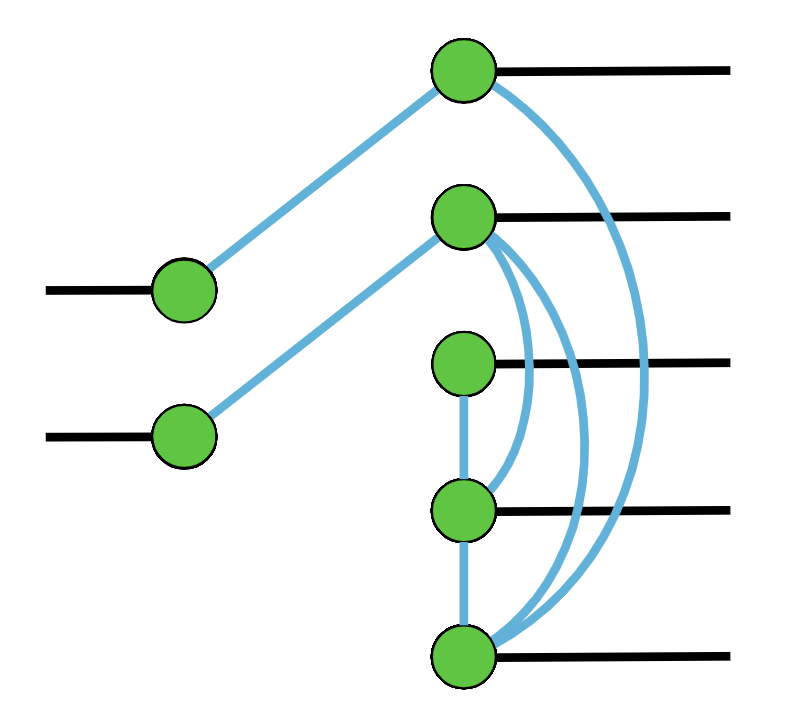}\\
        \hline
        540 & 972 & 1080 & 1080 & 1152 & 1188\\
             \hhline{=|=|=|=|=|=}
        \includegraphics[scale=0.08]{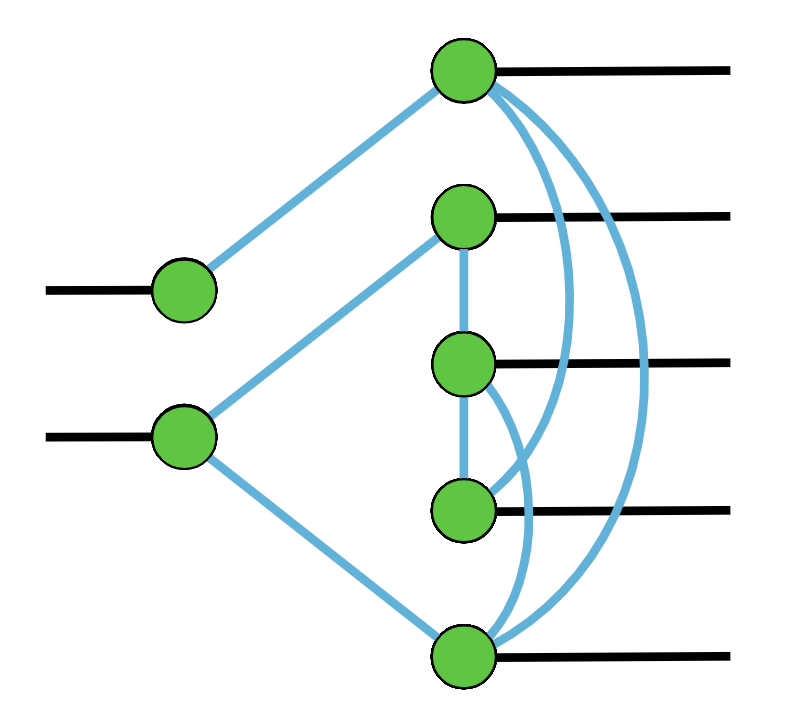} &                    \includegraphics[scale=0.08]{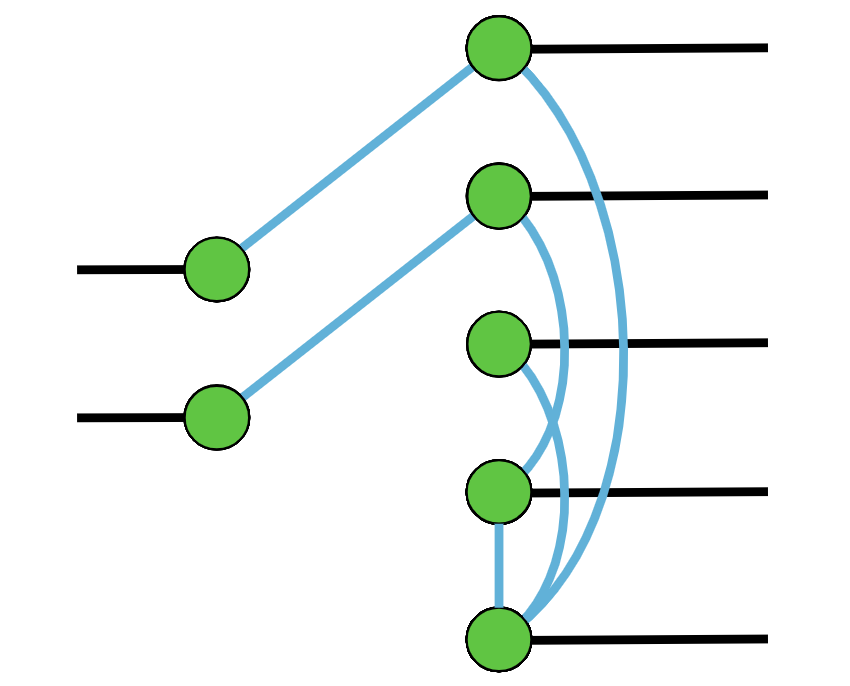}&                         \includegraphics[scale=0.08]{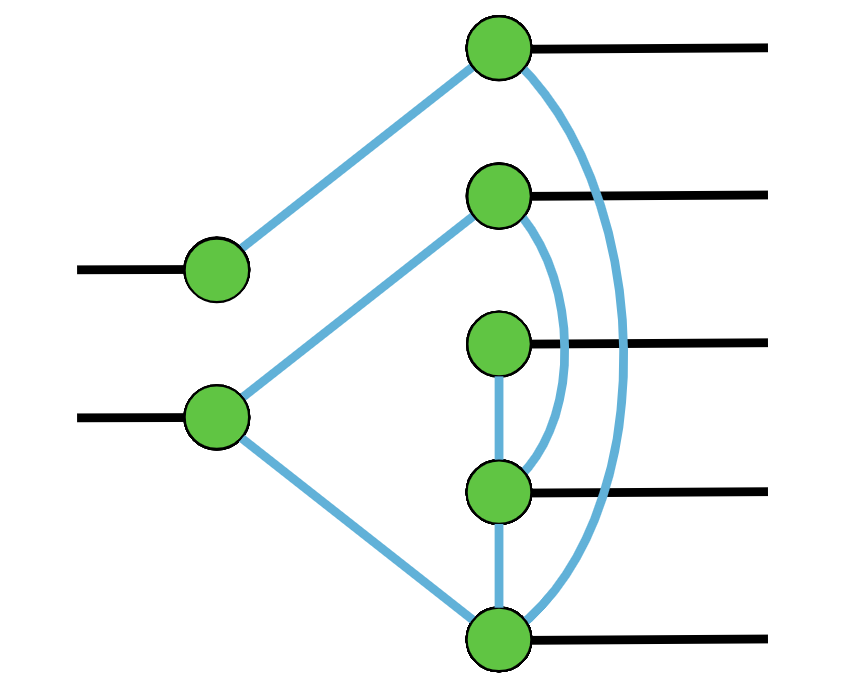} &                                           \includegraphics[scale=0.08]{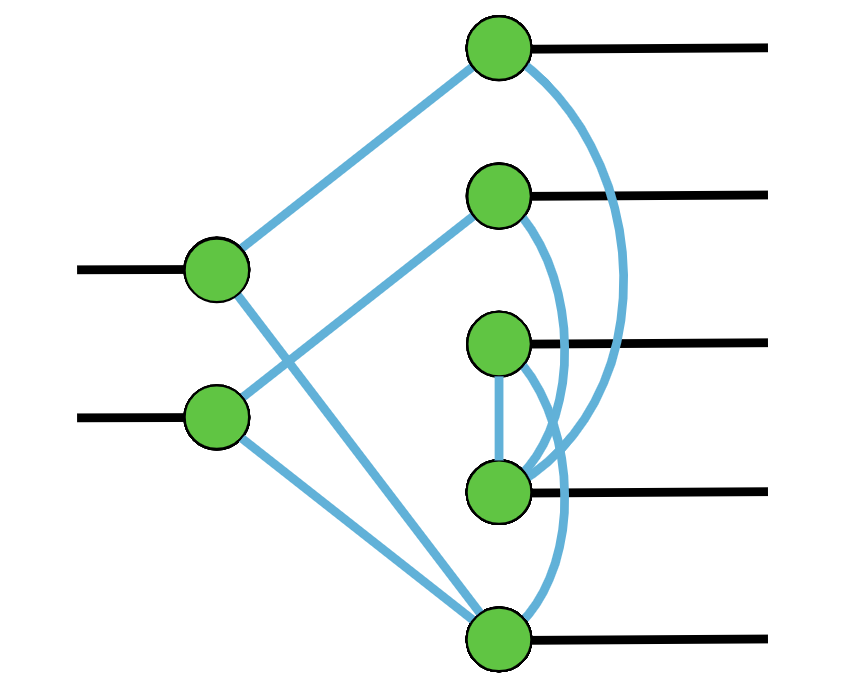} &                                 \includegraphics[scale=0.08]{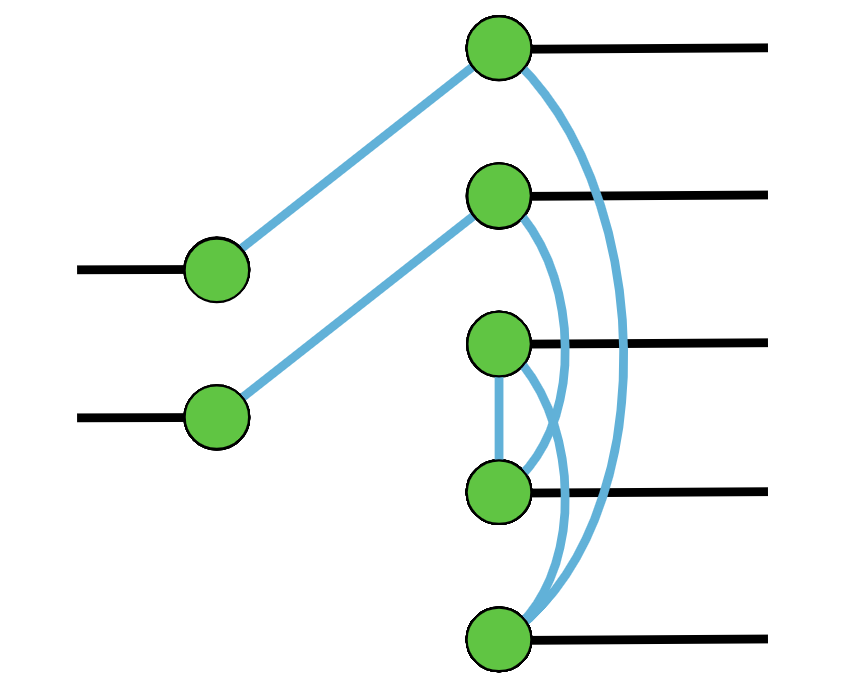} &                      \includegraphics[scale=0.08]{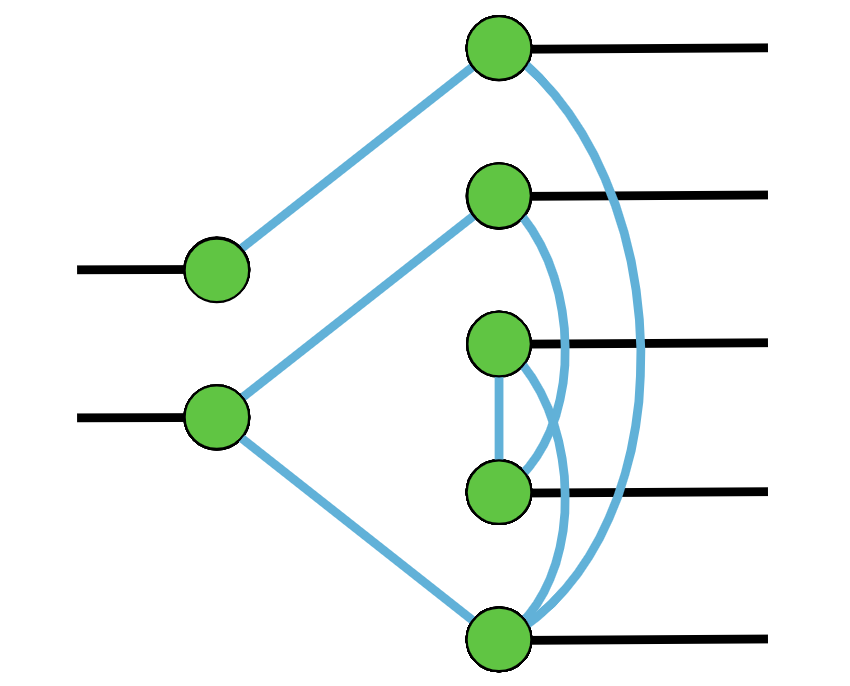} \\
        \hline
           1620 & 2268& 2484 & 4896 & 5184 & 5832\\

    \end{tabular}
    
\egroup

\bigskip

{\small (d) $[5,2]$ codes equivalence classes showing the size of the class underneath a representative.
}

\end{center}
    
    \caption{\justifying (a) shows the number of equivalence classes for $[n,2]$ encoder graphs. (b-d) show an element of the equivalence classes to denote the representative of the class and gives the size of the class. We only consider classes in which every graph is prime.}

    \label{coded codes}
\end{figure*}

\section{Tabulations from code}
\label{sec: tabulations from code}




Keeping in mind the new definition of equivalence and the simplifications made on the set of Clifford encoder graphs being considered, we now sort the encoders into their equivalence classes. To do this, we used the {disjoint set algorithm} to split encoders into equivalence classes based on whether an operation from Conjecture \ref{operations} caused one encoder to change into another.

Each encoder graph is converted into an integer based on the \textit{variable} edges present in the graph, which are the input-free edges, pivot-free edges, and free-free edges. Note that the input-pivot, input-input, and pivot-pivot edges are fixed, so these are not included among the variable edges.

The variable edges' values in the adjacency matrix are made into a single integer using a binary representation. Note that this adjacency matrix includes all vertices, so it is a $(n+k)\times (n+k)$ matrix.

For example, in the $[5,2]$ codes, the $7\times 7$ adjacency matrix would look like the following:
$$\begin{pmatrix}
    0 & 0 & 1 & 0 & a_{14} & a_{13} & a_{12} \\
    0 & 0 & 0 & 1 & a_{11} & a_{10} & a_{9} \\
    1 & 0 & 0 & 0 & a_{8} & a_7 & a_6 \\
    0 & 1 & 0 & 0  & a_5 & a_4 & a_3\\
    a_{14} & a_{11} & a_{8} & a_5 & 0 & a_2 & a_1\\
    a_{13} & a_{10} & a_7 & a_4 & a_2 & 0 & a_0\\
    a_{12} & a_{9} & a_6 & a_3 & a_1 & a_0 & 0\\
    
\end{pmatrix}$$

The top-left $4\times 4$ submatrix reflects the fixed input-pivot edges, as well as the lack of input-input edges and pivot-pivot edges. We place $a_0$ near the bottom-right corner and fill in the rows above from right to left.

After converting the ZX diagrams into integers, we use the {disjoint set algorithm}, which is useful for separating the whole set of possible encoder graphs into equivalence classes.

The code takes an integer representation, say $n$, of an encoder graph, performs one operation from Conjecture~\ref{operations} on the encoder graph, then merges the disjoint sets of $n$ and the integer representing the resulting encoder graph. All possible operations are applied, and the resulting values are merged with $n$'s disjoint set. 

When a local complementation is performed on a free output, it is possible that an input-pivot edge is removed. Furthermore, some input-input edges could be added. To fix this, we first employ operation 5 from Conjecture~\ref{operations} to set all input-input edges to 0. Then, operations 3 and 4 are used to turn the submatrix representing the input-to-output adjacency matrix into RREF. Operation 2 is used to put the pivots back into their fixed positions, so they once again correspond to their input vertices.

The results of the code are shown in Figure \ref{[n,1] codes big tables} and \ref{coded codes}. We have only shown the equivalence classes that have a single connected component, similar to the prime codes discussed in Section \ref{sec: prime codes}. The equivalence classes with multiple connected components can always be built up from connected components of smaller sizes, and finding these classes reduces to finding ways to partition the graph into groups of nodes within connected components. As shown in Figure \ref{[n,1] codes big tables}(a) and \Cref{coded codes}(a), the number of equivalence classes increases quickly as the number of output vertices increases.

For the $[n,1]$ codes, the number of equivalence classes for $n=1$ through 4 are $(n-1)!$. However, this pattern seems to break for larger values of $n$. Furthermore, for the $[n,2]$ codes, the number of equivalence classes for $n=2$ through 5 are $(n-2)\cdot (n-2)!$. It is possible that these patterns for small $n$ arise from the number of ways to permute the non-pivot output nodes.

Furthermore, in Figure \ref{[n,1] codes big tables}(b-e) and \Cref{coded codes}(b-d), the equivalence classes show a variety of sizes, with many of the sizes having a factor of 3 or 9. The values of the sizes generally have many divisors, suggesting nice combinatorial patterns.




We make the following conjecture regarding the recurring factors of 3 and 9.

\begin{conjecture}
    For positive integers $n > k$, the equivalence classes of $[n,k]$ codes with a single connected component have sizes divisible by $3^k$.
\end{conjecture}

Note that the factor of 9 is shared across the sizes of the equivalence classes of encoder graphs for the $[3,2]$, $[4,2]$, and $[5,2]$ codes. Similarly, there is a common factor of 3 across the sizes for the $[2,1], [3,1],[4,1],$ and $[5,1]$ codes. Also, based on the diagrams in Figure \ref{[n,1] codes big tables}(b) and Figure \ref{coded codes}(b), an analogous diagram can be drawn for $[4,3]$ codes, with the free output node having an edge to each pivot node. The resulting diagram is in a class of size 27 because each input-pivot pair has three different ways to share edges with the free node. Note that this is also the only class with prime graphs among $[4,3]$ codes. Thus, for $[n,n-1]$ codes, there is a clear way to see why the analogously constructed graphs have a class size of $3^{n-1}$. In general, we speculate that a factor of 3 arises from each input-pivot pair and the effects of different local Clifford gates (among $\{I,Z,S,SZ,H,HZ\}$) when applied on the pivot free edge. If this holds, then all equivalence classes for $[n,k]$ codes for $n > k$ would have sizes divisible by $3^k$.

\section{Equivalence classes with bipartite form(s)}
\label{sec:bipartite forms}

\begin{figure*}
    \includegraphics[scale=0.5]{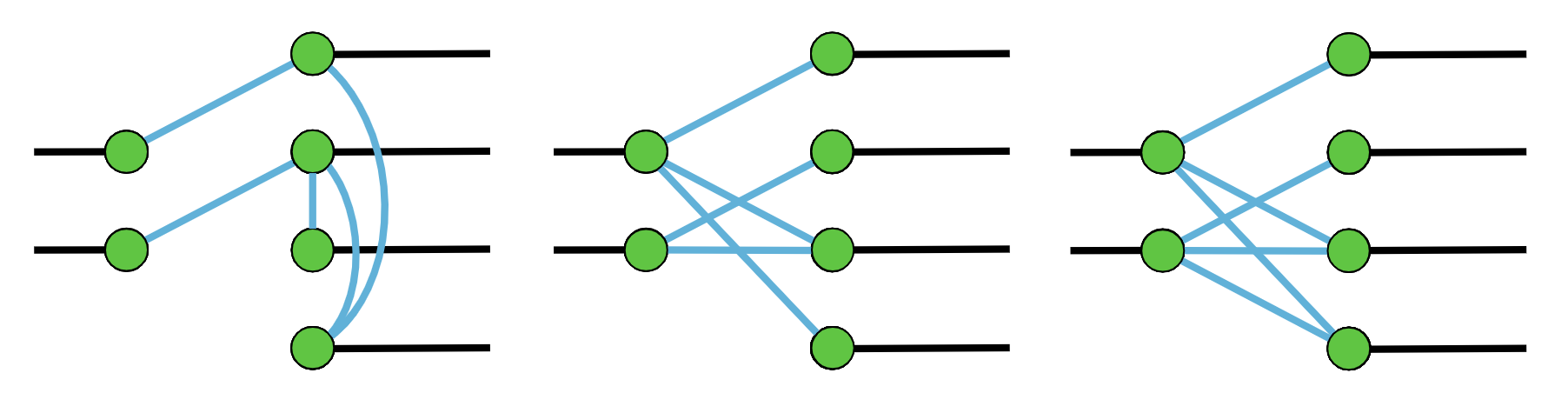}

    \caption{\justifying In the $[4,2]$ equivalence class with the representative shown at left, two possible bipartite forms are shown.}
    \label{diagram for 2-to-4}
\end{figure*}

When an equivalence class for an $[n,k]$ code has a bipartite form, we can narrow our search for a canonical form to the bipartite forms. These graphs are simpler, with no edges among output nodes, making them good candidates as canonical forms for these equivalence classes. We now turn to setting criteria for selecting a specific bipartite form in an equivalence class that could serve as a simple representative.

First, consider the case of $[4,2]$ codes. An example of the input-output adjacency matrix of a bipartite form, taking into account the RREF and pivot simplifications from Section \ref{sec 4}, could look like:
$$\begin{pmatrix}
    1 & 0 & \mathbf{1}&\mathbf{1}\\
    0 & 1& \mathbf{0} & \mathbf{1}\\
\end{pmatrix}$$

By changing the bolded entries between 0s and 1s, there are 16 possible bipartite forms among all $[4,2]$ codes. 

Similar to Section \ref{sec: prime codes} on Prime Codes, we will be considering codes that cannot be disconnected. In this section, codes that cannot be disconnected are called \textit{connected}.

\begin{proposition}
    \label{2-to-4 special}
    Among all connected $[4,2]$ bipartite ZX diagrams, there is only 1 distinct diagram up to equivalence through operations 2 through 4 from Conjecture~\ref{operations}.
\end{proposition}

\begin{proof}
    There are 16 total possible input-output adjacency matrices for $[4,2]$ bipartite codes, since there are $2^2$ ways for each input to connect to free outputs. If there are two or less input-to-free output edges, then either an output node is alone or the input nodes are in separate connected components.
    
    Then, there are only 5 possible input-output adjacency matrices in this case:
    \begin{align*}&\begin{pmatrix}
        1 & 0 & 1 & 1\\
        0 & 1 & 1 & 0\\
    \end{pmatrix}, \begin{pmatrix}
        1 & 0 & 1 & 1\\
        0 & 1 & 0 & 1\\
    \end{pmatrix},\\ &\begin{pmatrix}
        1 & 0 & 1 & 0\\
        0 & 1 & 1 & 1\\
    \end{pmatrix},
    \begin{pmatrix}
        1 & 0 & 0 & 1\\
        0 & 1 & 1 & 1\\
    \end{pmatrix},\\ &\begin{pmatrix}
        1 & 0 & 1 & 1\\
        0 & 1 & 1 & 1\\
    \end{pmatrix}.\end{align*}

    Consider the first matrix above. Switching the third and fourth output vertices (corresponding to the third and fourth columns of the matrix) results in $$\begin{pmatrix}
        1 & 0 & 1 & 1\\
        0 & 1 & 0 & 1\\
    \end{pmatrix},$$which is the second matrix. From here, we use operation 4 to replace the first row with the sum of the first and second rows modulo 2:
    $$\begin{pmatrix}
        1 & 1 & 1 & 0\\
        0 & 1 & 0 & 1\\
    \end{pmatrix}.$$Permuting the outputs achieves the third and fourth matrices. Lastly, starting from the above matrix, we use operation 4 to replace the second row with the sum of the current first and second rows modulo 2 to find
    $$\begin{pmatrix}
        1 & 1 & 1 & 0\\
        1 & 0 & 1 & 1\\
    \end{pmatrix}.$$This can be permuted to give the fifth matrix. Analogous sequences of operations can bring any of the other matrices to another, so all 5 of the graphs are equivalent, as desired.
\end{proof}

From Proposition \ref{2-to-4 special}, the representative form for the equivalence class that contains these 5 adjacency matrices can be chosen to be 
$$\begin{pmatrix}
    1 & 0 & 1 & 1 \\
    0 & 1 & 1 & 0\\
\end{pmatrix}.$$
Note that this graph has the property of having the least number of edges. We can distinguish between the four matrices with the least number edges by taking the graph with more edges on the first input and first output.

By writing one of these input-to-output adjacency matrices into the full $6\times 6$ matrix, we can produce its ZX diagram. In \Cref{diagram for 2-to-4}, the ZX diagram of the above matrix is shown in the center and the representative of the equivalence class (from Figure \ref{coded codes}) of this diagram is shown at left. The rightmost ZX diagram is another bipartite form in the same equivalence class.

\bigskip

\begin{figure*}
    \includegraphics[scale=0.5]{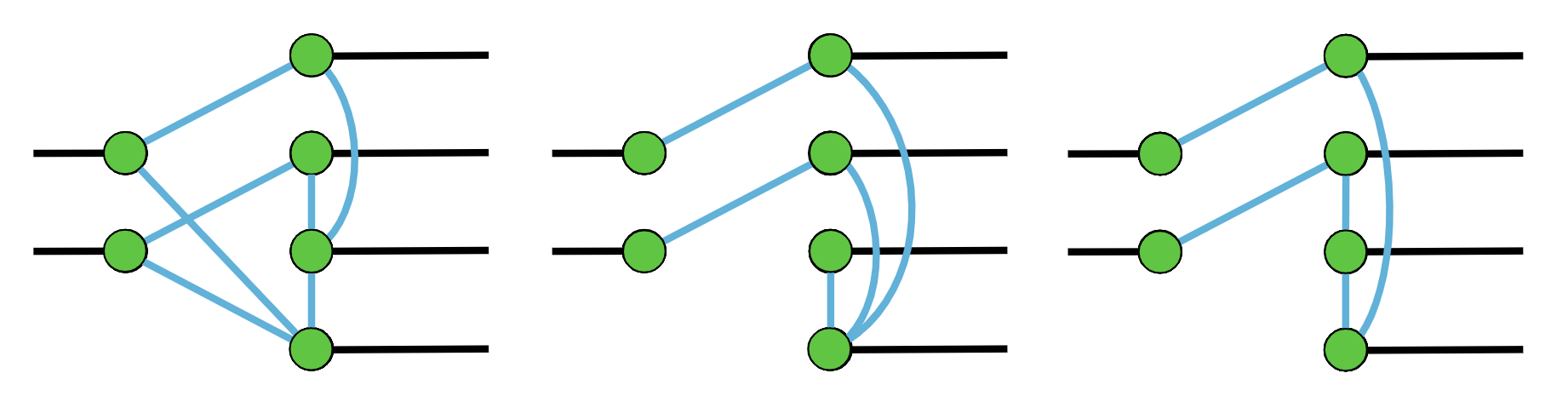}

    \caption{\justifying The representatives of the three different equivalence classes lacking bipartite forms among the classes for $[4,2]$ ZX diagrams.}

    \label{small non-disjoint}
\end{figure*}

Now, we present a general method of simplifying a bipartite $2\times n$ input-to-output adjacency matrix.

\begin{proposition}
    \label{minimum bound 2-to-n}
    Consider a $[n,2]$ encoder graph equivalent to some bipartite form. It is also equivalent to a bipartite form where the two inputs are both connected to at most $\big\lfloor\frac{n}{2}  - 1\big\rfloor$ of the same free outputs.
\end{proposition}

\begin{proof}In this proof, we only consider encoders that are equivalent to a bipartite form.

For the sake of contradiction, suppose all bipartite forms of this equivalence class have at least $\big\lfloor\frac{n}{2} \big\rfloor$ shared free outputs. In an input-to-output adjacency matrix, this would look like
$$\begin{pmatrix}
    1 & 0 & 1 & 1 & 1 & 1 & 0 \\
    0 & 1 & 1 & 1 & 1 & 0 & 1\\
\end{pmatrix}.$$The first two columns are fixed to be input-pivot edges, as usual. If there are at least $\big\lfloor\frac{n}{2} \big\rfloor$ shared free outputs, the other $n-2$ columns must have a majority of columns containing two 1's. In this example, 3 out of 5 columns contain two 1's.

However, using operation 4 from Conjecture~\ref{operations}, the top row can be replaced with the sum of the top and bottom row modulo 2. 

Note that this means all the free outputs that were shared by both inputs have their edges with the first input disconnected, so at least $\big\lfloor\frac{n}{2} \big\rfloor$ columns do not have two 1's. 

Furthermore, after the operation, the first column cannot possibly have two 1's, so one additional column does not have two 1's. The example matrix above turns into
$$\begin{pmatrix}
    1 & 1 & 0 & 0 & 0 & 1 & 1 \\
    0 & 1 & 1 & 1 & 1 & 0 & 1\\
\end{pmatrix}.$$We can rearrange output vertices to bring back the pivots. The following is thus equivalent
$$\begin{pmatrix}
1 & 0 & 1 & 1 & 1 & 0 & 0 \\
0 & 1 & 1 & 1 & 0 & 1 & 1 \\
\end{pmatrix}.$$

Thus, the maximum number of columns with two 1's is now $n - 1 - \big\lfloor\frac{n}{2}\big\rfloor$. However,
$$n - 1 - \Big\lfloor\frac{n}{2} \Big\rfloor < \Big\lfloor \frac{n}{2}\Big\rfloor,$$so we reach a contradiction, since there are now less than $\big\lfloor \frac{n}{2}\big\rfloor$ shared free outputs in an equivalent bipartite form. 

Therefore, the claim holds.
\end{proof}

Proposition \ref{minimum bound 2-to-n} demonstrates that we can choose a bipartite form that has a relatively small number of shared free outputs. In fact, if the top row is the horizontal vector $\textbf{a}$ and the bottom row is the horizontal vector $\textbf{b}$, by linear operations, there are only 3 possibilities of unordered combinations of two vectors in the rows. It could be $(\textbf{a}, \textbf{b}), (\textbf{a+b}, \textbf{b}), (\textbf{a+b}, \textbf{a})$. Then, we can choose which of these bipartite forms has the least number of shared free outputs and thus minimizes this number.

\bigskip

In an attempt to show the uniqueness of the bipartite forms in an equivalence class, we conjecture the following, which would allow an efficient way to check whether two bipartite forms are equivalent:
\begin{conjecture}
\label{conjecture: bipartite}
    In an equivalence class with bipartite forms, all bipartite forms in the class can be transformed from one to another using only output permutations, input permutations, and row operations on the input-output adjacency matrix.
\end{conjecture}

One straightforward approach starts by assuming for the sake of contradiction that two bipartite graphs, $G_1$ and $G_2$, are equivalent even though they cannot be transformed from one to another using only the three operations in Conjecture \ref{conjecture: bipartite}. A contradiction could be reached if we are able to show that the entanglement between two of the nodes is different in the diagrams. Quantifying the entanglement could be possible using the partial trace.

\begin{figure}[h]
\centering
\includegraphics[scale=0.40]{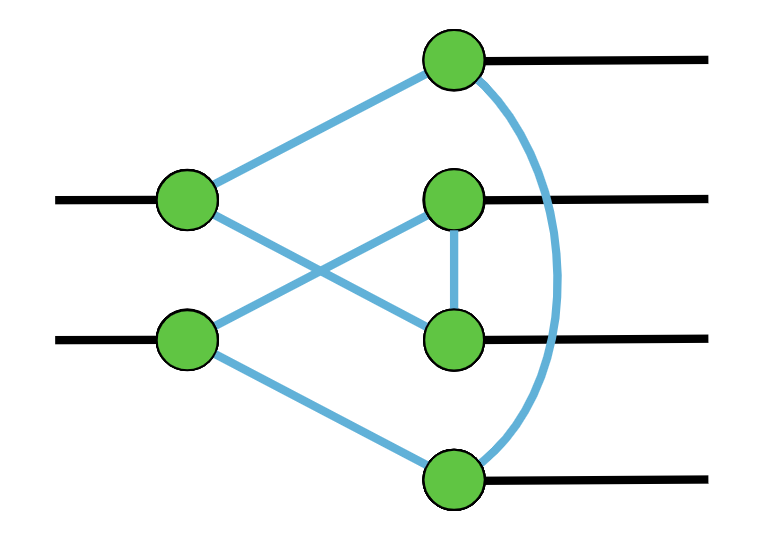}
\caption{\justifying A ZX diagram with the minimum number of output-output edges in its equivalence class. The equivalence class has no bipartite forms and does not contain any ZX diagrams with 0 output-output edges or 1 output-output edge.}
\label{new 300 graph}
\end{figure}

Besides equivalence classes with bipartite forms, there are also some classes that do not have bipartite forms. Among the equivalence classes of $[4,2]$ codes from Figure \ref{coded codes}, five of the classes have zero bipartite forms. Three of these classes have connected graphs, and they are shown in Figure \ref{small non-disjoint}. Finding a clean, representative form for these classes is less intuitive, but analyzing a few other graphs in these equivalence classes could give a clue as to what to choose. For example, Figure \ref{new 300 graph} shown above could be a better representative for the equivalence class containing the leftmost diagram in Figure \ref{small non-disjoint}. Figure \ref{new 300 graph} is symmetric, and it has the least number of output-output edges.

\section{Conclusion}

This paper presented our work on producing the KL canonical forms for CSS codes, extending the work done by \cite{KLS}. Furthermore, we show the resulting KL canonical forms of the toric code and certain surface codes. Furthermore, we introduced the notion of prime codes and proved our Fundamental Theorem of Clifford Codes. We also tabulated results found when considering codes with much looser equivalence conditions, and we analyzed possible representative forms (such as bipartite forms) for the equivalence classes found.

The work done on CSS codes conclusively finds an elegant, minimal form in the ZX calculus for CSS codes. From Kissinger \cite{kissinger2022phase}, it was known that phase-free ZX diagrams are CSS codes, but now we have shown an optimal representation, reducing the number of nodes to have one per input and output. Future works on CSS codes could use the models produced by this work and related works on finding more powerful CSS codes, with each model offering its own strengths and limitations.

Our work on prime codes provides a new framework when working with quantum error-correcting codes, formulating and proving the Fundamental Theorem of Clifford Codes. Future works could examine the structures of primes, generate larger primes, and find useful things to tabulate about prime codes in a similar manner to prime numbers and prime knots.

Furthermore, the tabulations of equivalence classes and the work on the bipartite forms could aid future works in determining patterns among equivalence class sizes and representatives. Extending the definition of equivalence to allow for permutation of outputs is physically significant as these permutations do not affect the manner that the inputs are transformed into outputs. A general formulation of codes representative of huge families of equivalent codes could lay the groundwork to determining stronger codes in the future.

\section{Acknowledgments}\label{sec:acknowledgements}

We would like to thank the MIT PRIMES-USA program for the opportunity to conduct this research. Thank you to Jonathan Lu for support in creating the diagrams. ABK was supported by the National Science Foundation (NSF) under Grant No. CCF-1729369.

\bibliography{bib}

\newpage

\appendix
\section{Constructing the 3-by-3 toric code}
\label{appendix sec: construct 3-by-3 toric}

\begin{figure}
\centering
\includegraphics[scale=0.35]{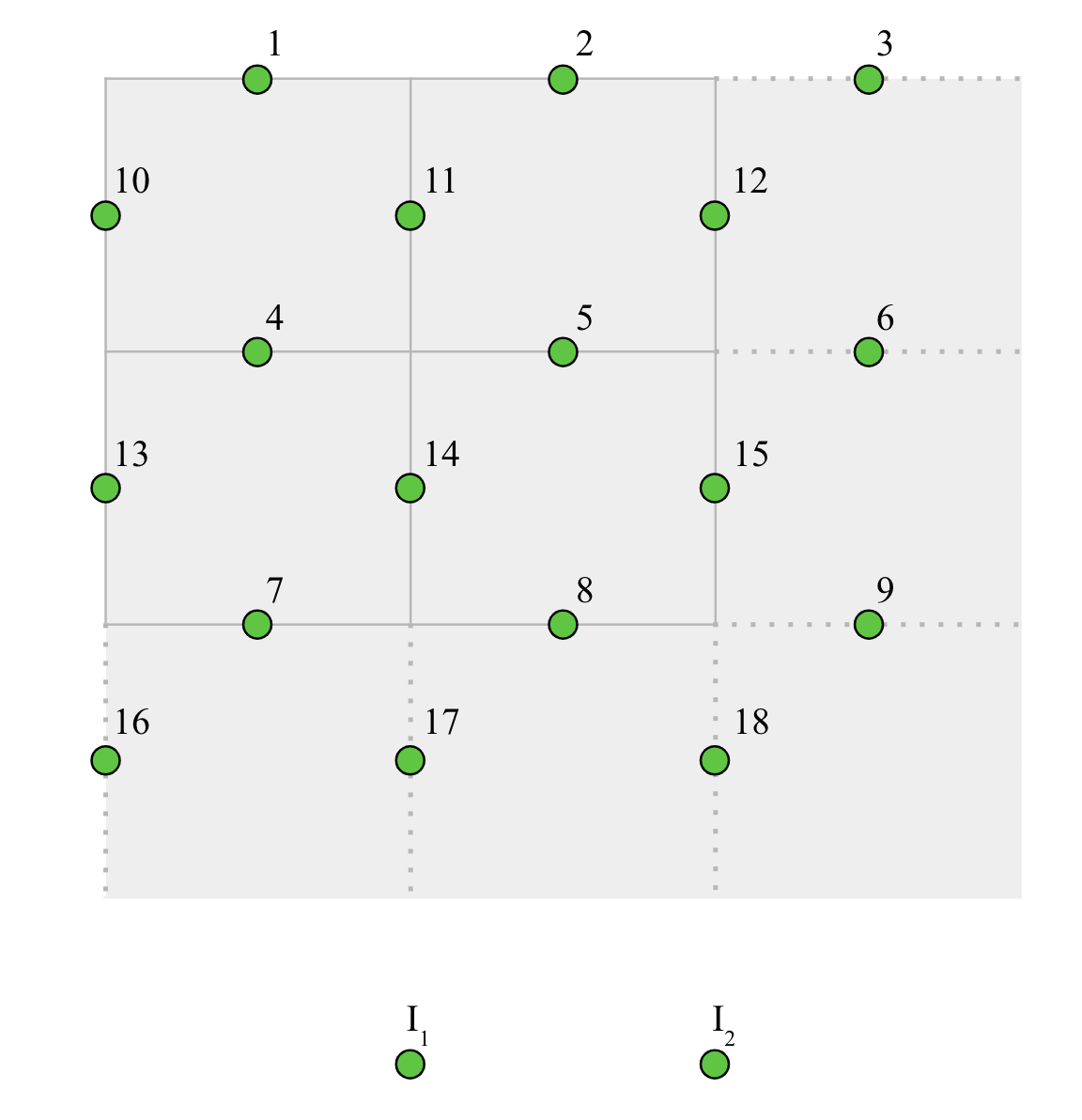}
\caption{\justifying The input and output nodes in the 3-by-3 toric code. The boundary conditions are periodic, and all nodes are initially set to $Z$ nodes.}
\label{3-by-3 fig}
\end{figure}

In Section \ref{section:surfacetoric}, we provided the 2-by-2 and 3-by-3 toric codes in ZX calculus. Here, we provide a more detailed description of the methodology used to determine the structure of the 3-by-3 toric code.

Among the 18 outputs in Figure \ref{3-by-3 fig}, we expect some of the free output edges to contain Hadamard gates. (For now, we are keeping the nodes as $Z$ nodes.) Suppose the output edge onto vertex 1 has a Hadamard. This implies that applying the stabilizer $Z_1Z_4Z_{10}Z_{11}$ would result in sliding a $Z$ gate from the end of the output edge, through the Hadamard (which converts the $Z$ gate to an $X$ gate), then through vertex 1 itself. By the $\pi$-copy rule (see Definition \ref{def: basic rewrite rules}), the $X$ gate, which is an $X$ node with phase $\pi$, copies itself onto the edges (excluding the output edge) connected to vertex 1.

Similarly, while continuing the assumption that output node 1 has a Hadamard, if we instead applied the stabilizer $X_1X_3X_{10}X_{16}$, we slide an $X$ gate from the end of the output edge, through the Hadamard (which converts the $X$ gate to a $Z$ gate), then onto vertex 1. By the merging rule (see Definition \ref{def: basic rewrite rules}), the $Z$ gate, which has phase $\pi$, merges with vertex 1, a phase 0 $Z$ node. This results in vertex 1 gaining a phase of $\pi$.

\bigskip

By the preceding paragraphs, the behavior of the $Z$ and $X$ gates on an output node with a Hadamard is understood. The analogous behavior occurs on an output node without a Hadamard by switching all the colors used in the processes above.

To determine all of the edges in the 3-by-3 toric code in Figure \ref{3-by-3 fig}, we consider the process of applying the stabilizers onto the output nodes. Note that all of the internal edges among nodes in the diagram must be edges with Hadamards so that the merging rule cannot be applied to merge multiple nodes into one. To simplify our work, we set the output nodes with Hadamards to be $1, 2, 3, 7, 8, 9,$ and $16, 17, 18$. Then, by stabilizer $Z_1Z_4Z_{10}Z_{11}$, the nodes 4, 10, and 11 gain phase $\pi$ from their $Z$ gates while node 1 will cause a $\pi$-copy rule to move $X$ gates onto the internal edges connected node 1. Since all internal edges have Hadamards, moving the $X$ gates through the Hadamards will result in $Z$ gates. If these $Z$ gates went to any nodes other than nodes 4, 10, and 11, the stabilizer would not have kept the configuration the same. Therefore, the $Z$ gates must arrive at only nodes 4, 10, and 11. This works because the $\pi$'s from these $Z$ gates cancel with the $\pi$s already at the nodes. Thus, the only internal edges to node 1 are from nodes 4, 10, and 11.

Using similar reasoning, we can deduce the rest of the internal edges among the output nodes. Furthermore, to determine the logical operators (to connect the input nodes to), we look for sets of nodes that, when any stabilizer is applied, keep the input node at phase 0. The resulting figure is given in the text (see Figure \ref{final 3 by 3}).

\begin{figure*}
    \centering
    
    \begin{minipage}{0.4\textwidth}
        \begin{subfigure}[b]{\textwidth}
            \includegraphics[width=0.9\linewidth]{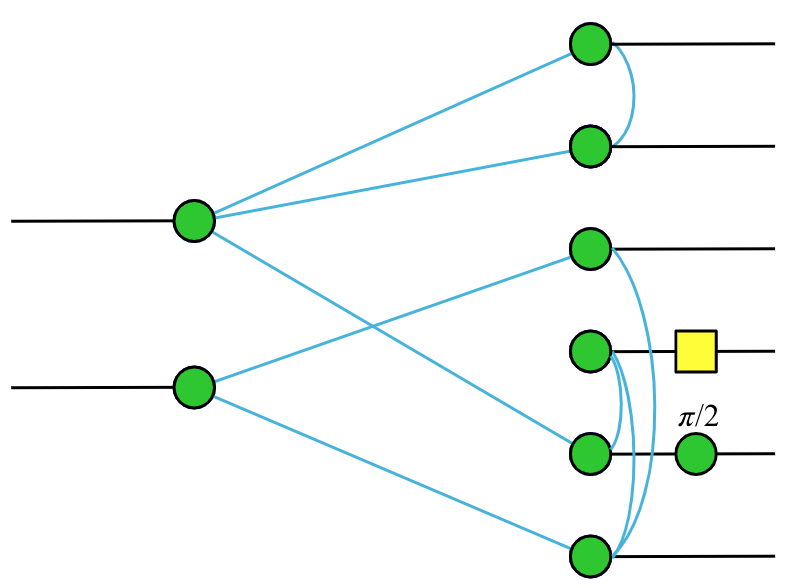}
            \caption{\justifying The KLS form of a QECC. Blue edges represent edges with Hadamards. Note that two of the output edges have local operations. One of the local operations is a Hadamard gate while the other is a green $\pi/2$ or $S$ gate.}
            \label{kls form of QECC}
        \end{subfigure}

        \begin{subfigure}[b]{\textwidth}
            \includegraphics[width=0.9\linewidth]{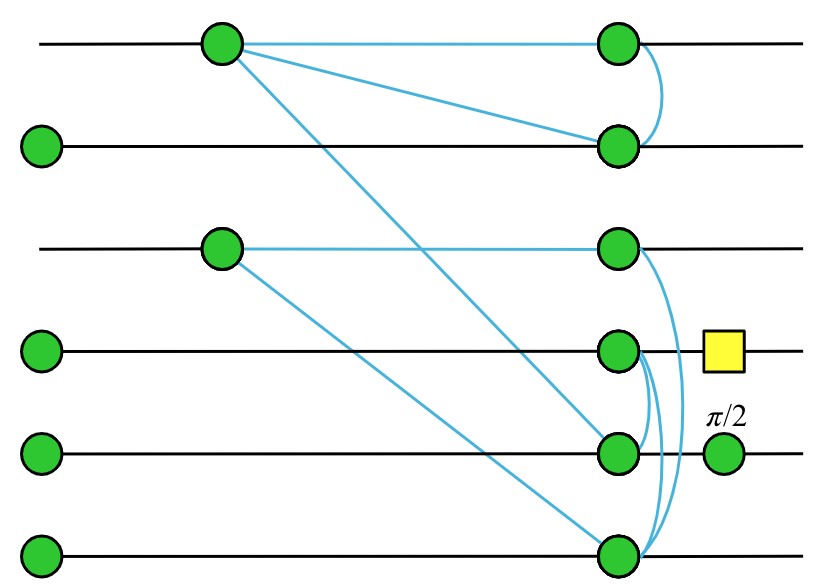}
            \caption{\justifying The non-pivot output nodes are un-merged into two $Z$ nodes each, and one of each pair is placed to the left.}
            \label{step1}
        \end{subfigure}

        \begin{subfigure}[b]{\textwidth}
            \includegraphics[width=0.9\linewidth]{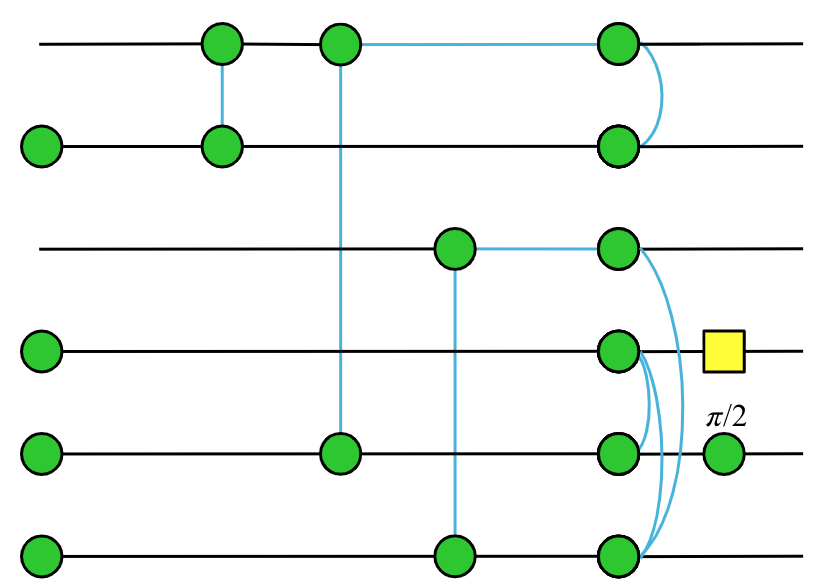}
            \caption{\justifying By un-merging the input nodes, the edges between the inputs and non-pivot outputs can be shown separately from each other.}
            \label{step2}
        \end{subfigure}
    \end{minipage}
    \hspace{0.1\textwidth} 
    \begin{minipage}{0.4\textwidth}
        \begin{subfigure}[b]{\textwidth}
            \includegraphics[width=\linewidth]{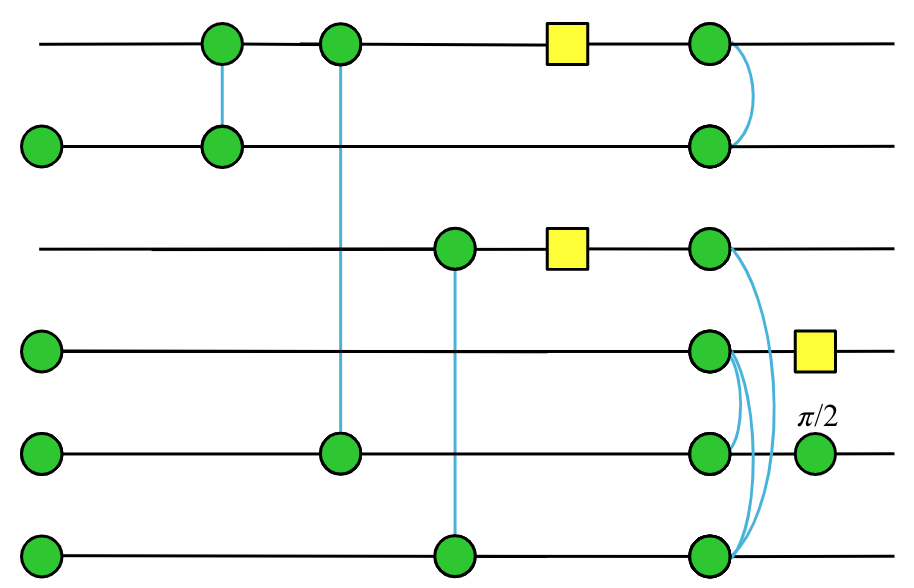}
            \caption{\justifying The input-pivot edges are exchanged for edges with a yellow Hadamard gate on them.}
            \label{step3}
        \end{subfigure}

        \vspace{0.6 cm}
        \begin{subfigure}[b]{\textwidth}
            \includegraphics[width=\linewidth]{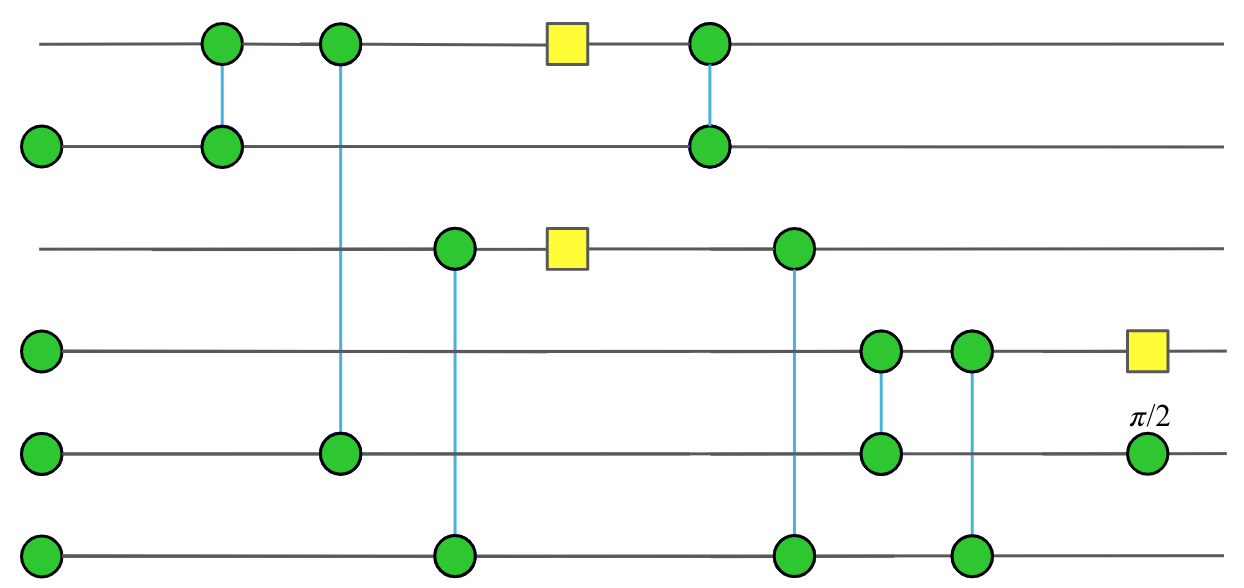}
            \caption{\justifying Similar to \ref{step2}, each of the output nodes with more than one connection are un-merged to separate the output-output edges from each other. Note that the local operations are still at the very right-hand side of the diagram.}
            \label{step4}
        \end{subfigure}

        \vspace{0.6 cm}
        \begin{subfigure}[b]{\textwidth}
            \includegraphics[width=\linewidth]{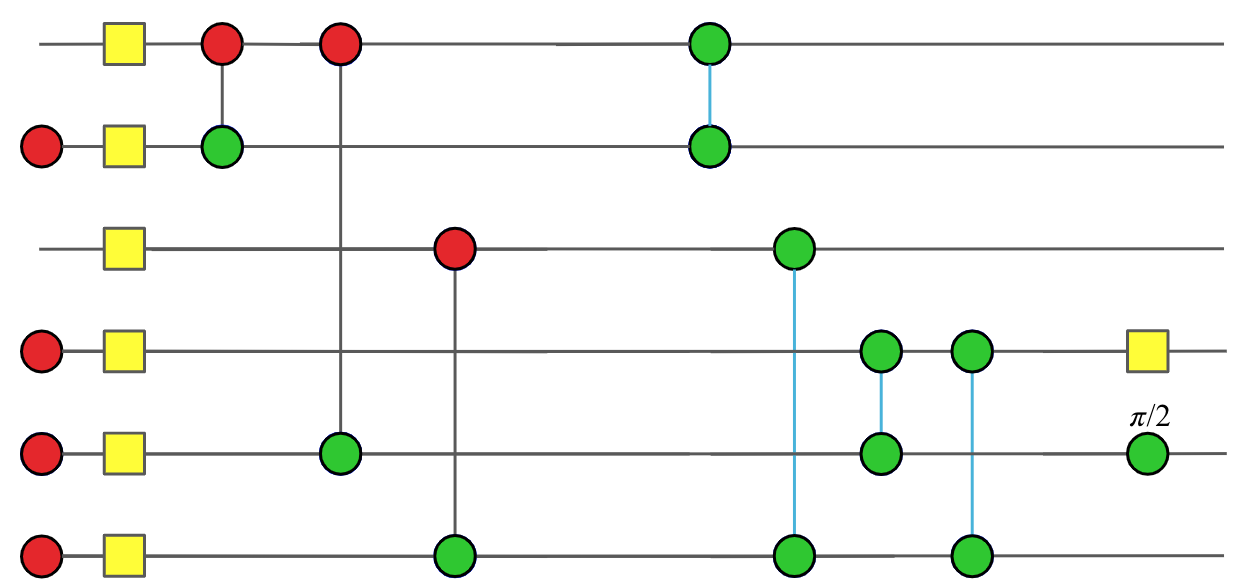}
            \caption{\justifying The Hadamards in the middle of \ref{step4} are pushed to the left, and the green $\ket{+}$ states are exchanged for the equivalent representation of $H\ket{0}$, which is a Hadamard on a red $\ket{0}$ state. }
            \label{step5}
        \end{subfigure}
    \end{minipage}
    
    \caption{\justifying Conversion of a KLS form ZX diagram into the equivalent quantum circuit diagram.}
    \label{overall}
\end{figure*}

\section{Converting a KLS canonical form into a quantum circuit}
\label{appendix sec: convert KLS to circuit}

In Section \ref{section:background}, we described a procedure for converting a KLS canonical form into a quantum circuit, which is repeated here.

\begin{enumerate}
\item Start with $n-k$ open wires representing the inputs of the circuit.
\item Add a $\ket{0}$ state for each of the $k$ non-pivot output nodes.
\item Apply an $H$ gate to all $n$ wires.
\item Apply a $CX$ gate between the wires corresponding to the edges between inputs and non-pivot outputs. The input node is the target qubit, and the output node is the controlled qubit.
\item Apply a $CZ$ gate between the wires corresponding to the edges between only outputs.
\item Apply the local operations attached to the outputs.
\end{enumerate}

We will now show why this works.

Consider the example given in \Cref{kls form of QECC}. We will convert this KLS form into a circuit. We can first move the input nodes to be along the same horizontal wire as their pivots nodes. Then, we split the non-pivot nodes by un-merging two zero-phase $Z$ nodes. The resulting diagram is in \Cref{step1}. From here, the edges from inputs to non-pivot outputs can be separated by un-merging nodes and expressing each edge separately, as in \Cref{step2}. In \Cref{step3}, the Hadamards between the inputs and pivots are shown explicitly. In \Cref{step4}, we do a similar un-merging of nodes to separately express the edges between nodes.

The steps used in these diagrams hold in general. We can un-merge each node until all the edges are expressed separately (and the non-pivot nodes have an initial state), and, to keep things organized, we can keep the input-output edges on the left side and the output-output edges on the right side.

From \Cref{step4}, note that the edges with Hadamards between the nodes of a ZX diagram are equivalent to the $CZ$ gates between the corresponding wires in a quantum circuit. Also, the $Z$ nodes at the start are equivalent to $\ket{+}$.

Now, consider sliding the two Hadamards in the middle towards the left of the diagram. Because $ZH = HX$, this means each of the $CZ$'s that the $H$'s pass through turns into a $CX$ with the target qubit on the input's wire. Also, we may exchange the $Z$ nodes at the start for a $X$ node and an $H$, since $H\ket{0} = \ket{+}$. This gives \Cref{step5}. From here, we can see why the procedure for creating the circuit from KLS form works, since each of the ZX calculus components in \Cref{step5} can quickly be converted to a circuit diagram component.

\end{document}